\documentclass[11pt,a4paper]{article}

\usepackage[utf8]{inputenc}
\usepackage[T1]{fontenc}
\usepackage{lmodern}

\usepackage{amsmath,amssymb,amsthm}
\usepackage{bm}
\usepackage{mathtools}

\usepackage{booktabs}
\usepackage{array}
\usepackage{tabularx}
\usepackage{float}

\usepackage[margin=2.5cm]{geometry}
\usepackage{setspace}
\usepackage{parskip}

\usepackage{graphicx}
\graphicspath{{figures/}}
\usepackage[section]{placeins}

\usepackage{enumitem}

\usepackage[authoryear,round]{natbib}
\usepackage[colorlinks=true,linkcolor=blue,citecolor=blue,urlcolor=blue]{hyperref}
\usepackage[capitalise,noabbrev]{cleveref}

\newtheorem{theorem}{Theorem}
\newtheorem{proposition}{Proposition}
\newtheorem{corollary}{Corollary}
\newtheorem{lemma}{Lemma}
\theoremstyle{definition}

\theoremstyle{remark}

\newcommand{\one}{\mathbf{1}}

\newcommand{\Sig}{\boldsymbol{\Sigma}}
\newcommand{\Sigi}{\boldsymbol{\Sigma}_1}
\newcommand{\SigH}{\boldsymbol{\Sigma}_H}
\newcommand{\eivec}{\mathbf{v}}
\newcommand{\kappastat}{\kappa}
\newcommand{\Omat}{\mathbf{O}}
\newcommand{\horizon}{H}

\newcommand{\R}{\mathbb{R}}

\newcommand{\Var}{\mathrm{Var}}
\newcommand{\Cov}{\mathrm{Cov}}
\newcommand{\tr}{\mathrm{tr}}

\title{A Spectral Generalisation of the Variance Ratio:\\
Eigenstructure of Long-Horizon Portfolio Covariance and\\
a Multi-Memory Factor Model of U.S.\ Equity Returns}
\author{Anders G.\ Fr{\o}seth\thanks{Independent Researcher.\
  E-mail: \href{mailto:indrefjorden@pm.me}{indrefjorden@pm.me}.}}
\date{\today}

\begin{document}

\maketitle

\begin{abstract}
\noindent
We propose a multivariate generalisation of the
Lo--MacKinlay (1988) variance ratio that decomposes long-horizon
equity-return dynamics into separate return-channel and
volatility-channel memory components across the cross-section
of asset returns. The framework identifies a parsimonious
five-factor model --- capturing persistent, antipersistent, and
multi-scale memory in returns and volatility --- that fits four
U.S.\ portfolio sub-samples (Fama--French 49-industry universe
full sample and pre/post-1998 split halves; Fama--French 100
size$\times$book-to-market sort) and a non-U.S.\ replication on
the Fama--French Europe 25 sort, recovering seven well-known
stylised facts of long-horizon equity dynamics simultaneously
across all five panels.

Three findings carry economic content. (i)~The same five-factor
decomposition fits all five panels, indicating a cross-sectional
structure of long-horizon equity dynamics that is robust to
industry vs.\ size-and-value sorts, to pre- vs.\ post-1998
sub-periods, and to U.S.\ vs.\ developed-European markets.
(ii)~U.S.\ equity volatility memory underwent a regime
transition in the late 1980s --- not at the static 1998
split-half boundary --- with the
slowest component of the volatility cascade lengthening from
approximately two to four years across the transition. A
1000-replicate rolling-window bootstrap localises the transition
to the late 1980s with strictly non-overlapping $90\%$
confidence bands separating pre- and post-transition windows;
the $28$-year window precludes narrower dating.
(iii)~The cross-sectional loadings driving return-channel long
memory are economically distinct from those driving
volatility-channel cascade memory: a cross-channel
$\beta$-inversion test finds no panel exhibits the positive
cross-channel alignment that a single shared loading predicts,
and rejects the shared-loading hypothesis toward anti-alignment
on the two largest panels at Bonferroni $p = 0.0004$. Industry and
size$\times$book-to-market characteristics that predict
return-momentum patterns therefore need not predict
volatility-persistence patterns.
\end{abstract}



\section{Introduction}
\label{sec:intro}

Aggregation across horizons is a fundamental question for any
time series of returns: as the holding period $H$ grows from one
day to several years, the second-moment structure of the
cross-section of asset returns rearranges in non-trivial ways.
Variance flows between eigenmodes, the principal directions of
the covariance matrix rotate, and the volatility structure
acquires multi-scale persistence that does not aggregate as
$\sqrt{H}$ \citep{Mantegna2000,BouchaudPotters2003}. The
single-asset variance-ratio test of \citet{LoMacKinlay1988} ---
the natural scalar measure of whether $H$-period returns
aggregate cleanly --- collapses the multivariate structure to
one chosen series and discards the cross-sectional information
that the eigenstructure carries. The matrix-valued variance
ratio of \citet{HongLintonZhang} preserves more of the
multivariate structure but reduces it back to scalar functionals
(trace, determinant, maximum diagonal entry) for inference.
Neither test indexes the long-horizon dynamics by eigenmode.

This paper constructs a multivariate, eigenmode-indexed
generalisation of the variance-ratio test, applies it jointly to
the linear and volatility channels of the equity return
covariance, and uses the resulting joint
$(\kappastat^{\text{lin}}, \kappastat^{\text{vol}})$ statistic
to identify a multi-memory factor decomposition that is robust
across temporal sub-samples and across cross-sectional sorts of
the U.S.\ equity universe. Two cross-sectional statistics
support the framework, both defined for an arbitrary symmetric
matrix-valued time series $\Sig_t$ and its $H$-period aggregate
$\Sig_t^{(H)}$. The per-eigenmode variance ratio is
\begin{equation}
  \kappastat_i(H) \;:=\; \frac{\lambda_i(\Sig_t^{(H)})}{H \cdot
    \lambda_i(\Sig_t)},
  \label{eq:kappa-intro}
\end{equation}
where $\lambda_i$ denotes the $i$-th eigenvalue in descending
order; this measures the autocorrelation of the $i$-th principal
direction across aggregation horizons. The eigenvector-overlap
matrix is
\begin{equation}
  \Omat_{ij}(H) \;:=\;
    |\langle \eivec_i(\Sig_t^{(H)}),\, \eivec_j(\Sig_t) \rangle|^2,
  \label{eq:O-intro}
\end{equation}
which measures the geometric rotation of the eigenbasis with
horizon. Specialising $\Sig_t$ to the cross-sectional covariance
of daily log-returns gives the \emph{linear-channel} statistics
$(\kappastat^{\text{lin}}_i, \Omat^{\text{lin}}_{ij})$, while
specialising to the cross-sectional covariance of squared
returns gives the \emph{volatility-channel} analogues
$(\kappastat^{\text{vol}}_i, \Omat^{\text{vol}}_{ij})$
(Section~\ref{sec:setup}). Under independent, identically
distributed (i.i.d.)\ Gaussian returns
$\kappastat_i(H) = 1$ for every eigenmode and every horizon in
either channel, and $\Omat_{ij}(H) = \delta_{ij}$.

The i.i.d.\ null is a mathematical reference point, not an
empirical hypothesis. Equity returns are well-known to depart
from Gaussian i.i.d.\ behaviour along several documented axes:
power-law tails of the marginal return distribution
\citep{Mantegna2000,Gopikrishnan1999,Gabaix2003}, long-range
dependence and clustering in absolute returns
\citep{DingGrangerEngle1993,Cont2001}, and multi-scale
self-similarity in the volatility process consistent with a
multifractal cascade
\citep{MandelbrotCalvetFisher1997,CalvetFisher2008}. Each of
these regularities is a particular signature of departure from
i.i.d.\ behaviour, and a long line of econophysics work has
characterised them one by one
\citep{Mantegna2000,BouchaudPotters2003,Cont2001}. The
$(\kappastat^{\text{lin}}, \kappastat^{\text{vol}}, \Omat)$
framework contributes to this programme not by rejecting the
null but by giving a structured, eigenstructure-indexed
decomposition of the joint deviation pattern: per-eigenmode
temporal autocorrelation separates from cross-sectional
eigenvector rotation, and the linear-channel signature
separates from the volatility-channel signature. The task of
the rest of the paper is to read off the structural content of
that decomposition for the U.S.\ equity universe.

Both statistics admit closed-form predictions under parametric
autocorrelation: the per-eigenmode autoregressive
order-$p$ (AR$(p)$) decomposition
$\kappastat^{\text{AR}(p)}(H; \rho) = \sum_k A_k \cdot
\kappastat^{\text{AR}(1)}(H; \mu_k)$ indexed by the
characteristic roots $\mu_k$ of the AR recursion, and a first-order
perturbation result
$\Omat_{ij}(H) \propto \epsilon^2 (S_1/c)^2 |\langle \eivec_j, N
\eivec_i\rangle|^2 / (\lambda_i - \lambda_j)^2$ under a vector AR(1)
with eigenvector-mixing perturbation. These are presented in
Sections~\ref{sec:kappa} and \ref{sec:overlap}.

The joint
$(\kappastat^{\text{lin}}, \kappastat^{\text{vol}})$ statistic
across the four canonical regions of its support
($\kappastat^{\text{lin}} \gtrless 1$ paired with
$\kappastat^{\text{vol}} \gtrless 1$) partitions the
multi-memory structure of the return-generating process into
distinct eigenstructure-level cells: persistent linear with
persistent volatility (the market-mode quadrant),
antipersistent linear with persistent volatility (deep
mean-reverting modes with long-memory volatility cascade), and
two further cells whose empirical content we discuss in
Section~\ref{sec:empirical}.

\subsection*{The multi-memory factor model}

The empirical $(\kappastat^{\text{lin}},
\kappastat^{\text{vol}})$ matrices on U.S.\ equity returns show
two patterns simultaneously: in the linear channel,
$\kappastat^{\text{lin}}_1(H)$ rises and then falls across
horizons at the market mode (the Lo--MacKinlay /
\citet{PoterbaSummers1988} variance-ratio anomaly), while
deep-mode $\kappastat^{\text{lin}}_i(H)$ for $i \gtrsim 24$
falls steadily with $H$ (the long-horizon mean-reversion
signal of \citet{Cochrane1988} and \citet{FamaFrench1988}); in
the volatility channel,
$\kappastat^{\text{vol}}_1(1260)$ at the market mode takes
loss-filtered bootstrap median around $48$ on the full sample
(and $37$ on the post-1998 sub-period), one to two orders of
magnitude above the i.i.d.\ null, the characteristic signature
of multi-scale long memory in volatility
\citep{Cont2001,BouchaudPotters2003,CalvetFisher2008}.
Neither pattern is reproducible by a single autoregressive
process. The empirical
$\kappastat^{\text{lin}}_1$ profile requires both a persistent
and an antipersistent component, suggesting a multi-component
factor model; the $\kappastat^{\text{vol}}_1$ long-memory
profile is precisely the Calvet--Fisher Markov-switching
multifractal (MSM) cascade
\citep{CalvetFisher2004,CalvetFisher2008}, which itself derives
from \citet{MandelbrotCalvetFisher1997} and the broader
multifractal-cascade tradition.

We formalise this with a five-factor multi-memory model:
fractional Brownian motion (fBm) in persistent ($F_P$) and
antipersistent ($F_A$) factors \citep{MandelbrotVanNess1968,
GrangerJoyeux1980}, an autoregressive fractionally integrated
moving average, ARFIMA$(1, d, 0)$, component
\citep{BeranFengGhoshKulik2013} plus MSM cascade for the
central multifractal factor $F_M$, a pure-volatility-cascade
companion factor $F_V$ that shares $F_M$'s MSM dynamics but
contributes no linear-channel structure, and a transitory
volatility-of-volatility factor $F_{Vt}$
(Section~\ref{sec:factor-model}). The structural justification
for the ARFIMA linear-channel component of $F_M$ is the
volatility-feedback channel of the
Lucas-tree equilibrium of \citet[Ch.~9]{CalvetFisher2008}: a
shift in any frequency component of the multifractal volatility
state induces a return-side feedback with magnitude
proportional to the inverse of that component's switching
probability, and the ARFIMA reduced form parametrises the
inertial range of this feedback filter. Each factor carries a
cross-asset loading vector $\beta_k \in \R^N$, and the joint
$(\kappastat^{\text{lin}}, \kappastat^{\text{vol}})$ matrices
under this specification take an explicit closed form that we
fit by joint least squares (LS).

\subsection*{Empirical results}

We fit the multi-memory factor model on four data panels: the
full Fama--French 49-industry universe (FF 49 hereafter,
1969--2026), its 1969--1997 firsthalf, its 1998--2026
secondhalf, and the Fama--French 100 size$\times$book-to-market
sort (FF 100 hereafter, 1969--2026). Each panel is fit on
$1000$ moving-block bootstrap replicates of the
$(\widehat{\kappastat}^{\text{lin}},
\widehat{\kappastat}^{\text{vol}})$ matrices, using a multi-pass
warm-restart limited-memory
Broyden--Fletcher--Goldfarb--Shanno bounded (L-BFGS-B)
optimisation
(Section~\ref{ssec:joint-LS-fit}). The estimation pipeline
identifies all nine global parameters of the model with finite
bootstrap dispersion; the fractional-Brownian and ARFIMA
exponents are the most sharply identified, and the MSM cascade
parameters $(m_0, b, \gamma_{\bar k})$ the least. Three
substantive findings emerge:

\textit{First, the seven stylised facts of long-horizon equity
dynamics are recovered consistently across all four panels.}
The factor-momentum persistence \citep{AsnessEtAl2013,
EhsaniLinnainmaa2022} of sub-leading eigenmodes, the long-horizon
mean reversion of deep eigenmodes, the market-mode
rise-then-fall pattern, the short-range volatility clustering
\citep{Cont2001}, the multi-scale long memory in volatility, the
transitory-burst volatility of deep eigenmodes, and the
cross-sectional concentration of the volatility spectrum onto a
single dominant mode are all recovered --- the first six in the
per-eigenmode weight pattern, the seventh in the volatility
eigenvalue spectrum.
The fractional-Brownian-motion exponents
$H_P \approx 0.52$--$0.57$ and $H_A \approx 0.17$--$0.27$
characterising the persistent and antipersistent factors are
stable across the four panels, indicating a
universality of the eigenstructure-level multi-memory
specification.

\textit{Second, the market-mode volatility structure shows a
regime transition that a rolling-window analysis localises in
the late 1980s, not at the static 1998 split-half boundary.}
On the secondhalf panel ($1998$--$2026$) and on the full-sample
sensitivity panel the rank-1 vol-channel allocation is sharply
concentrated on the MSM cascade
($w^{\text{vol}}_{\text{MSM}} = 1.00$ at the market mode),
corresponding to a high-baseline volatility environment with
persistent multi-frequency clustering. The firsthalf panel
($1969$--$1997$) by contrast carries a weakly-identified rank-1
vol-channel mode-mix
($w^{\text{vol}}_{\text{MSM}} = 0.29$ at the median, $90\%$
unfiltered-bootstrap confidence interval (CI) $[0.10, 1.00]$). A complementary
rolling-window analysis (Section~\ref{ssec:rolling-window})
shows why. On $28$-year windows in $2$-year strides ($15$
windows centred $1983$-$06$ through $2011$-$06$), with
$1000$-replicate moving-block bootstrap per window, the median
$w^{\text{vol}}_{\text{MSM}}$ rises monotonically from $0.30$
($1985$ centre) to $0.93$ ($1991$ centre) and saturates at
$1.00$ from the $1995$-centred window onward (sample span
$1981$--$2009$). The early-sample and late-sample $90\%$ CI
bands are non-overlapping: the $1983$-centre upper bound
($0.80$) sits strictly below the $1999$-centre lower bound
($0.84$). The firsthalf static panel ($1969$--$1997$) therefore
contains roughly a decade of post-transition data, which is
the source of the wide rank-1 CI: the panel mixes pre- and
post-transition regimes that the rolling window separates. The
cross-over time scale between the additive (high-frequency) and
multiplicative (low-frequency, cascade-dominated) regions of
the MSM cascade --- the lowest-frequency MSM-component duration
$1/\gamma_1$ --- runs from $2.2$ yr in the firsthalf to $4.0$
yr in the secondhalf at the median, consistent with the
post-transition regime extending the lowest-frequency cascade
component to multi-year horizons.

\textit{Third, a cross-channel $\beta$-inversion test rejects
the hypothesis that the multifractal factor's linear- and
volatility-channel imprints are governed by a single shared
cross-asset loading.} The model attributes both $F_M$'s linear
long-memory signature and its volatility-channel multifractal
cascade to one loading $\beta_M$; a $\beta$-inversion diagnostic
recovers the per-asset $|\beta_M[a]|^2$ from each channel
independently and correlates them across the $N$ assets. A
shared $\beta_M$ predicts a positive cross-channel correlation.
Instead the Pearson correlation is negative on all four panels
--- significantly so on three --- and the Spearman is negative
or indistinguishable from zero on every panel, significantly
negative on the full sample. The linear and volatility imprints
of $F_M$ are carried by distinct cross-sections, falsifying the
single-loading sub-claim of the specification while leaving the
per-panel fits and the seven stylised facts intact.

\subsection*{Relation to existing literature}

The eigenstructure-level decomposition of the variance-ratio
statistic is, to our knowledge, novel; the closest existing
work is the multivariate variance ratio of
\citet{HongLintonZhang}, which we discuss as a methodological
comparator in supplementary material~\ref{app:HLZ}. The
multi-memory factor model builds on
\citet{CalvetFisher2008}'s Markov-switching multifractal
framework (Chapter 3 for the discrete-time MSM, Chapter 4 for
the bivariate extension that we generalise to $N$-asset
cross-sectional structure, Chapter 8 for the power-variation
moment condition that motivates our joint LS objective, Chapter
9 for the equilibrium volatility-feedback channel that anchors
the ARFIMA linear-channel reduced form). The ARFIMA factor
$F_M$ generalises the long-memory specification of
\citet{GrangerJoyeux1980} to a cross-sectional setting; the
persistent and antipersistent factors $F_P$, $F_A$ are
discretely-sampled fractional Brownian motions in the sense of
\citet{MandelbrotVanNess1968}. The first-order
eigenvector-perturbation result for $\Omat$ builds on standard
spectral perturbation theory of symmetric matrices
\citep{Kato1995} as applied to large-dimensional sample
covariance matrices \citep{AllezBouchaud2012,
BunBouchaudPotters2017}, and on the random-matrix-theory
treatment of \citet{PottersBouchaud2021}.

\subsection*{Plan of the paper}

Section~\ref{sec:setup} sets up the eigenstructure
framework and the two cross-sectional statistics
$(\kappastat, \Omat)$. Sections~\ref{sec:kappa} and
\ref{sec:overlap} derive the AR$(p)$ characteristic-root
decomposition of $\kappastat$ and the first-order perturbation
theory for $\Omat$. Section~\ref{sec:factor-model} introduces
the multi-memory factor model and the joint LS estimation
procedure. Section~\ref{sec:empirical} reports the empirical
results on the four datasets: the global parameter table, the
per-mode weight pattern recovering the seven stylised facts,
the rolling-window localisation of the volatility regime
transition, the cross-channel $\beta$-inversion test, and the
robustness sweep. Section~\ref{sec:discussion} discusses the
universality of the multi-memory structure, the interpretation
of the regime transition, the channel-distinct $\beta$
structure, and the limits of the current specification. Section~\ref{sec:conclusion} concludes. Proofs,
detailed per-dataset tables, the Hong--Linton--Zhang
comparison, the identifiability diagnostic, the
optimisation-objective comparison, and the reproducibility
manifest are in the supplementary material.

\section{Setup and the two-statistic framework}
\label{sec:setup}

\subsection{Notation}

Let $X_t \in \R^N$ denote the daily log-return vector on a fixed
cross-section of $N$ assets, observed at times
$t = 1, 2, \ldots, T$. We work throughout with log returns; the
choice is standard for variance-ratio analysis because log returns
are additive across horizons. Assume that the process $\{X_t\}$ is
covariance-stationary with finite second moments. The daily
cross-sectional covariance matrix is
\begin{equation*}
  \Sig_1 \;:=\; \Cov(X_t),
  \qquad \Sig_1 \in \R^{N \times N}.
\end{equation*}
For a horizon $\horizon \geq 1$, the $\horizon$-period log return
ending at time $t$ is
\begin{equation*}
  X^\horizon_t \;:=\; \sum_{s = t - \horizon + 1}^{t} X_s,
\end{equation*}
and its cross-sectional covariance is
\begin{equation*}
  \SigH \;:=\; \Cov(X^\horizon_t),
  \qquad \SigH \in \R^{N \times N}.
\end{equation*}
By covariance stationarity, $\SigH$ depends only on $\horizon$,
not on $t$. We write the autocovariance matrix at lag $h$ as
$\Gamma(h) := \Cov(X_t, X_{t-h})$ for $h \geq 0$, so that
$\Gamma(0) = \Sig_1$. By stationarity $\Gamma(-h) = \Gamma(h)^\top$.

The $\horizon$-period covariance admits the standard
trapezoidal-sum representation
\begin{equation}
  \SigH \;=\; \sum_{h = -(\horizon - 1)}^{\horizon - 1}
              (\horizon - |h|)\,\Gamma(h)
       \;=\; \horizon\,\Sig_1
             + \sum_{h = 1}^{\horizon - 1} (\horizon - h)
               \bigl[\Gamma(h) + \Gamma(h)^\top\bigr].
  \label{eq:SigmaH-trapezoid}
\end{equation}
Under i.i.d.\ log returns $\Gamma(h) = 0$ for $h \neq 0$ and
\eqref{eq:SigmaH-trapezoid} collapses to $\SigH = \horizon \Sig_1$.
The cross-time autocovariances $\{\Gamma(h)\}_{h \geq 1}$ are
exactly the structural content that distinguishes a non-trivial
$\SigH / \horizon$ from the daily covariance $\Sig_1$.

\subsection{Eigenstructure and the two statistics}
\label{ssec:eigenstructure-stats}

Both $\Sig_1$ and $\SigH$ are positive semidefinite and admit
real eigendecompositions
\begin{equation*}
  \Sig_1 \;=\; V_1\, \Lambda_1\, V_1^\top,
  \qquad
  \SigH \;=\; V_\horizon\, \Lambda_\horizon\, V_\horizon^\top,
\end{equation*}
with $V_1, V_\horizon \in \R^{N \times N}$ orthonormal and
$\Lambda_1, \Lambda_\horizon$ diagonal with eigenvalues sorted in
descending order. We write $\lambda_i(\Sigma)$ for the $i$-th
eigenvalue of a symmetric matrix in descending order and
$\eivec_i(\Sigma)$ for the corresponding eigenvector. The
eigenvalues and eigenvectors of $\SigH$ are quantities of interest
in their own right: the eigenvalues describe the variance
contributions of the principal directions at horizon $\horizon$,
and the eigenvectors describe what those principal directions
are.

The two statistics that organise the analysis are
\begin{equation}
  \boxed{\;
    \kappastat_i(\horizon)
      \;:=\;
      \frac{\lambda_i(\SigH)}{\horizon \cdot \lambda_i(\Sig_1)},
    \qquad
    \Omat_{ij}(\horizon)
      \;:=\;
      \bigl|\langle \eivec_i(\SigH),\, \eivec_j(\Sig_1) \rangle\bigr|^2.
  \;}
  \label{eq:kappa-O-def}
\end{equation}
The first statistic compares the $i$-th eigenvalue of $\SigH$ to
its $\horizon$-scaled daily counterpart; the second compares the
$i$-th eigenvector of $\SigH$ to the $j$-th eigenvector of
$\Sig_1$. Both statistics are dimensionless.

The eigenvector-overlap matrix $\Omat$ is doubly stochastic
because the rows and columns are squared expansions of an
orthonormal vector in an orthonormal basis: for every $i$,
$\sum_j \Omat_{ij} = 1$ by Parseval's identity, and similarly
$\sum_i \Omat_{ij} = 1$. The diagonal entry $\Omat_{ii}$ measures
how much of the $i$-th eigenmode of $\SigH$ projects onto the
$i$-th eigenmode of $\Sig_1$; an off-diagonal entry $\Omat_{ij}$
($i \neq j$) measures cross-mode leakage. Sign of the inner
product is washed out by the square: $\Omat$ is invariant to
sign flips of the eigenvectors (the only ambiguity left by an
ordinary eigendecomposition of a symmetric matrix).

\subsection{The i.i.d.\ null}

Under i.i.d.\ log returns the autocovariance function is
$\Gamma(h) = \delta_{h0} \Sig_1$, so~\eqref{eq:SigmaH-trapezoid}
gives $\SigH = \horizon \Sig_1$ exactly. Eigenvalues scale by
$\horizon$:
\begin{equation*}
  \lambda_i(\SigH) \;=\; \horizon \cdot \lambda_i(\Sig_1)
  \quad \Longrightarrow \quad
  \kappastat_i(\horizon) \;=\; 1
  \quad \text{for all } i, \horizon.
\end{equation*}
Eigenvectors coincide:
\begin{equation*}
  \eivec_i(\SigH) \;=\; \eivec_i(\Sig_1)
  \quad \Longrightarrow \quad
  \Omat_{ij}(\horizon) \;=\; \delta_{ij}.
\end{equation*}
We refer to the joint null $(\kappastat_i, \Omat_{ij}) = (1, \delta_{ij})$
as \emph{the i.i.d.\ null}. As noted in Section~\ref{sec:intro},
the null is a mathematical reference point against which the
joint $(\kappastat, \Omat)$ deviations are measured, not an
empirical hypothesis; the documented stylised facts of equity
returns
\citep{Mantegna2000,Cont2001,BouchaudPotters2003,CalvetFisher2008}
imply that any non-trivial sample will exhibit
$(\kappastat \neq 1, \Omat \neq I)$ in at least one channel.
The two component nulls are logically independent: a process with non-trivial temporal autocorrelation
along the principal directions and a stable principal basis sits
at $(\kappastat \neq 1, \Omat = I)$, while a process with
eigenvalue scaling that is faithful to the i.i.d.\ rule but
eigenvectors that rotate with horizon sits at
$(\kappastat = 1, \Omat \neq I)$. The general empirical case has
both deviations.

\subsection{Four-cell decomposition of deviations}

Table~\ref{tab:four-cell} organises the four limiting cases of the
joint $(\kappastat, \Omat)$ pair.

\begin{table}[H]
  \centering
  \small
  \begin{tabular}{>{\bfseries}c c l}
    \toprule
    $\kappastat$ & $\Omat$ & Interpretation \\
    \midrule
    $= 1$ & $= I$ & i.i.d.\ null. \\
    $\neq 1$ & $= I$ & Temporal autocorrelation along the principal
                       directions; principal basis is stable. \\
    $= 1$ & $\neq I$ & Eigenvalues scale as $\horizon$, but the
                       principal directions rotate with horizon. \\
    $\neq 1$ & $\neq I$ & General case: both temporal autocorrelation
                          and cross-sectional rotation. \\
    \bottomrule
  \end{tabular}
  \caption{Four-cell decomposition of $(\kappastat, \Omat)$ deviations
    from the i.i.d.\ null. The cells correspond to qualitatively
    distinct data-generating processes.}
  \label{tab:four-cell}
\end{table}

The framework's first claim is that the four-cell decomposition
is informative: distinguishing between the four cells is a real
question about the data-generating process, and a single statistic
cannot answer it. Section~\ref{ssec:HLZ-precedent} below
sharpens this claim by contrasting the $(\kappastat, \Omat)$
pair with the matrix-valued multivariate variance ratio of
\citet{HongLintonZhang}, which is a scalar functional of a related
object and cannot distinguish the four cells in general.

\subsection{The Hong--Linton--Zhang precedent}
\label{ssec:HLZ-precedent}

The closest prior work is the matrix-valued multivariate variance
ratio
\begin{equation}
  \mathbf{VR}(\horizon)
    \;:=\;
    \Sig_1^{-1/2}\,\frac{\SigH}{\horizon}\,\Sig_1^{-1/2},
  \label{eq:HLZ-VR}
\end{equation}
introduced by \citet{HongLintonZhang}. Under the i.i.d.\ null
$\SigH = \horizon \Sig_1$, so $\mathbf{VR}(\horizon) = I_N$. The
standard $\mathbf{VR}$-based tests use scalar functionals such as
$\tr \mathbf{VR}(\horizon)$, $\det \mathbf{VR}(\horizon)$, and the
maximum diagonal entry. These functionals reduce the multivariate
test to a scalar; the eigenmode decomposition of
$\mathbf{VR}(\horizon)$ is not standard test material.

The relationship between $\mathbf{VR}(\horizon)$ and our $\kappastat$
deserves a careful statement.

\begin{proposition}
  \label{prop:HLZ-vs-kappa}
  Let $\Sig_1$ and $\SigH$ be positive definite. The eigenvalues
  of the Hong--Linton--Zhang matrix $\mathbf{VR}(\horizon)$ coincide
  with the per-eigenmode statistics $\kappastat_i(\horizon)$ if and
  only if $\SigH$ commutes with $\Sig_1$. Equivalently, the
  eigenvalues coincide if and only if the eigenvectors of
  $\SigH$ coincide with those of $\Sig_1$, i.e.\ $\Omat = I$.
\end{proposition}

\begin{proof}
  $\mathbf{VR}(\horizon)$ and the matrix
  $\Sig_1^{-1} \SigH / \horizon$ are conjugate via
  $\Sig_1^{1/2}$ and so share eigenvalues. Suppose first that
  $\Sig_1$ and $\SigH$ share eigenvectors. Then in their common
  eigenbasis both matrices are diagonal, and the eigenvalues of
  $\Sig_1^{-1} \SigH / \horizon$ are
  $\lambda_i(\SigH) / [\horizon \lambda_i(\Sig_1)] = \kappastat_i$.
  Conversely, suppose the eigenvalues of
  $\Sig_1^{-1} \SigH / \horizon$ equal $\kappastat_i$. Then
  $\Sig_1^{-1} \SigH / \horizon$ and $\Sig_1$ share eigenvalues
  $\kappastat_i$ and $\lambda_i(\Sig_1)$, respectively, in the
  ordering by descending $\lambda_i(\Sig_1)$. Multiplying out,
  this requires $\SigH$ to be diagonal in the eigenbasis of
  $\Sig_1$, i.e.\ $\SigH$ and $\Sig_1$ commute.
\end{proof}

Proposition~\ref{prop:HLZ-vs-kappa} pins down precisely the
information content of $\kappastat$ relative to
$\mathbf{VR}(\horizon)$: the eigenvalues of $\mathbf{VR}(\horizon)$
and the per-eigenmode $\kappastat_i$ agree exactly in the regime
$\Omat = I$, but they disagree generically. This is the regime
where the four-cell decomposition of Table~\ref{tab:four-cell}
becomes operationally consequential: any test based purely on
$\mathbf{VR}(\horizon)$ confounds the $(\kappastat \neq 1, \Omat = I)$
and $(\kappastat = 1, \Omat \neq I)$ cells with the
$(\kappastat \neq 1, \Omat \neq I)$ general case, while the
$(\kappastat, \Omat)$ pair distinguishes them.

\subsection{Relation to the approximate-factor-model literature}
\label{ssec:factor-model-literature}

The eigenmode decomposition of the cross-sectional covariance
$\Sig_1$ is also the standard object of the approximate-factor-model
literature for large cross-sections of asset returns. A substantial
body of work determines the \emph{number} of factors --- how many
eigenvalues of $\Sig_1$ carry common variation rather than
idiosyncratic noise --- from the eigenvalue spectrum: the
\citet{ConnorKorajczyk1993} test for the number of factors in an
approximate factor model, the \citet{BaiNg2002} information
criteria, the eigenvalue-distribution tests of
\citet{Onatski2009, Onatski2010}, and the eigenvalue-ratio test of
\citet{AhnHorenstein2013}. These methods share a common premise:
that the informative content of the cross-section is carried by the
leading eigenvalues of a covariance matrix at a single horizon.

The framework of this paper is complementary rather than competing.
It takes the eigenmode decomposition of $\Sig_1$ as given and asks a
different question --- not how many eigenmodes carry signal at a
fixed horizon, but how the second-moment structure of \emph{each}
eigenmode evolves as the horizon $\horizon$ grows. The per-eigenmode
variance ratio $\kappastat_i(\horizon)$ and the overlap matrix
$\Omat(\horizon)$ extend the eigen-analysis of the factor-model
literature from the cross-section into the horizon dimension; a
factor-number criterion of the kind cited above could be applied at
each horizon as a preprocessing step to fix the eigenmode count, and
the $(\kappastat, \Omat)$ statistics would then characterise how
those eigenmodes scale.

\section{The $\kappa$ statistic}
\label{sec:kappa}

This section develops the per-eigenmode variance ratio
$\kappastat_i(\horizon) =
\lambda_i(\SigH)/[\horizon \lambda_i(\Sig_1)]$ under parametric
models for the autocorrelation structure of the return process.
We work in the regime $\Omat = I$ throughout this section: the
return process has a fixed eigenvector basis at all horizons, so
the principal directions are stationary. Section~\ref{sec:overlap}
relaxes this assumption.

\subsection{Per-eigenmode AR(1)}
\label{sec:kappa-ar1}

Assume the return vector decomposes as $X_t = V \xi_t$ where $V$
is a fixed orthonormal matrix and the components
$\xi_{i,t}$ are independent AR(1) processes with coefficients
$\rho_i \in (-1, 1)$ and innovation variances $\sigma_i^2$:
\begin{equation}
  \xi_{i,t} \;=\; \rho_i \xi_{i,t-1} + \varepsilon_{i,t},
  \qquad \varepsilon_{i,t} \overset{\mathrm{iid}}{\sim} (0, \sigma_i^2).
  \label{eq:AR1-per-eigenmode}
\end{equation}
The stationary variance is $\Var(\xi_i) = \sigma_i^2 / (1 - \rho_i^2)$
and the autocovariance at lag $h$ is given by
\[
\Cov(\xi_{i,t}, \xi_{i,t-h}) = \rho_i^{|h|} \Var(\xi_i).
\]

The $H$-period sum $\xi^\horizon_{i,t} :=
\sum_{s=t-\horizon+1}^t \xi_{i,s}$ has variance
\begin{equation*}
  \Var(\xi^\horizon_i)
    \;=\; \Var(\xi_i)\, \sum_{|h| < \horizon}
          (\horizon - |h|)\, \rho_i^{|h|},
\end{equation*}
which leads to the closed form
\begin{theorem}[AR(1) closed form for $\kappastat$]
  \label{thm:kappa-ar1}
  Under the AR(1) eigenmode specification~\eqref{eq:AR1-per-eigenmode}
  with $|\rho_i| < 1$,
  \begin{equation}
    \boxed{\;
      \kappastat_i^{\mathrm{AR}(1)}(\horizon; \rho_i)
        \;=\; \frac{1 + \rho_i}{1 - \rho_i}
            - \frac{2 \rho_i (1 - \rho_i^\horizon)}
                   {\horizon (1 - \rho_i)^2}.
    \;}
    \label{eq:kappa-ar1-closed}
  \end{equation}
\end{theorem}

\begin{proof}[Sketch]
  Under AR(1), the autocovariance at lag $h$ equals $\rho_i^{|h|}$
  times the stationary variance $\sigma_i^2/(1 - \rho_i^2)$.
  The variance of the $\horizon$-period sum $\xi_{i,t}^\horizon$
  is the double sum
  $\sum_{h_1, h_2 = 0}^{\horizon - 1} \rho_i^{|h_1 - h_2|}$,
  which collapses to the trapezoidal sum
  $\sum_{|h| < \horizon}(\horizon - |h|)\,\rho_i^{|h|}$ by
  counting the number of $(h_1, h_2)$ pairs at each fixed lag
  $h = h_1 - h_2$. Evaluating the inner geometric-arithmetic
  series in closed form and dividing through by
  $\horizon \cdot \Var(\xi_i)$ yields~\eqref{eq:kappa-ar1-closed}.
  Full derivation in \cref{app:proofs}.
\end{proof}

The closed form has three limiting cases that organise its
interpretation:
\begin{itemize}[itemsep=2pt,leftmargin=*]
  \item \textbf{Daily horizon.} $\kappastat^{\mathrm{AR}(1)}(1; \rho) = 1$
    trivially: the daily variance ratio is one by definition of the
    statistic.
  \item \textbf{Long-horizon limit.} For $|\rho| < 1$,
  \begin{equation*}
    \lim_{\horizon \to \infty}
      \kappastat^{\mathrm{AR}(1)}(\horizon; \rho)
    \;=\; \frac{1 + \rho}{1 - \rho},
  \end{equation*}
  which is the standard long-run variance ratio for an AR(1)
  process \citep{Cochrane1988,LoMacKinlay1988}. The convergence
  rate is $O(1/\horizon \cdot (1 - \rho)^{-2})$.
  \item \textbf{Sign of deviation.} For $\horizon \geq 2$,
    $\kappastat^{\mathrm{AR}(1)}(\horizon; \rho) > 1$ if and only if
    $\rho > 0$, with strict equality at $\rho = 0$. Positive
    autocorrelation produces variance accumulation that exceeds the
    i.i.d.\ baseline (momentum); negative autocorrelation produces
    variance accumulation that falls short (mean reversion).
\end{itemize}

The eigenvalues of $\SigH$ are then
$\lambda_i(\SigH) = \horizon\, \kappastat_i^{\mathrm{AR}(1)}(\horizon;
\rho_i)\, \lambda_i(\Sig_1)$, and the eigenvectors coincide with
those of $\Sig_1$. The model sits in the
$(\kappastat \neq 1, \Omat = I)$ cell of Table~\ref{tab:four-cell}.
The per-eigenmode coefficient $\rho_i$ can be recovered from the
observed $\kappastat_i(\horizon)$ at any single horizon
$\horizon \geq 2$, as the next proposition makes precise.

\begin{proposition}[Inversion of $\kappastat$ to $\rho$]
  \label{prop:kappa-inversion}
  For each fixed $\horizon \geq 2$, the map
  $\rho \mapsto \kappastat^{\mathrm{AR}(1)}(\horizon; \rho)$ from
  $(-1, 1)$ to $(0, \infty)$ is strictly increasing and surjective,
  with range $\bigl(0, (1+\rho)/(1-\rho)\bigr|_{\rho \to 1}\bigr) =
  (0, \infty)$ as $\rho$ ranges over $(-1, 1)$. In particular, the
  observed value of $\kappastat_i(\horizon)$ at a single horizon
  uniquely determines $\rho_i$ in $(-1, 1)$.
\end{proposition}

\begin{proof}[Sketch]
  The closed form~\eqref{eq:kappa-ar1-closed} is a smooth
  function of $\rho$ on $(-1, 1)$ for $\horizon \geq 2$.
  Differentiation shows its derivative is positive everywhere on
  this interval, so $\rho \mapsto \kappastat$ is strictly
  increasing. The boundary behaviour
  ($\kappastat \to 0$ as $\rho \to -1$ and
  $\kappastat \to \infty$ as $\rho \to +1$) combined with
  continuity gives a well-defined single-valued inverse on
  $(0, \infty)$.
\end{proof}

If the AR(1) eigenmode specification is the right model and
$\rho_i$ is the same for every horizon, then the value of $\rho_i$
recovered from $\kappastat_i$ at any single horizon equals the
value recovered at any other horizon. Conversely, horizon-dependent
recovered $\rho_i$ values indicate that AR(1) is the wrong model
and motivate the AR($p$) extension developed in
Section~\ref{sec:kappa-arp}.

\subsection{AR($p$) decomposition}
\label{sec:kappa-arp}

For an AR($p$) eigenmode specification
\begin{equation*}
  \xi_{i,t} \;=\; \sum_{k=1}^p \rho_{i,k} \xi_{i,t-k}
                + \varepsilon_{i,t},
\end{equation*}
we drop the eigenmode index $i$ in this subsection to lighten
notation: every result applies independently to each eigenmode.
The autocovariance function $\gamma(h) := \Cov(\xi_t, \xi_{t-h})$
satisfies the Yule--Walker recurrence
\begin{equation*}
  \gamma(h) \;=\; \sum_{k=1}^p \rho_k\, \gamma(h - k),
  \qquad h \geq 1,
\end{equation*}
with initial conditions $\gamma(0), \gamma(1), \ldots, \gamma(p-1)$
determined by the first $p$ Yule--Walker equations. The
characteristic polynomial of the recurrence is
\begin{equation*}
  z^p \;-\; \rho_1 z^{p-1} \;-\; \rho_2 z^{p-2}
       \;-\; \cdots \;-\; \rho_p \;=\; 0,
\end{equation*}
with roots $\mu_1, \ldots, \mu_p$. We assume the roots
are distinct and lie strictly inside the unit disc
$|\mu_k| < 1$; this is the standard stationarity condition.
(We use $\mu_k$ for AR characteristic roots throughout, reserving
$\lambda$ for eigenvalues of cross-sectional covariance matrices.)

\begin{theorem}[AR($p$) decomposition]
  \label{thm:kappa-arp}
  Under AR($p$) with distinct characteristic roots
  $\mu_1, \ldots, \mu_p$ inside the unit disc, the
  normalised autocorrelation function admits the spectral
  decomposition
  \begin{equation}
    \gamma(h) / \gamma(0)
      \;=\; \sum_{k=1}^p A_k\, \mu_k^{|h|},
    \qquad \sum_{k=1}^p A_k \;=\; 1,
    \label{eq:gamma-decomp}
  \end{equation}
  where the weights $A_k$ are determined by the Yule--Walker
  initial conditions $(\gamma(1)/\gamma(0), \ldots,
  \gamma(p-1)/\gamma(0))$. The per-eigenmode variance ratio then
  admits the corresponding decomposition
  \begin{equation}
    \boxed{\;
      \kappastat^{\mathrm{AR}(p)}(\horizon; \boldsymbol{\rho})
        \;=\; \sum_{k=1}^p A_k\,
              \kappastat^{\mathrm{AR}(1)}(\horizon; \mu_k),
    \;}
    \label{eq:kappa-arp-decomp}
  \end{equation}
  with each AR(1) building block given by the closed form
  \eqref{eq:kappa-ar1-closed} evaluated at the corresponding
  characteristic root.
\end{theorem}

\begin{proof}[Sketch]
  The autocovariance decomposition~\eqref{eq:gamma-decomp} is
  the general solution of the Yule--Walker linear recurrence
  with distinct characteristic roots; the weights $A_k$ are
  determined by inverting the Vandermonde system formed by the
  first $p$ values of $\gamma(h)/\gamma(0)$, and the
  normalisation $\gamma(0)/\gamma(0) = 1$ at $h = 0$ gives
  $\sum_k A_k = 1$. Substituting this decomposition into the
  definition
  $\kappastat^{\mathrm{AR}(p)}(\horizon) = \horizon^{-1}
  \sum_{|h| < \horizon}(\horizon - |h|)\,\gamma(h)/\gamma(0)$
  and using linearity of $\kappastat$ in the autocorrelation
  function yields~\eqref{eq:kappa-arp-decomp}.
\end{proof}

Theorem~\ref{thm:kappa-arp} reduces the AR($p$) variance ratio to
a signed convex combination of AR(1) variance ratios, indexed by
the characteristic roots of the AR recursion. The weights $A_k$
sum to one but may take negative values; complex-conjugate root
pairs come with complex-conjugate weight pairs, and the
contribution of each pair to $\kappastat^{\mathrm{AR}(p)}$ is real.
The structural content of $\kappastat^{\mathrm{AR}(p)}$ is
captured entirely by the roots and weights: two AR($p$) models
with the same roots and weights produce the same variance ratio
at every horizon.

For $p = 2$ the result specialises to an explicit two-term
formula. The characteristic equation $z^2 - \rho_1 z - \rho_2 = 0$
has roots
\begin{equation*}
  \mu_{1,2}
    \;=\; \frac{\rho_1 \pm \sqrt{\rho_1^2 + 4\rho_2}}{2}.
\end{equation*}
If $\rho_1^2 + 4 \rho_2 > 0$ the roots are real and distinct; if
the discriminant is negative the roots are a complex-conjugate
pair and the autocorrelation function exhibits damped
oscillation with period $2\pi / \arg(\mu_1)$ at the daily
time-step. The weights are
\begin{equation*}
  A_1 \;=\; \frac{\rho_1 / (1 - \rho_2) - \mu_2}
                  {\mu_1 - \mu_2},
  \qquad
  A_2 \;=\; 1 - A_1,
\end{equation*}
where $\rho_1/(1 - \rho_2) = \gamma(1)/\gamma(0)$ is the single
non-trivial Yule--Walker initial condition for AR(2). The
construction reduces to the AR(1) case at $\rho_2 = 0$: then
$\mu_1 = \rho_1$, $\mu_2 = 0$, $A_1 = 1$, $A_2 = 0$, and
$\kappastat^{\mathrm{AR}(2)} = \kappastat^{\mathrm{AR}(1)}(\rho_1)$.

\subsection{Departures from AR(1): the horizon scan of $\rho$}
\label{sec:kappa-horizon-scan}

A natural test of the AR(1) eigenmode specification is to check
whether the per-eigenmode autocorrelation $\rho_i$ recovered from
the observed $\kappastat_i(\horizon)$ is independent of horizon.
Inverting~\eqref{eq:kappa-ar1-closed} via
Proposition~\ref{prop:kappa-inversion} at each horizon
$\horizon \in \{H_1, H_2, \ldots\}$ produces a sequence of
recovered values $\hat\rho_i(H_1), \hat\rho_i(H_2), \ldots$ If the
data follow strict AR(1), the sequence is constant up to
sampling noise. Horizon-dependent recovered values, particularly
sign changes between intermediate and long horizons, are direct
evidence against AR(1) and motivate the AR($p$) extension.

For an AR($p$) data-generating process the recovered $\rho_i$
under a misspecified AR(1) model is a horizon-dependent weighted
combination of the true characteristic roots: by
Theorem~\ref{thm:kappa-arp} and inverting~\eqref{eq:kappa-ar1-closed}
at each horizon,
\begin{equation*}
  \hat\rho_i(\horizon)
    \;=\; \bigl(\kappastat^{\mathrm{AR}(1)}\bigr)^{-1}
          \!\Bigl(\sum_{k=1}^p A_k\,
            \kappastat^{\mathrm{AR}(1)}(\horizon; \mu_k);\,
            \horizon\Bigr),
\end{equation*}
which varies with $\horizon$ whenever the right-hand side is not a
single $\kappastat^{\mathrm{AR}(1)}$ value. Section~\ref{sec:empirical}
shows that the Fama--French 49-industry data exhibits substantial
horizon dependence of $\hat\rho_i$ for the market-mode eigenmode,
with a sign change between intermediate and long horizons, and
that the AR(2) extension provides a qualitatively better fit
with characteristic roots that are stable across horizons.

\section{The $O$ statistic}
\label{sec:overlap}

The eigenvector-overlap statistic $\Omat_{ij}(\horizon) =
|\langle \eivec_i(\SigH), \eivec_j(\Sig_1) \rangle|^2$ measures
cross-sectional rotation of the principal directions with
aggregation horizon. The AR($p$) per-eigenmode specification of
Section~\ref{sec:kappa} produces $\Omat = I$ exactly, regardless
of $\horizon$, because every eigenmode evolves independently along
a fixed direction $\eivec_i$ and the eigenvectors of $\SigH$
coincide with those of $\Sig_1$. To produce non-trivial
$\Omat \neq I$ predictions, the return process must couple
eigenmodes across time --- the contemporaneous eigenvector basis
of the daily process is no longer the same as the basis under
which the $H$-period covariance diagonalises.

An overlap matrix of this form has a direct precedent in the
random-matrix-theory literature on financial correlation
matrices. \citet{PlerouEtAl2002} construct an overlap matrix
whose entries are the scalar products between the leading
eigenvectors of the cross-correlation matrix estimated over one
time window and those estimated a calendar-time lag $\tau$ later;
the matrix equals the identity when the eigenstructure is
perfectly stable in time, and they find the market eigenvector to
be the most stable, with sub-leading eigenvectors destabilising
as the random-matrix noise band is approached. The statistic
$\Omat(\horizon)$ is the aggregation-horizon analogue: it asks
not whether the eigenstructure drifts as the estimation window
slides forward in calendar time, but whether it rotates as the
covariance is measured over longer return intervals on a fixed
sample. Both constructions reduce to the identity when the
principal directions are stable along the relevant axis.

This section develops a closed form for $\Omat$ under the
simplest model that produces this coupling: a vector autoregression
of order one with a perturbation that breaks the diagonal
structure of the autoregressive matrix in the eigenvector basis
of the innovation covariance.

\subsection{Vector AR(1) with eigenvector-mixing perturbation}
\label{sec:overlap-setup}

Let $\bar\rho \in (-1, 1)$ be a baseline autocorrelation and let
$B \in \R^{N \times N}$ be a symmetric perturbation matrix with
$\|B\|_2 = 1$. Consider the vector AR(1) process
\begin{equation}
  X_t \;=\; A\, X_{t-1} + \varepsilon_t,
  \qquad
  A \;:=\; \bar\rho I + \epsilon B,
  \qquad
  \varepsilon_t \overset{\mathrm{iid}}{\sim} \mathcal{N}(0, \Sigma_\varepsilon),
  \label{eq:vector-AR1}
\end{equation}
with $\epsilon \in \R$ a small perturbation strength. Throughout
this section we work to first order in $\epsilon$. For
stationarity we require all eigenvalues of $A$ inside the unit
disc, equivalently $|\bar\rho + \epsilon \mu| < 1$ for every
eigenvalue $\mu$ of $B$; for $\|B\|_2 = 1$ a sufficient condition
is $|\bar\rho| + |\epsilon| < 1$.

At $\epsilon = 0$ the process reduces to scalar AR(1) at every
component: $X_t = \bar\rho X_{t-1} + \varepsilon_t$, with
stationary covariance $\Sig_1^{(0)} = \Sigma_\varepsilon /
(1 - \bar\rho^2)$ from the discrete Lyapunov equation. The
$H$-period covariance is
\begin{equation*}
  \SigH^{(0)}
    \;=\; \horizon\,
          \kappastat^{\mathrm{AR}(1)}(\horizon; \bar\rho)\,
          \Sig_1^{(0)}
    \;=\; c(\horizon, \bar\rho)\, \Sig_1^{(0)},
\end{equation*}
where we abbreviate
\begin{equation}
  c(\horizon, \bar\rho)
    \;:=\; \horizon \cdot
           \kappastat^{\mathrm{AR}(1)}(\horizon; \bar\rho).
  \label{eq:c-def}
\end{equation}
At zero order in $\epsilon$, $\SigH^{(0)}$ is proportional to
$\Sig_1^{(0)}$, so they share eigenvectors exactly and
$\Omat^{(0)}(\horizon) = I_N$. The whole eigenvector-rotation
content of $\Omat$ comes from the first-order perturbation.

\subsection{First-order perturbation of the covariance}

The first-order corrections to $\Sig_1$ and $\SigH$ follow from
expansion of the Lyapunov equation and the
trapezoidal-sum representation~\eqref{eq:SigmaH-trapezoid}.

\begin{lemma}[First-order corrections]
  \label{lem:first-order-Sigma}
  Under the vector AR(1) specification~\eqref{eq:vector-AR1} with
  symmetric $B$,
  \begin{align}
    \Sig_1
      &\;=\; \Sig_1^{(0)} + \epsilon \Sig_1^{(1)} + O(\epsilon^2),
      &
    \Sig_1^{(1)}
      &\;=\; \frac{\bar\rho}{1 - \bar\rho^2}
             \bigl(B \Sig_1^{(0)} + \Sig_1^{(0)} B^\top\bigr),
    \\
    \SigH
      &\;=\; \SigH^{(0)} + \epsilon \SigH^{(1)} + O(\epsilon^2),
      &
    \SigH^{(1)}
      &\;=\; c(\horizon, \bar\rho)\, \Sig_1^{(1)}
             \;+\; S_1(\horizon, \bar\rho)\,
                   \bigl(B \Sig_1^{(0)} + \Sig_1^{(0)} B^\top\bigr),
    \nonumber
  \end{align}
  where
  \begin{equation}
    S_1(\horizon, \bar\rho)
      \;:=\; \sum_{h=1}^{\horizon - 1}
              (\horizon - h)\, h\, \bar\rho^{\,h-1}.
    \label{eq:S1-def}
  \end{equation}
\end{lemma}

\begin{proof}[Sketch]
  For $\Sig_1$: substitute $A = \bar\rho I + \epsilon B$ into the
  discrete Lyapunov equation
  $\Sig_1 = A \Sig_1 A^\top + \Sigma_\varepsilon$ and collect
  terms order by order. The order-zero piece gives
  $\Sig_1^{(0)} = \Sigma_\varepsilon/(1 - \bar\rho^2)$; the
  order-one piece $(1 - \bar\rho^2)\Sig_1^{(1)} = \bar\rho
  (B \Sig_1^{(0)} + \Sig_1^{(0)} B^\top)$ yields the stated
  expression for $\Sig_1^{(1)}$. For $\SigH$: expand
  $A^h = \bar\rho^h I + \epsilon h \bar\rho^{h-1} B
  + O(\epsilon^2)$ and substitute into the
  trapezoidal-sum representation~\eqref{eq:SigmaH-trapezoid}.
  At order zero this gives $\SigH^{(0)} = c\, \Sig_1^{(0)}$;
  at order one, two distinct contributions appear --- a
  $c\, \Sig_1^{(1)}$ piece from the order-one expansion of
  $\Sig_1$, and a piece $\sum_{h \geq 1}(\horizon - h) h
  \bar\rho^{h-1} (B \Sig_1^{(0)} + \Sig_1^{(0)} B^\top)$ from
  the order-one expansion of $A^h$, with the sum exactly
  $S_1(\horizon, \bar\rho)$ from~\eqref{eq:S1-def}.
\end{proof}

The combinatorial factor $S_1(\horizon, \bar\rho)$ defined by
\eqref{eq:S1-def} captures the entire horizon dependence of the
first-order correction beyond what is already in
$c(\horizon, \bar\rho) \Sig_1^{(1)}$. The closed form is
\begin{equation}
  S_1(\horizon, \bar\rho)
    \;=\; (\horizon - 1)\, f'(\bar\rho)
         \;-\; \bar\rho\, f''(\bar\rho),
  \qquad
  f(r) \;:=\; \frac{1 - r^\horizon}{1 - r},
  \label{eq:S1-closed}
\end{equation}
which one verifies by direct differentiation. The value at
$\bar\rho = 0$ is $S_1(\horizon, 0) = \horizon - 1$ (only the
$h = 1$ term contributes).

\subsection{First-order eigenvector overlap}
\label{sec:overlap-first-order}

Both $\Sig_1$ and $\SigH$ are now smooth perturbations of
$\Sig_1^{(0)}$ and $\SigH^{(0)} = c \Sig_1^{(0)}$ respectively.
Standard non-degenerate first-order perturbation theory
\citep{Kato1995} gives the first-order corrections to the
eigenvectors. The key cancellation: at first order, the
contribution of $\Sig_1^{(1)}$ to $\eivec_i(\SigH)$ exactly
cancels the contribution to $\eivec_j(\Sig_1)$ in the inner
product $\langle \eivec_i(\SigH), \eivec_j(\Sig_1) \rangle$,
leaving only the contribution from the genuinely new term
$S_1\, (B \Sig_1^{(0)} + \Sig_1^{(0)} B^\top)$.

\begin{theorem}[First-order overlap closed form]
  \label{thm:O-first-order}
  Under the vector AR(1) specification~\eqref{eq:vector-AR1} with
  symmetric $B$ and distinct eigenvalues of $\Sig_1^{(0)}$, the
  first-order eigenvector overlap is
  \begin{equation}
    \boxed{\;
      \Omat_{ij}(\horizon)
        \;=\; \epsilon^2
          \left(\frac{S_1(\horizon, \bar\rho)}{c(\horizon, \bar\rho)}\right)^{\!\!2}
          \frac{(\lambda_i + \lambda_j)^2\, B_{ji}^2}
               {(\lambda_i - \lambda_j)^2}
          + O(\epsilon^4),
      \quad i \neq j,
    \;}
    \label{eq:O-first-order}
  \end{equation}
  where $\lambda_i := \lambda_i(\Sig_1^{(0)})$ are the
  unperturbed eigenvalues and $B_{ij}$ are the matrix elements of
  $B$ in the unperturbed eigenbasis. The diagonal entries are
  $\Omat_{ii}(\horizon) = 1 - \sum_{j \neq i} \Omat_{ij}(\horizon) +
  O(\epsilon^4)$.
\end{theorem}

\begin{proof}[Sketch]
  Lemma~\ref{lem:first-order-Sigma} gives the first-order
  corrections to $\Sig_1$ and $\SigH$. Standard perturbation
  theory of symmetric matrices~\citep{Kato1995} yields
  \begin{align*}
    \langle \eivec_j^{(0)}, \delta \eivec_i(\SigH) \rangle
      &\;=\;
        \frac{\langle \eivec_j^{(0)}, \SigH^{(1)} \eivec_i^{(0)} \rangle}
             {c (\lambda_i - \lambda_j)},
    \\
    \langle \eivec_i^{(0)}, \delta \eivec_j(\Sig_1) \rangle
      &\;=\;
        \frac{\langle \eivec_i^{(0)}, \Sig_1^{(1)} \eivec_j^{(0)} \rangle}
             {\lambda_j - \lambda_i}.
  \end{align*}
  Summing and substituting $\SigH^{(1)} = c \Sig_1^{(1)} +
  S_1 N$ where $N := B \Sig_1^{(0)} + \Sig_1^{(0)} B^\top$, the
  $\Sig_1^{(1)}$ contributions cancel because $\Sig_1^{(1)}$ is
  symmetric. The surviving term is
  $\epsilon (S_1/c) \langle \eivec_j, N \eivec_i \rangle / (\lambda_i - \lambda_j)$.
  For symmetric $B$ in the eigenbasis of $\Sig_1^{(0)}$,
  $\langle \eivec_j, N \eivec_i \rangle = (\lambda_i + \lambda_j) B_{ji}$.
  Squaring gives~\eqref{eq:O-first-order}. Full derivation in
  \cref{app:proofs}.
\end{proof}

The closed form~\eqref{eq:O-first-order} has two structural
features that organise its interpretation.

\paragraph{Eigenvalue-gap dependence (Lorentzian factor).}
The denominator $(\lambda_i - \lambda_j)^2$ produces the
Lorentzian eigenvalue-gap factor familiar from
\citet{BunBouchaudPotters2017} in the random-matrix context.
Eigenvectors corresponding to nearly-degenerate eigenvalues rotate
strongly: small gaps in the denominator amplify the mixing.
Eigenvectors corresponding to well-separated eigenvalues stay
stable. The numerator $(\lambda_i + \lambda_j)^2$ partially
compensates near degeneracies of small absolute size but does not
overpower the denominator when the eigenvalues are at the same
order of magnitude.

\paragraph{Horizon-dependent factor (saturation).}
The factor $\bigl(S_1(\horizon, \bar\rho) / c(\horizon, \bar\rho)
\bigr)^2$ carries the entire horizon dependence of the first-order
overlap. At the daily horizon $\horizon = 1$ the factor is zero
($S_1(1, \bar\rho) = 0$, $c(1, \bar\rho) = 1$) and
$\Omat(1) = I$ trivially. As $\horizon$ grows the factor
increases, reflecting accumulating eigenvector rotation. The
crucial property is that the factor \emph{saturates} at long
horizons:

\begin{proposition}[Saturation]
  \label{prop:O-saturation}
  For $|\bar\rho| < 1$,
  \begin{equation}
    \lim_{\horizon \to \infty}
      \frac{S_1(\horizon, \bar\rho)}{c(\horizon, \bar\rho)}
    \;=\; \frac{1}{1 - \bar\rho^2}.
    \label{eq:S1-over-c-limit}
  \end{equation}
  The convergence rate is $O(1/\horizon \cdot (1 - \bar\rho^2)^{-2})$.
\end{proposition}

\begin{proof}[Sketch]
  Asymptotics of \eqref{eq:S1-closed} and \eqref{eq:c-def}:
  $S_1 \sim \horizon / (1 - \bar\rho)^2$ and
  $c \sim \horizon (1 + \bar\rho)/(1 - \bar\rho)$ at leading order in
  $\horizon$; the ratio is
  $1 / [(1 - \bar\rho)(1 + \bar\rho)] = 1/(1 - \bar\rho^2)$. The
  next-order correction is $O(1/\horizon)$.
\end{proof}

Proposition~\ref{prop:O-saturation} predicts that eigenvector
rotation does not grow without bound at long horizons: the
$(S_1/c)^2$ amplification factor settles into a finite plateau
governed by the baseline autocorrelation $\bar\rho$. Asset-class
universes with weak baseline autocorrelation
($|\bar\rho| \ll 1$) saturate near unity, while universes with
strong autocorrelation saturate at $1/(1 - \bar\rho^2)^2$, which
can be considerably larger.

\subsection{Limits of leading-order theory}
\label{sec:overlap-limits}

The first-order closed form~\eqref{eq:O-first-order} is the
leading-order result of a power series in $\epsilon$, and is
quantitatively valid when the perturbation is small compared to
the eigenvalue gaps:
\begin{equation*}
  \epsilon\,
  \bigl|\langle \eivec_j, N \eivec_i \rangle\bigr|
  \;\ll\;
  |\lambda_i - \lambda_j|.
\end{equation*}
For a universe with a wide eigenvalue spread, the smallest gaps
in the bulk may be small in absolute terms even when the
spectrum has well-defined principal directions, and the
leading-order theory may underestimate the true overlap at deep
ranks. The qualitative gap-dependent pattern remains: large gaps
imply stable eigenvectors, small gaps imply substantial rotation.
The quantitative magnitudes at deep ranks --- where the bulk
spectrum approaches the noise floor of the empirical covariance ---
require either higher-order perturbation theory or a non-perturbative
treatment via the rotationally-invariant-estimator machinery of
\citet{BunBouchaudPotters2017}.

Section~\ref{sec:empirical} returns to this point: the
empirical $\Omat$ pattern on the Fama--French 49-industry universe
is qualitatively consistent with~\eqref{eq:O-first-order}, but the
deep-rank quantitative magnitudes diverge from the leading-order
prediction in the way the validity condition warns.

\section{Multi-memory factor model}
\label{sec:factor-model}

The $(\kappastat, \Omat)$ framework of
Sections~\ref{sec:kappa}--\ref{sec:overlap} characterises departures from
the i.i.d.\ null at the level of cross-sectional eigenmodes. To turn
those characterisations into a structural model of the
return-generating process we need a parametric form for the
underlying factors. This section introduces the five-factor
multi-memory specification used in our empirical work
(Section~\ref{sec:empirical}). The factor profiles are standard
constructions from the long-memory and multifractal-volatility
literatures: fractional Brownian motion
\citep{MandelbrotVanNess1968,BouchaudPotters2003} provides the
persistent and antipersistent components; ARFIMA$(1, d, 0)$
\citep{GrangerJoyeux1980} and the Markov-switching multifractal
cascade of \citet{CalvetFisher2008} provide the multi-scale
components.

\subsection{Two-channel cross-sectional statistic}
\label{ssec:two-channel}

The variance-ratio statistic of Section~\ref{ssec:eigenstructure-stats}
is constructed from the cross-sectional covariance of $H$-period
returns $\Sigi$ and $\SigH$. We extend it to a second channel by
applying the same construction to the matrix of squared returns.
Let $r_t \in \R^N$ denote daily log-returns of $N$ assets and
$r_t^{\odot 2} \in \R^N$ their element-wise squares. Define the
linear- and volatility-channel covariance matrices at horizon $H$ as
\begin{equation}
  \Sigi^{\text{lin}} \;:=\; \Cov(r_t), \qquad
  \SigH^{\text{lin}} \;:=\; \Cov\bigg(\sum_{s=1}^H r_{t-s+1}\bigg),
  \label{eq:Sig-lin}
\end{equation}
\begin{equation}
  \Sigi^{\text{vol}} \;:=\; \Cov(r_t^{\odot 2}), \qquad
  \SigH^{\text{vol}} \;:=\; \Cov\bigg(\sum_{s=1}^H r_{t-s+1}^{\odot 2}\bigg),
  \label{eq:Sig-vol}
\end{equation}
and the per-eigenmode variance ratios
\begin{equation}
  \kappastat^{\text{lin}}_i(H) := \frac{\lambda_i(\SigH^{\text{lin}})}
    {H \cdot \lambda_i(\Sigi^{\text{lin}})}, \qquad
  \kappastat^{\text{vol}}_i(H) := \frac{\lambda_i(\SigH^{\text{vol}})}
    {H \cdot \lambda_i(\Sigi^{\text{vol}})}.
  \label{eq:kappa-two-channel}
\end{equation}
Under i.i.d.\ Gaussian returns both
$\kappastat^{\text{lin}}_i(H) = 1$ and
$\kappastat^{\text{vol}}_i(H) = 1$ for every eigenmode and horizon.
Empirical deviations from these two nulls carry distinct
information: $\kappastat^{\text{lin}}_i \neq 1$ measures temporal
autocorrelation of the $i$-th linear eigenmode;
$\kappastat^{\text{vol}}_i \neq 1$ measures persistence in the
volatility of the $i$-th volatility eigenmode.

The two channels are not equivalent. The eigenvectors of
$\Sigi^{\text{lin}}$ and $\Sigi^{\text{vol}}$ are in general
distinct: a cross-asset structure expressed by loadings
$\beta \in \R^N$ acts on linear covariance through $\beta\beta^\top$
and on volatility covariance through
$(\beta \odot \beta)(\beta \odot \beta)^\top$, and these matrices
have coincident eigenvectors only when $\beta$ has approximately
uniform entries
(Proposition~\ref{prop:linvol-eigvec-coincidence} below).

\subsection{Five-factor multi-memory specification}
\label{ssec:five-factors}

Write the daily return vector as a linear combination of five
factor processes plus an idiosyncratic component,
\begin{equation}
  r_t \;=\; \sum_{k \in \mathcal{K}} \beta_k \, F_{k, t}
            \;+\; \varepsilon_t,
  \qquad \mathcal{K} = \{P, A, M, V, V\!t\},
  \label{eq:factor-decomp}
\end{equation}
where $\beta_k \in \R^N$ is the cross-asset loading vector of
factor $k$ and $\varepsilon_t \sim \mathcal{N}(0, \sigma_\varepsilon^2 I_N)$
is asset-specific noise uncorrelated across time and across the
$\{F_{k, t}\}_{k \in \mathcal{K}}$.

Each factor $F_{k, t}$ is a scalar stochastic process with a
specific multi-memory signature in the linear and volatility
channels. Table~\ref{tab:factor-profiles} summarises the five
specifications. Sections~\ref{sssec:fbm-factors} and
\ref{sssec:multifractal-factors} below give the parametric forms.

\begin{table}[tb]
\centering
\caption{Factor profiles in the linear and volatility channels.
$\kappastat^{\text{lin}}_k(H)$ denotes the scalar variance ratio at
horizon $H$ for factor $k$ in the linear channel; analogously for
the volatility channel.}
\label{tab:factor-profiles}
\begin{tabular}{lll}
\toprule
Factor & Linear-channel profile & Volatility-channel profile \\
\midrule
$F_P$    & fBm-persistent ($H_P > 1/2$)
         & constant ($\kappastat^{\text{vol}}_P = 1$) \\
$F_A$    & fBm-antipersistent ($H_A < 1/2$)
         & constant ($\kappastat^{\text{vol}}_A = 1$) \\
$F_M$    & ARFIMA$(1, d_M, 0)$
         & MSM cascade $(m_0, b, \gamma_{\bar k}, \bar k)$ \\
$F_V$    & i.i.d.\ ($\kappastat^{\text{lin}}_V = 1$)
         & MSM cascade (shared with $F_M$) \\
$F_{Vt}$ & i.i.d.\ ($\kappastat^{\text{lin}}_{Vt} = 1$)
         & ARFIMA$(1, d_{Vt}, 0)$ on $\sigma_{Vt, t}^2$ \\
\bottomrule
\end{tabular}
\end{table}

\subsubsection{Persistent and antipersistent factors}
\label{sssec:fbm-factors}

The persistent factor $F_P$ is a discretely sampled fractional
Brownian motion with Hurst exponent $H_P > 1/2$
\citep{MandelbrotVanNess1968}. Its per-eigenmode variance ratio
follows the well-known scaling law
\begin{equation}
  \kappastat^{\text{lin}}_P(H) \;=\; H^{2 H_P - 1},
  \label{eq:kappa-FP}
\end{equation}
which is increasing in $H$ and recovers the random-walk limit
$\kappastat^{\text{lin}}_P = 1$ at $H_P = 1/2$. The antipersistent
factor $F_A$ has the same functional form with $H_A < 1/2$, hence
$\kappastat^{\text{lin}}_A(H) = H^{2 H_A - 1}$ is decreasing in $H$.
Neither factor carries non-trivial volatility structure beyond a
homoscedastic Gaussian innovation, so
$\kappastat^{\text{vol}}_P = \kappastat^{\text{vol}}_A = 1$ at all
horizons.

\subsubsection{Multifractal factors $F_M$, $F_V$, $F_{Vt}$}
\label{sssec:multifractal-factors}

The remaining three factors carry the multi-scale temporal
structure of the model. $F_M$ is the central factor: it has a
non-trivial linear-channel profile (ARFIMA) and a non-trivial
volatility-channel profile (MSM cascade). $F_V$ shares $F_M$'s
volatility cascade but contributes no linear-channel structure
beyond Gaussian noise; structurally $F_V$ is a pure-volatility
companion to $F_M$.\footnote{The decomposition $F_M + F_V$ is not
identifiable from the volatility-channel covariance alone, because
both contribute the same MSM profile. They are separated by the
linear channel: only $F_M$ contributes to
$\kappastat^{\text{lin}}_i$. Identification at the cross-asset
level is discussed in Section~\ref{ssec:joint-LS-fit}.}
$F_{Vt}$ is a transitory-volatility factor: its linear channel is
i.i.d.\ noise but its volatility-of-volatility process $\sigma_{Vt, t}^2$
follows ARFIMA$(1, d_{Vt}, 0)$, providing a finite-horizon
volatility cluster that decays at long horizons.

\paragraph{ARFIMA$(1, d, 0)$.} An autoregressive
fractionally-integrated process with one AR lag and fractional
differencing parameter $d$ has scalar variance ratio
\begin{equation}
  \kappastat^{\text{ARFIMA}}(H; \phi, d) \;=\; \frac{1}{H}
    \sum_{|\ell| < H} (H - |\ell|) \, \gamma_\ell(\phi, d) / \gamma_0(\phi, d),
  \label{eq:kappa-arfima}
\end{equation}
where $\gamma_\ell(\phi, d)$ is the autocovariance at lag $\ell$ of
the ARFIMA$(1, d, 0)$ process. The closed form for
$\gamma_\ell(\phi, d)$ via the spectral density and inverse FFT is
standard \citep{BeranFengGhoshKulik2013}. For
$d \in (-1/2, 1/2)$ and $|\phi| < 1$ the process is stationary;
$d > 0$ gives long-memory persistence and $d < 0$ gives
antipersistent long memory. We use $F_M$'s linear profile in the
ARFIMA family because the empirical
$\kappastat^{\text{lin}}_1(H)$ at the market mode displays a
characteristic rise-then-fall pattern that no single fBm can
generate (Section~\ref{sec:empirical}).

\paragraph{MSM cascade.} The Markov-switching multifractal of
\citet{CalvetFisher2004,CalvetFisher2008} is the multi-scale
volatility model
\begin{equation}
  \sigma_t^2 \;=\; \bar\sigma^2
    \prod_{k=1}^{\bar k} M_{k, t},
  \label{eq:msm-spec}
\end{equation}
where $\{M_{k, t}\}_{k = 1}^{\bar k}$ are independent
multiplicative components, each a two-state Markov chain on
$\{m_0, 2 - m_0\}$ with switching probability
\begin{equation}
  \gamma_k \;=\; 1 - (1 - \gamma_1)^{b^{k - 1}},
  \qquad b > 1,
  \label{eq:msm-switching}
\end{equation}
indexed by frequency $k$. The lowest-frequency component
$M_{1, t}$ has switching probability $\gamma_1 \ll 1$ and duration
$O(1 / \gamma_1)$ trading days; the highest-frequency component
$M_{\bar k, t}$ has switching probability
$\gamma_{\bar k} \approx 1$ and switches near-daily. The parameter
vector is $\psi^{\mathrm{MSM}} = (m_0, b, \gamma_{\bar k}, \bar k)
\in \R_+^4 \times \mathbb{N}$ with $\gamma_1$ implicit through
\eqref{eq:msm-switching}. Following the recommendation of
\citet[\S3.3]{CalvetFisher2008} we fix $\bar k = 8$ throughout.
The MSM variance ratio
$\kappastat^{\mathrm{MSM}}(H; m_0, b, \gamma_{\bar k}, \bar k)$
admits a closed-form expression in terms of the multipliers'
autocorrelations (Proposition~1 of
\citet{CalvetFisher2008}); we use it directly in the joint LS
fit of Section~\ref{ssec:joint-LS-fit}.

\subsection{Eigenmode-level model predictions}
\label{ssec:per-mode-predictions}

Given the factor decomposition \eqref{eq:factor-decomp} the
$H$-horizon linear-channel covariance matrix is
\begin{equation}
  \SigH^{\text{lin}} \;=\; \sum_{k \in \mathcal{K}}
    \nu_k^{\text{lin}}(H) \cdot \beta_k \beta_k^\top
    \;+\; H \cdot \sigma_\varepsilon^2 I_N,
  \label{eq:Sig-H-lin-decomp}
\end{equation}
where $\nu_k^{\text{lin}}(H) := H \cdot \Var(F_{k, t})
\cdot \kappastat^{\text{lin}}_k(H)$ is factor $k$'s contribution to
linear variance at horizon $H$. Analogously, the
volatility-channel covariance is
\begin{equation}
  \SigH^{\text{vol}} \;=\; \sum_{k \in \mathcal{K}}
    \nu_k^{\text{vol}}(H) \cdot
    (\beta_k \odot \beta_k)(\beta_k \odot \beta_k)^\top
    \;+\; H \cdot \sigma_{\varepsilon,2}^2 I_N,
  \label{eq:Sig-H-vol-decomp}
\end{equation}
where $\sigma_{\varepsilon, 2}^2$ is the variance of
$\varepsilon_t^{\odot 2}$ and
$\nu_k^{\text{vol}}(H) := H \cdot \Var(F_{k, t}^2)
\cdot \kappastat^{\text{vol}}_k(H)$ is factor $k$'s
contribution to volatility-channel variance at horizon $H$
(the direct analogue of $\nu_k^{\text{lin}}(H)$, with the
squared factor $F_{k, t}^2$ in place of $F_{k, t}$ to match
the squared-return covariance object on the left-hand side).

A direct consequence of \eqref{eq:Sig-H-lin-decomp} and
\eqref{eq:Sig-H-vol-decomp} is the eigenvector-coincidence result
that motivates our cross-channel analysis:

\begin{proposition}[Linear/volatility eigenvector coincidence]
\label{prop:linvol-eigvec-coincidence}
Let $\{\eivec_i^{\text{lin}}\}_{i = 1}^N$ and
$\{\eivec_i^{\text{vol}}\}_{i = 1}^N$ denote the eigenvectors of
$\Sigi^{\text{lin}}$ and $\Sigi^{\text{vol}}$ respectively. Under
\eqref{eq:factor-decomp} with a single dominant factor
$k^* \in \mathcal{K}$, the two eigenbases coincide if and only if
the cross-asset loading $\beta_{k^*}$ has approximately uniform
entries.
\end{proposition}

\begin{proof}
Sketch (full proof in
Appendix~\ref{app:linvol-eigvec-coincidence}, supplementary
material). The dominant terms of $\Sigi^{\text{lin}}$ and
$\Sigi^{\text{vol}}$ are
$\nu_{k^*}^{\text{lin}}(1) \beta_{k^*} \beta_{k^*}^\top$ and
$\nu_{k^*}^{\text{vol}}(1) (\beta_{k^*} \odot \beta_{k^*})
(\beta_{k^*} \odot \beta_{k^*})^\top$. Their leading eigenvectors
are proportional to $\beta_{k^*}$ and $\beta_{k^*} \odot \beta_{k^*}$,
which are colinear if and only if all entries of $\beta_{k^*}$ are
equal in absolute value.
\end{proof}

Proposition~\ref{prop:linvol-eigvec-coincidence} establishes that
the linear and volatility channels diagonalise distinct matrices.
Empirically the market mode (rank $i = 1$) is the unique eigenmode
where $\beta_k$ has near-uniform entries (its eigenvector is the
$1/\sqrt{N}$ vector to leading order), and accordingly empirical
overlap $|\langle \eivec_1^{\text{lin}}(H), \eivec_1^{\text{vol}}(H) \rangle|^2$
is close to one at short horizons and decays only modestly across
horizons. For sub-leading eigenmodes the cross-channel overlap is
small, even at $H = 1$.

The per-eigenmode model prediction at horizon $H$, conditional on
the empirical eigenvectors $\eivec_i^{\text{lin}}$ at $H = 1$, takes
the form
\begin{equation}
  \kappastat^{\text{lin}}_i(H) \;=\;
    \sum_{k \in \mathcal{K}} w^{\text{lin}}_{i, k}
      \cdot \kappastat^{\text{lin}}_k(H)
    \;+\; w^{\text{lin}}_{i, \varepsilon},
  \label{eq:kappa-lin-prediction}
\end{equation}
\begin{equation}
  \kappastat^{\text{vol}}_i(H) \;=\;
    \sum_{k \in \mathcal{K}} w^{\text{vol}}_{i, k}
      \cdot \kappastat^{\text{vol}}_k(H)
    \;+\; w^{\text{vol}}_{i, \varepsilon},
  \label{eq:kappa-vol-prediction}
\end{equation}
where the per-mode weights $w^{\text{lin}}_{i, k} \geq 0$ and
$w^{\text{vol}}_{i, k} \geq 0$ are constrained to sum to one
across factors (including the idiosyncratic $\varepsilon$ term)
within each channel and each eigenmode. They are
proportional to the squared inner products
$|\langle \eivec_i^{\text{lin}}, \beta_k \rangle|^2$ and
$|\langle \eivec_i^{\text{vol}}, \beta_k \odot \beta_k \rangle|^2$
respectively, normalised so that the variance attributed to each
mode sums correctly.

\subsection{Calvet--Fisher structural anchor}
\label{ssec:cf-anchor}

The multi-memory factor model \eqref{eq:factor-decomp} is not
free-standing. It maps onto the equilibrium asset-pricing model
of \citet[Ch.~9]{CalvetFisher2008}, in which an
Epstein--Zin--Weil representative agent prices a Lucas tree whose
log-dividend follows MSM-volatility dynamics. The
price-dividend ratio $Q(M_t)$ solves a fixed-point Euler
equation, and the equilibrium excess return decomposes as
\citep[eq.~9.7]{CalvetFisher2008}
\begin{equation}
  r_{t+1} \;=\; \ln\frac{1 + Q(M_{t+1})}{Q(M_t)}
    \;+\; \bar\mu_d - r_f
    \;-\; \tfrac{1}{2} \sigma_d^2(M_{t+1})
    \;+\; \sigma_d(M_{t+1}) \varepsilon_{d, t+1}.
  \label{eq:cf-equilibrium-return}
\end{equation}
The first term is volatility feedback: a positive innovation in
the multifractal state $M_{t+1}$ raises expected future volatility,
lowers the discount factor $Q$, and reduces the realised return.
\citet[\S9.2.2]{CalvetFisher2008} show that the magnitude of this
feedback from a shift in component $M_{k, t+1}$ is approximately
proportional to $1/\gamma_k$, the inverse of the component's
switching probability. Lower-frequency components produce larger
return-side feedback; higher-frequency components essentially
none.

This is the structural mechanism behind our $F_M$ factor's
non-trivial linear-channel profile. The ARFIMA$(1, d_M, 0)$ form
in our linear-channel prediction
\eqref{eq:kappa-lin-prediction} is a flexible reduced-form
parametrisation of the inertial range of the
volatility-feedback filter; the MSM cascade in
\eqref{eq:kappa-vol-prediction} is the underlying volatility-
generating mechanism. Empirically CF08 \S9.3 estimates
unconditional feedback $\Var(r) / \Var(\Delta d) - 1 \in [0.20,
0.50]$ on U.S.\ aggregate equity returns at $\bar k \in \{6, 7,
8\}$. Our eigenstructure-level analysis inherits this feedback
mechanism through the cross-asset projection of $F_M$ onto the
linear-channel eigenvectors.

\subsection{Joint LS fit and identification}
\label{ssec:joint-LS-fit}

We estimate the model by joint least-squares matching of the
empirical $(\kappastat^{\text{lin}}, \kappastat^{\text{vol}})$
matrices to their model predictions
\eqref{eq:kappa-lin-prediction}--\eqref{eq:kappa-vol-prediction}.
The objective, defined on the parameter vector
\[
\theta = \bigl(H_P,\, H_A,\, \phi_M,\, d_M,\, m_0,\, b,\,
\gamma_{\bar k},\, \phi_{Vt},\, d_{Vt},\,
\{w_{i, k}^{\text{lin}}\},\, \{w_{i, k}^{\text{vol}}\}\bigr),
\]
reads
\begin{equation}
  \mathcal{L}(\theta) \;=\;
    \sum_{i, H} \frac{[\kappastat^{\text{lin}}_i(H; \theta)
        - \widehat{\kappastat}^{\text{lin}}_i(H)]^2}
      {\max(1, |\widehat{\kappastat}^{\text{lin}}_i(H)|)^2}
    \;+\;
    \sum_{i, H} \frac{[\kappastat^{\text{vol}}_i(H; \theta)
        - \widehat{\kappastat}^{\text{vol}}_i(H)]^2}
      {\max(1, |\widehat{\kappastat}^{\text{vol}}_i(H)|)^2},
  \label{eq:joint-ls-objective}
\end{equation}
plus quadratic penalties enforcing the per-mode sum-to-one
weight constraints. The per-(mode, horizon) residual scaling in
the denominator makes each entry contribute an O(1) magnitude to
the loss regardless of the absolute $\kappastat$ scale, which is
essential because $\kappastat^{\text{vol}}_1(H)$ values at long
horizons exceed sub-leading-mode values by two orders of
magnitude.

The fit uses bootstrap-resampled $\widehat{\kappastat}$ matrices
(moving-block resampling with block size
$2 H_{\max}$) and multi-start L-BFGS-B with warm-restart across
three passes per replicate. The nine global parameters
$(H_P, H_A, \phi_M, d_M, m_0, b, \gamma_{\bar k}, \phi_{Vt},
d_{Vt})$ are identified by the $\widehat{\kappastat}$ moment with
finite bootstrap dispersion: across 1000 replicates the
fractional-Brownian exponents and the ARFIMA parameters carry
modest credible intervals, while the MSM parameters
$(m_0, b, \gamma_{\bar k})$ are identified less sharply ---
$\gamma_{\bar k}$ the widest. The empirical credible intervals,
and a bound-limited case ($d_M$), are reported in
Section~\ref{ssec:global-params}. Per-mode weights vary across
replicates, as expected from the resampled $\widehat{\kappastat}$.
Implementation details and a Hessian-based identifiability
diagnostic are provided in supplementary
material~\ref{app:identifiability}.

\subsection{Connection to existing parametric models}
\label{ssec:related-models}

The multi-memory factor specification \eqref{eq:factor-decomp}
nests several familiar models. With only $F_P + \varepsilon$ the
linear-channel statistic reduces to the per-eigenmode AR(1)
form of Section~\ref{sec:kappa-ar1}, with $\rho =
\rho(H_P)$ implicitly determined by the Hurst exponent. With only
$F_V + \varepsilon$ the volatility-channel statistic reduces to
the univariate MSM of \citet[Ch.~3]{CalvetFisher2008}. The full
five-factor specification is more general than any of these
nested cases because it admits separate cross-asset loadings
$\beta_k$ for each factor's contribution, which our empirical
work shows is necessary to capture the stylised facts at the
eigenstructure level (Section~\ref{sec:empirical}).

\section{Empirical results}
\label{sec:empirical}

We estimate the multi-memory factor model
\eqref{eq:factor-decomp} on four equity datasets, each via 1000
moving-block bootstrap replicates of the joint
$(\widehat{\kappastat}^{\text{lin}}, \widehat{\kappastat}^{\text{vol}})$
matrices.\footnote{The bootstrap pipeline, the code that produced
the results in this section, and a complete reproducibility
manifest are listed in supplementary
material~\ref{app:reproducibility}.} The four primary datasets
are chosen to test the model along two robustness axes: temporal
stability (full sample vs.\ split halves), and cross-sectional
grouping (industry sort vs.\ size$\times$book-to-market sort).
A fifth panel constructed from the \citet{KennethFrenchDataLibrary}
Europe 25 size$\times$book-to-market sort provides a cross-region
replication and is reported in the robustness subsection
(Section~\ref{ssec:robustness}).

\subsection{Data and methodology}
\label{ssec:data-methodology}

\textbf{Datasets.} All four panels are constructed from
\citet{KennethFrenchDataLibrary}-distributed daily value-weighted
portfolios. The \emph{sensitivity} panel (FF 49 industries,
1969-07-01 to 2026-03-31, $T = 14{,}309$ trading days, $N = 48$
industries after dropping any column with missing values) is the
default cross-section in the literature on industry-level
long-horizon dynamics. The \emph{firsthalf}
($T = 7{,}205$, 1969-07-01 to 1997-12-31, $N = 48$) and
\emph{secondhalf} ($T = 7{,}104$, 1998-01-02 to 2026-03-31,
$N = 49$, with the \texttt{Softw} industry recovered as its
pre-1998 missing entries no longer apply) split the sensitivity
panel into pre- and post-1998 sub-samples. The \emph{FF 100} panel
(FF 100 size$\times$book-to-market portfolios in a $10 \times 10$
sort, $1969-07-01$ to 2026-03-31, $T = 14{,}309$, $N = 95$
after column-wise drop of any portfolio with missing
observations) tests robustness to a completely different
cross-sectional grouping.

\textbf{Empirical statistics.} For each panel we compute the
empirical $\kappastat^{\text{lin}}_i(H)$ and
$\kappastat^{\text{vol}}_i(H)$ matrices at horizons
$H \in \{1, 5, 21, 63, 252, 1260\}$ trading days
(daily through 5-year aggregation). Eigenmodes are indexed by
the eigenvalues of $\Sigi^{\text{lin}}$ and $\Sigi^{\text{vol}}$
in descending order.

\textbf{Eigenmode signal and the random-matrix noise floor.}
The $\kappastat$ and $\Omat$ statistics are computed for all $N$
eigenmodes; the canonical ranks displayed in the tables below
($1$, $2$, $3$, $6$, $10$, $24$, $31$, $48$) are representative
points spanning the spectrum. How much of the eigenmode structure is genuine
cross-sectional signal rather than estimation noise is the
question addressed by the random-matrix-theory analysis of
financial covariance matrices \citep{LalouxEtAl1999,
PlerouEtAl2002}: for a covariance estimated from $L$ observations
on $N$ assets, the eigenvalues of a pure-noise Wishart matrix
fall, at ratio $Q = L / N$, within the Marchenko--Pastur band
$\lambda_\pm = 1 + Q^{-1} \pm 2\, Q^{-1/2}$, and only eigenvalues
outside that band carry genuine correlation structure. At the
daily horizon the sensitivity panel has $Q = T / N \approx 298$,
so the noise band is narrow --- of width $4\, Q^{-1/2} \approx
0.23$ around unity --- and essentially the entire eigenvalue
spread of the daily covariance is genuine structure; the
per-eigenmode decomposition is well-posed at short horizons. At
long horizons the picture changes: although the $H$-period
covariance $\Sigma_H$ is computed from overlapping returns and is
numerically full-rank, the number of effectively independent
$H$-period observations is only $\lfloor T / H \rfloor$ --- eleven
at the five-year horizon --- so the effective ratio
$Q_H = \lfloor T / H \rfloor / N$ falls well below one, the
random-matrix noise band widens, and the deep-rank eigenvalues of
$\Sigma_H$ carry a growing share of estimation noise. This is the
random-matrix reading of the leading-order overlap theory's
breakdown at deep ranks (Section~\ref{sec:overlap-limits}): in the
tables below the market mode and sub-leading modes are robust
signal at all horizons, while the deep-rank long-horizon entries
are noise-floor-limited.

\textbf{Bootstrap procedure.} We resample the daily-return time
series via the moving-block bootstrap with block size $2 H_{\max} =
2520$ trading days, large enough to preserve the
long-range autocorrelation structure within each block. For each
replicate $r \in \{1, \ldots, 1000\}$ we (i) draw blocks with
replacement to form a resampled return series, (ii) compute its
$(\kappastat^{\text{lin}}, \kappastat^{\text{vol}})$ matrices,
and (iii) fit the multi-memory factor model
\eqref{eq:factor-decomp} via the least-squares procedure of
Section~\ref{ssec:joint-LS-fit}. The fit is L-BFGS-B from each of
nine basin-aware starting points, run for three warm-restart
passes with the previous pass's best per-mode weights, and we
report the lowest-loss fit across the resulting $27$ candidate
local minima per replicate.

\textbf{Loss-conditional credible intervals.} The objective
$\mathcal{L}(\theta)$ in \eqref{eq:joint-ls-objective} is
non-convex; multi-start with $n_{\text{starts}} = 9$ and three
warm-restart passes lands the best-of-$27$ fit in the same basin
across the bulk of bootstrap resamples on the sensitivity,
secondhalf, and FF~100 panels (rank-1 vol-channel allocations
concentrated in a $\sim 200/1000$ to $\sim 990/1000$ majority
cluster), with a heavier-tailed minority cluster at lower MSM
allocations whose loss medians are within $1.5$--$2.8$ loss units
of the main cluster (Mann--Whitney $p$-values from $0.006$ on
secondhalf to $0.07$ on FF~100). The firsthalf panel is the
exception: the rank-1 volatility-weight bootstrap distribution
is unimodal but heavy-tailed, with the $90\%$ unfiltered CI
spanning $[0.10, 1.00]$, reflecting genuine weak identification
of the firsthalf vol channel on a sample that straddles a
regime transition (Sections~\ref{ssec:regime-change}--%
\ref{ssec:rolling-window} and Appendix~\ref{app:identifiability}).

Throughout Sections~\ref{ssec:global-params}--%
\ref{ssec:beta-inversion} we report medians and $90\%$ credible
intervals over the loss-conditional subset of replicates with
loss below the panel median ($500$ per panel). This applies a
uniform median threshold across all four panels with no
panel-specific tuning; the filter trims the higher-loss tail of
each bootstrap and leaves the headline medians essentially
unchanged. Appendix~\ref{app:identifiability} discusses the
filter's effect in detail.

\subsection{Global parameters and loss distribution}
\label{ssec:global-params}

Table~\ref{tab:global-params} summarises the fitted global
parameters across the four datasets. Three features of the table
are immediate:

\begin{table}[tb]
\centering
\caption{Bootstrap medians ($90\%$ credible intervals in brackets)
of the global parameters of the multi-memory factor model
across the four datasets, $1000$ replicates each. Statistics
report the loss-conditional bootstrap distribution (replicates
with loss below the panel median; $500$ replicates per panel),
removing optimiser basin-stuck replicates per the
methodology of Section~\ref{ssec:data-methodology}. All nine
global parameters carry finite bootstrap credible intervals; the
fractional-Brownian exponents $(H_P, H_A)$ and the ARFIMA
parameters $(\phi_M, d_M, \phi_{Vt}, d_{Vt})$ are the most
sharply identified, the MSM parameters
$(m_0, b, \gamma_{\bar k})$ the least. The antipersistence
parameter $d_M$ rests against its lower bound $-0.49$ on the
firsthalf, secondhalf and FF~100 panels --- the credible
interval is bound-limited there. Loss is the value of the
per-(mode, horizon)-scaled LS objective
\eqref{eq:joint-ls-objective}.}
\label{tab:global-params}
\resizebox{\textwidth}{!}{%
\begin{tabular}{lcccc}
\toprule
Parameter & Sensitivity & Firsthalf & Secondhalf & FF 100 \\
\midrule
Loss median   & 19.3 & 24.1 & 18.6 & 47.7 \\
Loss CI       & [16.8, 20.4] & [21.5, 25.0] & [14.8, 20.3] & [40.7, 51.4] \\
$H_P$         & 0.552 [0.53, 0.57] & 0.567 [0.51, 0.61] & 0.526 [0.51, 0.57] & 0.521 [0.51, 0.56] \\
$H_A$         & 0.233 [0.14, 0.30] & 0.242 [0.14, 0.29] & 0.265 [0.20, 0.30] & 0.175 [0.06, 0.27] \\
$\phi_M$      & 0.825 [0.78, 0.91] & 0.878 [0.81, 0.90] & 0.883 [0.84, 0.90] & 0.861 [0.82, 0.90] \\
$d_M$         & $-0.324$ [$-0.49$, $-0.27$] & $-0.483$ [$-0.49$, $-0.36$] & $-0.481$ [$-0.49$, $-0.39$] & $-0.485$ [$-0.49$, $-0.44$] \\
$m_0$         & 1.41 [1.36, 1.95] & 1.61 [1.43, 1.81] & 1.62 [1.44, 1.95] & 1.75 [1.50, 1.95] \\
$b$           & 2.11 [1.65, 2.52] & 2.60 [2.19, 2.81] & 2.58 [2.05, 2.77] & 2.35 [1.57, 2.74] \\
$\gamma_{\bar k}$ & 0.46 [0.06, 0.75] & 0.79 [0.59, 0.93] & 0.58 [0.10, 0.72] & 0.63 [0.11, 0.81] \\
$\phi_{Vt}$   & 0.347 [0.30, 0.40] & 0.431 [0.40, 0.46] & 0.312 [0.28, 0.36] & 0.345 [0.27, 0.46] \\
$d_{Vt}$      & $-0.175$ [$-0.21$, $-0.16$] & $-0.228$ [$-0.26$, $-0.21$] & $-0.233$ [$-0.28$, $-0.21$] & $-0.245$ [$-0.34$, $-0.21$] \\
\bottomrule
\end{tabular}%
}
\end{table}

\textbf{(i) The global parameters are identified with finite
bootstrap dispersion.} All nine global parameters are resolved
by the joint $\widehat{\kappastat}$ moment on every resampled
matrix. The fractional-Brownian exponents $(H_P, H_A)$ and the
ARFIMA parameters $(\phi_M, d_M, \phi_{Vt}, d_{Vt})$ are the most
sharply identified, with narrow bootstrap credible intervals on
all four panels; the MSM parameters $(m_0, b, \gamma_{\bar k})$
are identified less tightly, $\gamma_{\bar k}$ most loosely. The
antipersistence parameter $d_M$ is a bound-limited case: its
bootstrap distribution rests against the lower bound $-0.49$ on
the firsthalf, secondhalf and FF 100 panels, indicating that the
data favour an $F_M$ memory at the antipersistent edge of the
stationary ARFIMA range.

\textbf{(ii) The factor-profile parameters are stable across
datasets.} $H_P \approx 0.52$--$0.57$ across all four panels, and
$H_A \approx 0.25$ on the three FF 49 panels (lower, $\approx
0.17$, on the FF 100 size$\times$book-to-market sort). The
fractional-Brownian-motion exponents that drive the persistent
and antipersistent factor profiles take nearly the same values
regardless of which equity universe or sub-period the model is
fitted to, and the ARFIMA persistence parameters are likewise
close across panels ($\phi_M \approx 0.85$--$0.88$). The MSM
cascade parameters $(m_0, b, \gamma_{\bar k})$ show more
variation, mostly through $\gamma_{\bar k}$, the MSM parameter
with the widest bootstrap CI.

\textbf{(iii) The fit quality is comparable across the FF 49
panels and higher on FF 100.} The full-sample sensitivity panel
and the secondhalf panel attain essentially the same loss
($19.3$ and $18.6$); the firsthalf panel is somewhat harder to
fit ($24.1$); the FF 100 panel has the highest loss ($47.7$).
The shorter split-half panels do not carry a uniformly higher
loss than the full sample --- the secondhalf in fact fits as
well as the full sample --- so the loss scale is not a simple
function of sample length. The FF 100 loss tracks its larger
residual parameter count ($5 N = 475$ per-mode weights versus
$240$ on the industry sorts) and its finer cross-sectional sort.

\subsection{The seven stylised facts of long-horizon equity
dynamics}
\label{ssec:stylised-facts}

The empirical $\widehat{\kappastat}^{\text{lin}}$ and
$\widehat{\kappastat}^{\text{vol}}$ matrices on the sensitivity
panel are shown in Figure~\ref{fig:kappa-heatmaps}. The
volatility-channel rank-1 cascade signature
($\kappa^{\text{vol}}_1$ climbing from $\approx 1$ at the daily
horizon to $\approx 50$ at the 5-year horizon) is visible at a
glance, alongside the deep-rank overshoot to values well below
unity at long horizons. The linear-channel panel shows the
much milder rank-1 rise-then-fall and the steady deep-rank
decay.

\begin{figure}[tb]
\centering
\includegraphics[width=0.95\linewidth]{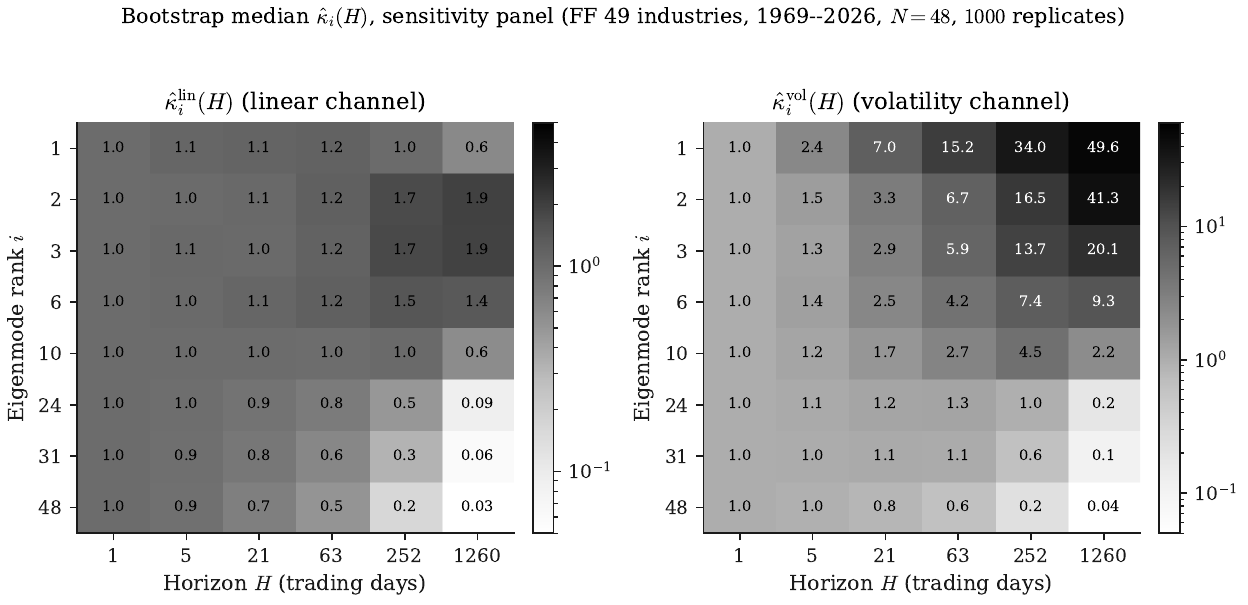}
\caption{Bootstrap medians of
$\widehat{\kappastat}^{\text{lin}}_i(H)$ (left) and
$\widehat{\kappastat}^{\text{vol}}_i(H)$ (right) at canonical
eigenmode ranks and horizons, sensitivity panel ($N = 48$,
$1000$ moving-block bootstrap replicates). Cells annotated with
the median value; colour scale is logarithmic. The i.i.d.\ null
is $\kappastat = 1$ at every $(i, H)$ entry.}
\label{fig:kappa-heatmaps}
\end{figure}

The multi-memory factor model recovers seven stylised facts of
long-horizon equity dynamics. The first six are read directly off
the per-mode weight pattern of Table~\ref{tab:rank-weights}; the
seventh is a property of the volatility eigenvalue spectrum. They
are:

\begin{enumerate}[itemsep=3pt]
\item \emph{Market-mode variance-ratio anomaly.} The rank-1
      linear-channel variance ratio $\kappastat^{\text{lin}}_1(H)$
      rises through intermediate horizons and falls at long
      horizons \citep{LoMacKinlay1988, PoterbaSummers1988}.
\item \emph{Factor momentum at sub-leading eigenmodes.} The
      sub-leading principal directions (ranks 2--3) carry
      monotone-rising $\kappastat^{\text{lin}}_i(H)$
      \citep{AsnessEtAl2013, EhsaniLinnainmaa2022}.
\item \emph{Long-horizon mean reversion at deep eigenmodes.} The
      deep eigenmodes (ranks 24--48) carry steeply decreasing
      $\kappastat^{\text{lin}}_i(H)$
      \citep{Cochrane1988, FamaFrench1988}.
\item \emph{Short-range volatility clustering.} The
      volatility-channel variance ratio rises sharply within the
      first trading month \citep{Cont2001}.
\item \emph{Multi-scale volatility long memory.} The market-mode
      volatility ratio $\kappastat^{\text{vol}}_1(H)$ climbs to
      $\approx 50$ by the five-year horizon --- the multifractal
      cascade signature \citep{CalvetFisher2008}.
\item \emph{Transitory volatility at deep eigenmodes.} The deep
      eigenmodes carry an episodic-burst transitory volatility
      component --- volatility that clusters over finite horizons
      and then decays \citep{Cont2001}, in contrast to the
      multi-scale long memory of fact 5 --- captured in the model
      by the factor $F_{Vt}$.
\item \emph{Cross-sectional concentration of volatility.} The
      volatility cross-section concentrates onto a single
      dominant eigendirection as the horizon grows, while the
      linear cross-section broadens. The dominance of a
      market-wide common eigenvalue in financial covariance
      matrices is a central finding of random-matrix theory
      \citep{LalouxEtAl1999, PlerouEtAl2002}; its horizon
      dependence is documented at the end of this subsection.
\end{enumerate}

The per-mode weights $\{w^{\text{lin}}_{i, k},
w^{\text{vol}}_{i, k}\}$ in
\eqref{eq:kappa-lin-prediction}--\eqref{eq:kappa-vol-prediction}
allocate each eigenmode's $\kappastat$ profile across the five
factor types. Table~\ref{tab:rank-weights} reports their bootstrap
medians at canonical eigenmode ranks for the sensitivity panel;
the supplementary material gives the analogous tables for the
other three datasets and the basin-conditional analysis.

\begin{table}[H]
\centering
\caption{Loss-filtered bootstrap medians of the per-mode weights
at canonical ranks, sensitivity panel (FF 49 industries,
$1969$--$2026$, $N = 48$, $500$ loss-filtered replicates).
Several of the stylised facts are
visible directly in the weight pattern: factor-momentum
persistence concentrated at sub-leading modes
($w^{\text{lin}}_P \approx 0.70$ at ranks 2--3); deep-mode mean
reversion ($w^{\text{lin}}_A \approx 0.79$ at rank 48);
volatility long memory loading on the market mode
($w^{\text{vol}}_{\text{MSM}} = 1.00$ at rank 1); transitory
volatility concentrating onto the deepest modes
($w^{\text{vol}}_{Vt} \to 1.00$ at rank 48).}
\label{tab:rank-weights}
\begin{tabular}{cccccc}
\toprule
Rank & $w^{\text{lin}}_P$ & $w^{\text{lin}}_A$ & $w^{\text{lin}}_M$ &
$w^{\text{vol}}_{\text{MSM}}$ & $w^{\text{vol}}_{Vt}$ \\
\midrule
 1 & 0.16 & 0.31 & 0.19 & 1.00 & 0.00 \\
 2 & 0.72 & 0.11 & 0.00 & 1.00 & 0.00 \\
 3 & 0.69 & 0.11 & 0.00 & 0.82 & 0.00 \\
 6 & 0.49 & 0.12 & 0.00 & 0.47 & 0.00 \\
10 & 0.15 & 0.35 & 0.13 & 0.10 & 0.00 \\
24 & 0.00 & 0.66 & 0.24 & 0.00 & 0.58 \\
31 & 0.00 & 0.74 & 0.24 & 0.00 & 0.80 \\
48 & 0.00 & 0.79 & 0.20 & 0.00 & 1.00 \\
\bottomrule
\end{tabular}
\end{table}

Reading the table along each row exposes the structural identity
of each eigenmode:

\begin{itemize}[itemsep=2pt]
\item \emph{Rank 1 (market mode):}
      $w^{\text{lin}}_A = 0.31$, $w^{\text{lin}}_M = 0.18$,
      $w^{\text{lin}}_P = 0.16$ --- antipersistent and ARFIMA
      components dominate the linear channel, with a smaller
      persistent contribution. The
      composite linear-channel profile is the
      \citet{LoMacKinlay1988} / \citet{PoterbaSummers1988}
      variance-ratio anomaly: $\kappastat^{\text{lin}}_1(H)$
      rises through intermediate horizons (the ARFIMA component)
      and falls at long horizons (the antipersistent
      component). The volatility-channel pattern at rank 1 is
      pure MSM cascade ($w^{\text{vol}}_{\text{MSM}} = 1.00$),
      with no transitory contribution.
\item \emph{Ranks 2--3 (sub-leading principal directions):}
      $w^{\text{lin}}_P \approx 0.70$ dominant, with
      $w^{\text{lin}}_M = 0.00$. The composite linear-channel
      profile is monotone-rising $\kappastat^{\text{lin}}_i(H)$,
      consistent with the factor-momentum literature
      \citep{AsnessEtAl2013,EhsaniLinnainmaa2022}. The
      volatility channel retains substantial MSM loading
      ($w^{\text{vol}}_{\text{MSM}} = 1.00, 0.82$) --- the
      sub-leading principal directions inherit the market mode's
      cascade structure.
\item \emph{Rank 10 (crossover):} The linear channel is
      antipersistence-led ($w^{\text{lin}}_A = 0.35$), with
      smaller persistent and ARFIMA contributions; the
      volatility channel has largely handed off to the
      iid noise floor ($w^{\text{vol}}_{\text{MSM}}$ and
      $w^{\text{vol}}_{Vt}$ both near zero) --- the crossover
      between the MSM-dominated leading modes and the
      $F_{Vt}$-dominated deep modes.
\item \emph{Ranks 24--48 (deep eigenmodes):}
      $w^{\text{lin}}_A \in [0.65, 0.79]$ dominant, with
      $w^{\text{lin}}_M \in [0.19, 0.24]$ residual. The composite
      linear-channel profile is steeply decreasing
      $\kappastat^{\text{lin}}_i(H)$ --- the long-horizon
      mean-reversion signal of \citet{Cochrane1988} and
      \citet{FamaFrench1988}, here localised to the bulk of the
      eigenspectrum rather than the index level. The
      volatility-channel pattern in this range is dominated by
      $F_{Vt}$, rising from $w^{\text{vol}}_{Vt} \approx 0.57$ at
      rank 24 to $\approx 1.00$ at rank 48, with no MSM
      contribution.
\end{itemize}

\noindent
Comparable rank-pattern tables for the firsthalf, secondhalf,
and FF 100 datasets, in the supplementary material, show that
the same qualitative pattern reproduces in every cross-section:
sub-leading factor momentum, deep-mode mean reversion, and
market-mode rise-then-fall coexist with the volatility long
memory at the market mode and the transitory volatility
contribution at deep modes. The first six stylised facts ---
the weight-pattern facts (1)--(6) above --- are thus recovered
by the multi-memory factor model uniformly across all four
datasets.

The seventh stylised fact is a property of the volatility
eigenvalue spectrum rather than the per-mode weight pattern. On
the sensitivity panel, the share of total variance carried by
the leading eigenvalues of the $H$-horizon volatility covariance
$\Sigma^{\text{vol}}_H$ grows with the horizon: the top
eigenvalue carries $48\%$ of the volatility cross-section at the
daily horizon, rises to $76\%$ by the quarterly horizon, and
plateaus near $74$--$76\%$ thereafter, while the top three
eigenvalues reach $91\%$ at the five-year horizon. The
linear-channel covariance $\Sigma^{\text{lin}}_H$ behaves
oppositely: its top-eigenvalue share is roughly flat through the
quarterly horizon ($54$--$59\%$) and then falls to $45\%$ at five
years, while its top-ten share broadens to $93\%$. Volatility
long memory therefore concentrates cross-sectionally onto a
single common mode as the horizon lengthens, whereas linear
long-horizon structure diffuses across many modes --- the
spectral counterpart of the market-mode dominance of
$\kappastat^{\text{vol}}_1$ documented in
Table~\ref{tab:rank-weights}. The dominance of a single leading
``market'' eigenvalue in the cross-section of financial returns
is a central finding of the random-matrix-theory analysis of
correlation matrices \citep{LalouxEtAl1999, PlerouEtAl1999,
PlerouEtAl2002}; the volatility channel exhibits an analogous
concentration, with the distinguishing feature that it
intensifies with the horizon rather than holding fixed.

\subsection{Volatility-channel rank-1 cascade and the firsthalf
weak-identification}
\label{ssec:regime-change}

We anchor the analysis with the magnitude of the volatility-channel
signal at the market mode. Table~\ref{tab:kappa-vol-rank1} reports
loss-filtered bootstrap medians of $\kappastat^{\text{vol}}_1(H)$
across the four datasets at three canonical horizons.

\begin{table}[tb]
\centering
\caption{Loss-filtered bootstrap medians (and $90\%$ credible
intervals) of $\kappastat^{\text{vol}}_1(H)$ at the market mode
across the four datasets, at the monthly ($H = 21$), annual
($H = 252$) and five-year ($H = 1260$) horizons. Under the
i.i.d.\ null $\kappastat^{\text{vol}}_1(H) = 1$ at every horizon;
values shown sit one to two orders of magnitude above this null
on three of the four panels (sensitivity, secondhalf, FF 100)
and only modestly above on the firsthalf.}
\label{tab:kappa-vol-rank1}
\begin{tabular}{lccc}
\toprule
Dataset & $\kappastat^{\text{vol}}_1(21)$ &
          $\kappastat^{\text{vol}}_1(252)$ &
          $\kappastat^{\text{vol}}_1(1260)$ \\
\midrule
Sensitivity & $7.1\;[4.3, 8.9]$  & $34.0\;[12.4, 49.3]$ & $48.1\;[21.4, 76.9]$ \\
Firsthalf   & $2.9\;[2.8, 3.3]$  & $5.0\;[3.9, 8.1]$    & $3.4\;[1.8, 8.6]$    \\
Secondhalf  & $8.7\;[7.9, 8.9]$  & $42.1\;[20.0, 50.7]$ & $36.8\;[16.9, 57.2]$ \\
FF 100      & $6.9\;[5.0, 7.8]$  & $31.6\;[13.9, 42.6]$ & $44.2\;[20.5, 67.9]$ \\
\bottomrule
\end{tabular}
\end{table}

The factor weights at the market mode (rank 1) carry the
cross-panel cascade structure. Table~\ref{tab:rank1-weights}
reports the rank-1 bootstrap medians for all four panels.

\begin{table}[tb]
\centering
\caption{Loss-filtered bootstrap medians of the per-mode weights at
the market mode (rank 1) across the four datasets. The volatility
channel concentrates on the MSM cascade
($w^{\text{vol}}_{\text{MSM}} = 1.00$) on sensitivity and
secondhalf and on $w^{\text{vol}}_{\text{MSM}} = 0.66$ on FF~100;
the firsthalf panel is the exception, with a weakly-identified
rank-1 vol-channel mode-mix
($w^{\text{vol}}_{\text{MSM}} = 0.29$ at the median, $90\%$
unfiltered-bootstrap CI $[0.10, 1.00]$).}
\label{tab:rank1-weights}
\begin{tabular}{lccccc}
\toprule
Dataset & $w^{\text{lin}}_P$ & $w^{\text{lin}}_A$ &
$w^{\text{lin}}_M$ & $w^{\text{vol}}_{\text{MSM}}$ &
$w^{\text{vol}}_{Vt}$ \\
\midrule
Sensitivity & 0.16 & 0.31 & 0.19 & 1.00 & 0.00 \\
Firsthalf   & 0.31 & 0.07 & 0.41 & 0.29 & 0.04 \\
Secondhalf  & 0.23 & 0.38 & 0.12 & 1.00 & 0.00 \\
FF 100      & 0.37 & 0.10 & 0.16 & 0.66 & 0.02 \\
\bottomrule
\end{tabular}
\end{table}

Three of the four panels show clean MSM-cascade dominance at the
market mode: sensitivity and secondhalf both place
$w^{\text{vol}}_{\text{MSM}} = 1.00$ at rank 1, and FF~100 places
$w^{\text{vol}}_{\text{MSM}} = 0.66$. The exception is the
firsthalf panel ($1969$--$1997$), where the rank-1 vol-channel
allocation is weakly identified ---
$w^{\text{vol}}_{\text{MSM}} = 0.29$ at the median with a $90\%$
unfiltered-bootstrap CI of $[0.10, 1.00]$, spanning essentially
the entire $[0, 1]$ allocation space. The loss-conditional filter
(Appendix~\ref{app:identifiability}) tightens the firsthalf rank-1
CI only modestly, to $[0.08, 0.75]$.

Figure~\ref{fig:regime-change} visualises the cross-panel
comparison. The narrow whiskers on sensitivity, secondhalf, and
FF~100 at the rank-1 vol-channel mark these panels as containing
data from a post-transition regime where the multi-scale
volatility cascade has taken hold. The wide firsthalf rank-1
whisker reflects that the $1969$--$1997$ sample mixes pre- and
post-transition data: the cluster-conditional loss test of
Appendix~\ref{app:identifiability} shows the high-MSM and low-MSM
firsthalf clusters fit the data with near-equivalent loss,
inconsistent with the wide CI being optimiser-stuck noise on a
single basin. Section~\ref{ssec:rolling-window} uses a
rolling-window analysis to localise the transition.

\begin{figure}[H]
\centering
\includegraphics[width=0.95\linewidth]{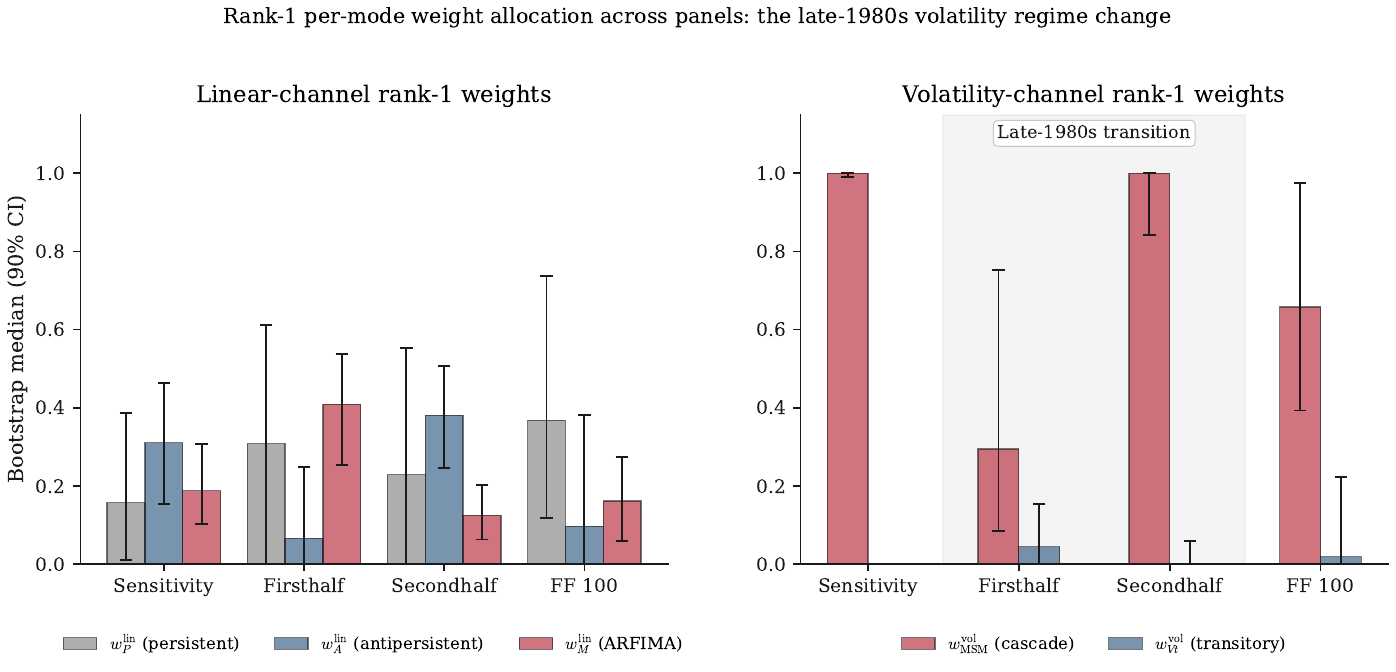}
\caption{Rank-1 per-mode weight allocation across the four panels.
Left: linear-channel weights ($w^{\text{lin}}_P$, $w^{\text{lin}}_A$,
$w^{\text{lin}}_M$). Right: volatility-channel weights
($w^{\text{vol}}_{\text{MSM}}$, $w^{\text{vol}}_{Vt}$). Bars are
bootstrap medians; whiskers are $90\%$ credible intervals. The
wide firsthalf whiskers on the rank-1 vol-channel allocation
reflect a weakly-identified mode-mix on a sample that straddles
the regime transition localised by the rolling-window analysis
(Section~\ref{ssec:rolling-window}).}
\label{fig:regime-change}
\end{figure}

The full-sample sensitivity and secondhalf panels both place
$w^{\text{vol}}_{\text{MSM}} = 1.00$ at the market mode. The
sensitivity panel pools the full $1969$--$2026$ sample and is
dominated by the post-transition signature because the
post-transition sub-period contains the larger raw volatility
variance (Section~\ref{ssec:data-methodology}). The FF~100 panel
places $w^{\text{vol}}_{\text{MSM}} = 0.66$ at rank 1; the
slightly diluted MSM share on FF~100 reflects the cross-sectional
sorting distributing the market mode's volatility cascade across
multiple top eigenmodes (see supplementary
material~\ref{app:ff100-detail}).

The firsthalf rank-1 vol-channel allocation cannot be sharpened by
adding $n_{\text{starts}}$ or tightening the optimiser tolerance:
it reflects a genuine mode-mix on a sample that contains both
pre- and post-transition data. A clean sub-sample contrast
requires localising the transition itself, which the next
subsection does.

\paragraph{Cross-over time scale of the MSM cascade.} A
quantitative correlate of the rank-1 cascade is the duration
$1/\gamma_1$ of the lowest-frequency MSM component, which sets
the time scale over which the multiplicative-cascade structure
unfolds. Working backward from the fitted
$(\gamma_{\bar k}, b)$ via $\gamma_k = 1 - (1 - \gamma_1)^{b^{k - 1}}$
\eqref{eq:msm-switching}, the loss-filtered bootstrap medians (with
$90\%$ CIs) of $1/\gamma_1$ across the four datasets are:
\begin{equation*}
\begin{aligned}
\text{Sensitivity:}\quad & 1.1 \text{ yr } [0.5, 18.9], &
\text{Firsthalf:}\quad   & 2.2 \text{ yr } [0.5, 4.5], \\
\text{Secondhalf:}\quad  & 4.0 \text{ yr } [0.7, 32.6], &
\text{FF 100:}\quad      & 2.0 \text{ yr } [0.5, 6.8].
\end{aligned}
\end{equation*}
The secondhalf cross-over time scale is approximately
$1.8\times$ the firsthalf's at the median. The wide upper CIs on
sensitivity and secondhalf reflect the weak identification of
$\gamma_{\bar k}$ in the post-transition cascade regime
(Section~\ref{ssec:global-params}); the firsthalf and FF~100
panels have tighter upper CIs because their fits assign less mass
to the very slow tail of the cascade. Figure~\ref{fig:crossover}
displays the full bootstrap distributions. The cross-over
time-scale contrast is the quantitative counterpart of the
rank-1 cascade contrast: the post-transition cascade regime
extends the lowest-frequency component of the volatility cascade
to multi-year horizons, while the firsthalf sample --- which
mixes pre- and post-transition data --- sits at a shorter
median time scale.

\begin{figure}[tb]
\centering
\includegraphics[width=0.85\linewidth]{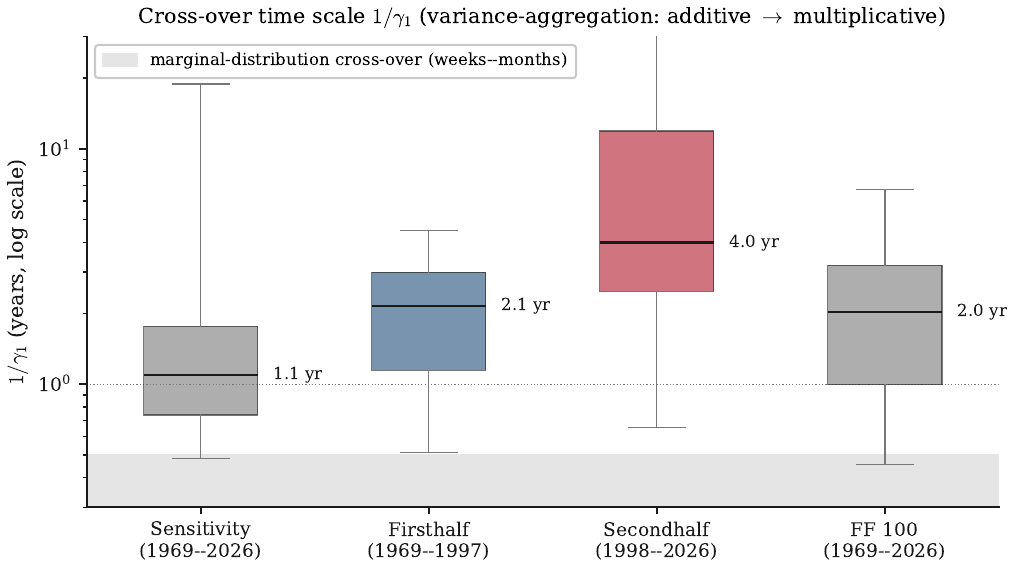}
\caption{Bootstrap distributions of the cross-over time scale
$1/\gamma_1$ (years, log scale) across the four panels. Boxes
span the interquartile range; whiskers reach the $5\%$ and
$95\%$ percentiles. Annotated medians: sensitivity $1.1$ yr,
firsthalf $2.2$ yr, secondhalf $4.0$ yr, FF 100 $2.0$ yr. The
pale band at the bottom marks the much-faster
marginal-distribution cross-over (truncated-L\'evy-flight to
approximately Gaussian aggregation, on the order of weeks to
months), discussed in Section~\ref{ssec:1998-regime-change}.}
\label{fig:crossover}
\end{figure} The mid-frequency
component duration $1/\gamma_4$ (the geometric midpoint of the
cascade for $\bar k = 8$) is approximately $0.2$ yr ($\sim 50$
trading days) across all four datasets --- a time scale that
matches typical empirical volatility-autocorrelation decay
constants in the econophysics literature
\citep{Cont2001,BouchaudPotters2003}. Our $1/\gamma_1$ estimates
sit somewhat below the longest-duration $10$--$20$ yr range
reported by \citet[\S9.3]{CalvetFisher2008} for U.S.\ aggregate
equity at preferred $\bar k \in \{6, 8\}$; the discrepancy is
plausibly because their identification uses maximum likelihood
estimation (MLE) on a univariate
aggregate-return likelihood while ours uses joint LS on
eigenmode-decomposed $(\kappastat^{\text{lin}},
\kappastat^{\text{vol}})$ statistics across the full
cross-section.

\subsection{Rolling-window localisation of the regime transition}
\label{ssec:rolling-window}

The static split-half analysis sets the firsthalf--secondhalf
boundary at $1998$, the approximate midpoint of the
$1969$--$2026$ sample. The firsthalf panel's weakly-identified
rank-1 vol-channel mode-mix
(Section~\ref{ssec:regime-change}) is the first sign that the
underlying regime transition does not coincide with the $1998$
boundary: a $28$-year window placed before the transition would
produce a tight rank-1 allocation, and one placed after the
transition would produce a tight (and different) allocation;
the wide CI on $1969$--$1997$ is the signature of a sample
that straddles the transition. We localise the transition with
a rolling-window fit.

The rolling window slides $28$-year sub-samples (matching the
firsthalf and secondhalf panel lengths for direct
apples-to-apples comparison) in $2$-year strides across the
full $1969$--$2026$ sample, producing $15$ windows centred
$1983$-$06$ through $2011$-$06$. At each window we run a
$1000$-replicate moving-block bootstrap on the joint
factor model of Section~\ref{ssec:joint-LS-fit}, with the
same multi-start configuration and loss-conditional below-
median filter as the static panels
(Section~\ref{ssec:data-methodology}) applied per-window.
Chaining the windows by centre date traces both the parameter
medians and the $90\%$ confidence bands across the sample
period.

Three trajectories on the FF~49 industry panel carry the
transition signal. The rank-1 MSM weight
$w^{\text{vol}}_{\text{MSM}}$ has median $0.31$ in the
earliest window (centre $1983$-$06$, sample span
$1969$--$1997$), drifts through $0.30$ ($1985$ centre), rises
to $0.47$ ($1987$ centre), $0.77$ ($1989$ centre), $0.93$
($1991$ centre), $0.99$ ($1993$ centre), and saturates at
median $1.00$ from the $1995$-centred window onward (sample
span $1981$--$2009$). The rank-1 linear $F_M$ weight
$w^{\text{lin}}_M$ drifts monotonically downward from median
$0.40$ in the $1983$-centre window to $0.12$ in the
$2011$-centre window. The ARFIMA persistence parameter
$\phi_{Vt}$ falls from median $0.43$ in the early-sample windows
to $0.31$ stable from the mid-$1990$s centres onward.

The $90\%$ CI bands around these medians make the regime change
statistically robust to the bootstrap resampling. On
$w^{\text{vol}}_{\text{MSM}}$, the $1983$-centre window has CI
$[0.11, 0.80]$ and the $1999$-centre window has CI
$[0.84, 1.00]$ --- the upper bound of the early-sample CI sits
strictly below the lower bound of the late-sample CI, so the
early and late distributions are non-overlapping at the $90\%$
level. On $w^{\text{lin}}_M$, the analogous bands separate by
the $1995$ centre: the $1983$-centre window has CI
$[0.25, 0.54]$ against the $2011$-centre CI $[0.06, 0.20]$, so
the linear-market mode rank-1 weight ceding to the
volatility-MSM mode rank-1 weight is itself a CI-separated
transition rather than a noise artefact.

The transition-period windows carry a substantive bimodality.
Across the $1987$-, $1989$-, $1991$- and $1993$-centre windows,
the median $w^{\text{vol}}_{\text{MSM}}$ rises from $0.47$ to
$0.99$, but the $90\%$ CIs span $[0.12, 1.00]$ through
$[0.23, 1.00]$ --- the lower bounds remain near the
pre-transition floor while the upper bounds saturate at the
post-transition ceiling. The per-replicate bootstrap fits in
these windows are not concentrated near the median; they
scatter between the burst and cascade regimes within each
window's resampled sample. The $28$-year window straddles the
transition for the $1987$--$1993$ centres, so each resample
captures a different mix of pre- and post-transition data and
the joint fit lands in whichever basin matches the dominant
content of the resampled panel. By the $1999$ centre (sample
span $1985$--$2013$), the cascade regime dominates every
resample and the CI tightens to $[0.84, 1.00]$.

The rolling-window data alone do not pin the transition to a
single year --- the $28$-year window absorbs the transition
over a multi-year band, and the per-window bimodality on
$1987$--$1993$ centres is consistent with either a sharp switch
near the late 1980s or a more gradual diffusion through the
1990s. The trajectory is inconsistent with a transition near
$1998$, the static split-half boundary: $w^{\text{vol}}_{\text{MSM}}$
on the windows whose samples are dominantly pre-$1998$ (the
$1985$- and $1987$-centre windows) has median $0.30$--$0.47$,
not the cascade-saturated $1.00$ of the post-1998 windows. The
firsthalf static panel ($1969$--$1997$) contains roughly a
decade of post-transition data, which is the source of the
rank-1 vol-channel weak identification documented in
Section~\ref{ssec:regime-change}.
Figure~\ref{fig:rolling-window} shows the
$w^{\text{vol}}_{\text{MSM}}$ and $w^{\text{vol}}_{Vt}$ median
trajectories with their $90\%$ CI ribbons on the FF~49 panel.

\begin{figure}[tb]
\centering
\includegraphics[width=0.95\linewidth]{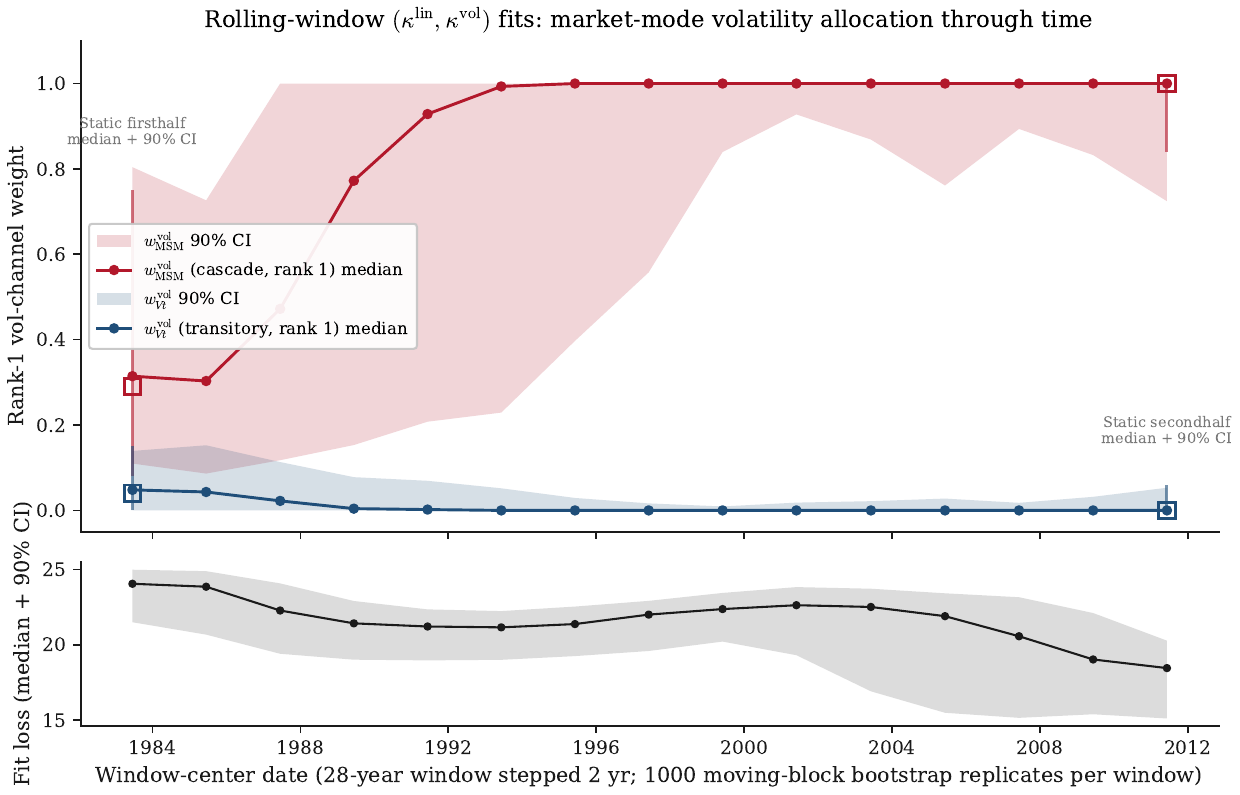}
\caption{Rolling-window rank-1 volatility-channel weights on the
FF~49 industry panel: $28$-year windows in $2$-year strides,
$15$ windows centred $1983$-$06$ through $2011$-$06$, with
$1000$-replicate moving-block bootstrap per window. Top: rank-1
$w^{\text{vol}}_{\text{MSM}}$ (cascade weight, red) and
$w^{\text{vol}}_{Vt}$ (transitory weight, blue) median
trajectories as a function of window-centre date, with $90\%$
CI ribbons. The static firsthalf and secondhalf loss-filtered
medians and $90\%$ CIs are overlaid as horizontal reference
bands at the corresponding centres. Bottom: per-window median
fit loss with $90\%$ CI ribbon. The MSM-cascade weight median
rises from $0.30$ ($1985$ centre) through $0.93$ ($1991$
centre) and saturates at $1.00$ from the $1995$ centre onward;
the $1983$-centre CI upper bound ($0.80$) sits strictly below
the $1999$-centre CI lower bound ($0.84$), so the early- and
late-sample $90\%$ CI bands on $w^{\text{vol}}_{\text{MSM}}$
are non-overlapping. The transition is localised in the late
1980s rather than at the static $1998$ boundary.}
\label{fig:rolling-window}
\end{figure}

A point-estimate rolling-window sweep (single fit per window,
$1$-year stride, $29$ windows) on the FF~100 size $\times$
book-to-market panel ($N = 95$) reproduces the on-average
pattern (rank-1 MSM weight predominantly above $0.5$ across
windows, $w^{\text{lin}}_M$ drifting downward over the run) but
with substantially higher per-window noise:
$w^{\text{vol}}_{\text{MSM}}$ bounces $0.32$--$1.00$ across the
$29$ windows rather than stepping cleanly. The noise is driven
by the larger per-mode weight parameter count ($5N = 475$ vs
$240$ on the industry sorts), which makes a single-fit window
inherit more sampling variability. The cross-dataset
universality of the cascade-dominance structure holds; the
per-window resolution is too low to localise the transition
date independently of the FF~49 analysis. We did not run the
$1000$-replicate bootstrap on the FF~100 rolling-window because
the FF~49 bootstrap CIs above already pin the transition
robustly; the FF~100 single-fit sweep functions as a
cross-dataset universality check rather than an independent
localisation.

Section~\ref{ssec:1998-regime-change} returns to the structural
reading of the transition.

\subsection{Cross-channel \texorpdfstring{$\beta$}{β}-inversion:
a falsification test}
\label{ssec:beta-inversion}

The factor model \eqref{eq:factor-decomp} predicts that the
$F_M$ factor's cross-asset loading $\beta_M \in \R^N$ governs both
the linear and the volatility channels: the per-mode weights
$w^{\text{lin}}_{i, M}$ are proportional to
$|\langle \eivec_i^{\text{lin}}, \beta_M \rangle|^2$ and
$w^{\text{vol}}_{i, \text{MSM}}$ are proportional to
$|\langle \eivec_i^{\text{vol}}, \beta_M \odot \beta_M \rangle|^2$
(Section~\ref{ssec:per-mode-predictions}). A non-trivial cross-channel
test of the model is therefore to recover the asset-wise
$|\beta_M[a]|^2$ from each channel independently and compare them.

Define the sign-averaged $\beta$-inversion diagnostics
\begin{equation}
  \widehat{\beta_M^2}[a]^{\text{lin}} \;:=\;
    \sum_{i = 1}^N (\eivec_i^{\text{lin}}[a])^2 \cdot w^{\text{lin}}_{i, M},
  \label{eq:beta-inv-lin}
\end{equation}
\begin{equation}
  \widehat{\beta_M^2}[a]^{\text{vol}} \;:=\;
    \bigg(\sum_{i = 1}^N (\eivec_i^{\text{vol}}[a])^2 \cdot
      w^{\text{vol}}_{i, \text{MSM}}\bigg)^{1/2},
  \label{eq:beta-inv-vol}
\end{equation}
which extract the asset-level contribution of $F_M$ to each
channel up to a sign-pattern that averages out over the bootstrap.
Under perfect cross-channel consistency, the per-replicate
Spearman correlation between $\widehat{\beta_M^2}^{\text{lin}}$
and $\widehat{\beta_M^2}^{\text{vol}}$ across the $N$ assets
should be positive with a credible interval excluding zero.
A per-asset scatter visualisation of the test is provided in
supplementary material~\ref{app:beta-inv-detail}.

\begin{table}[tb]
\centering
\caption{Cross-channel $\beta$-inversion correlations across the
four datasets: Pearson and Spearman of
$\widehat{\beta_M^2}^{\text{lin}}$ vs.\
$\widehat{\beta_M^2}^{\text{vol}}$ across the $N$ assets,
loss-filtered bootstrap median and $90\%$ CI ($500$ replicates per
panel). Bold values identify CIs excluding zero. The Pearson
correlation is negative on all four panels and significantly so
on sensitivity, secondhalf and FF 100; the Spearman is
non-positive on three panels and significantly negative on the
full-sample sensitivity panel. On no panel is the correlation
significantly positive --- the cross-channel rank alignment
predicted by a single shared $\beta_M$ is absent. A formal
permutation test is reported in
Section~\ref{ssec:permutation-test}.}
\label{tab:beta-inversion}
\begin{tabular}{lcccc}
\toprule
Dataset & Pearson median & Pearson CI & Spearman median & Spearman CI \\
\midrule
Sensitivity & \textbf{$-0.771$} & $[-0.852,\, -0.656]$
            & \textbf{$-0.537$} & $[-0.698,\, -0.188]$ \\
Firsthalf   & $-0.322$          & $[-0.625,\, +0.079]$
            & $+0.123$          & $[-0.284,\, +0.443]$ \\
Secondhalf  & \textbf{$-0.642$} & $[-0.842,\, -0.224]$
            & $-0.142$          & $[-0.546,\, +0.108]$ \\
FF 100      & \textbf{$-0.798$} & $[-0.884,\, -0.414]$
            & $-0.571$          & $[-0.747,\, +0.019]$ \\
\bottomrule
\end{tabular}
\end{table}

Table~\ref{tab:beta-inversion} rejects the cross-channel
prediction. Under a single multifractal loading $\beta_M$ shared
between the linear and volatility channels, the per-asset
attributions $\widehat{\beta_M^2}^{\text{lin}}$ and
$\widehat{\beta_M^2}^{\text{vol}}$ would be positively
correlated. They are not. The Pearson correlation is negative on
every panel --- significantly so, with a $90\%$ CI excluding
zero, on sensitivity, secondhalf, and FF 100 --- and the
Spearman correlation is significantly negative on the sensitivity
panel while its bootstrap interval includes zero on the other
three. On no panel is either correlation significantly positive.

The verdict is consistent across the temporal and
cross-sectional cuts. A pooled-sample artefact would show up as
a negative full-sample correlation that reverses within each
sub-period; instead the firsthalf and secondhalf panels are
themselves negative or null, so the negative full-sample
correlation is not an artefact of mixing the two halves. The
assets that carry $F_M$'s long-memory signature in the linear
channel are not the assets that carry its multifractal-cascade
signature in the volatility channel: the two channels do not
share a common cross-asset loading.
Section~\ref{ssec:permutation-test} formalises this with a
permutation test against the null of random cross-asset
alignment, and Section~\ref{ssec:rank-consistency} draws out the
structural reading.

\subsection{Formal hypothesis test of cross-channel rank
consistency}
\label{ssec:permutation-test}

The bootstrap credible intervals reported in
Table~\ref{tab:beta-inversion} are estimates of parameter
uncertainty around the per-panel point estimate; on their own they
are not formal tests of the rank-consistency claim against the
null of random cross-asset alignment. We complement them with a
permutation test of the cross-channel Spearman correlation under
the null hypothesis that the linear- and volatility-channel
$\beta_M^2$ attributions are independent across assets.

The test statistic is the Spearman correlation between
$\widehat{\beta_M^2}^{\text{lin}}$ and
$\widehat{\beta_M^2}^{\text{vol}}$ from
\eqref{eq:beta-inv-lin}--\eqref{eq:beta-inv-vol}, computed on the
asset-wise median of the loss-filtered bootstrap distribution
(Section~\ref{ssec:data-methodology}). The permutation null is
generated by re-shuffling one channel's attribution across the
$N$ assets and recomputing the Spearman, with $K = 10{,}000$
permutations. We report two-sided $p$-values, both raw and after
Bonferroni (FWER) and Benjamini--Hochberg (FDR) correction across
the four panels (Table~\ref{tab:permutation}). As a
bootstrap-aware robustness check we additionally evaluate the
cross-channel Spearman test on each loss-filtered replicate and
report the fraction of replicates reaching $p < 0.05$ (two-sided).

\begin{table}[tb]
\centering
\caption{Permutation test for cross-channel rank consistency on
the loss-filtered bootstrap. $\hat\rho$ is the observed Spearman
between the linear- and volatility-channel sign-averaged
$\beta_M^2$ attributions on the asset-wise median vectors. The
raw two-sided $p$-value is the fraction of $K = 10{,}000$
permutations of one channel's attribution that produce a
Spearman at least as extreme as $\hat\rho$. Bonferroni and
Benjamini--Hochberg adjustments are applied across the four
panels. The final column reports the fraction of loss-filtered
bootstrap replicates ($500$ per panel) whose per-replicate
Spearman test rejects at the $\alpha = 0.05$ level.
Bold rows: Bonferroni-adjusted $p < 0.01$.}
\label{tab:permutation}
\begin{tabular}{lrrrrr}
\toprule
Panel        & $\hat\rho$ & $p$ (raw) & Bonferroni
             & BH--FDR    & frac $p < 0.05$ \\
\midrule
\textbf{Sensitivity}       & $\mathbf{-0.578}$ & $\mathbf{<10^{-4}}$
                           & $\mathbf{0.0004}$ & $\mathbf{0.0002}$
                           & $\mathbf{0.920}$ \\
Firsthalf                  & $+0.181$ & $0.215$
                           & $0.858$  & $0.286$ & $0.240$ \\
Secondhalf                 & $-0.056$ & $0.700$
                           & $1.000$  & $0.700$ & $0.112$ \\
\textbf{FF 100}            & $\mathbf{-0.615}$ & $\mathbf{<10^{-4}}$
                           & $\mathbf{0.0004}$ & $\mathbf{0.0002}$
                           & $\mathbf{0.806}$ \\
\bottomrule
\end{tabular}
\end{table}

\textbf{Sensitivity and FF 100 reject the independence null ---
toward anti-alignment.} The observed Spearman is
$\hat\rho = -0.578$ and $-0.615$ on the two panels, with
Bonferroni-adjusted $p = 0.0004$ in both
(Table~\ref{tab:permutation}). On the two largest panels the
linear- and volatility-channel $\beta_M$ attributions are not
merely uncorrelated but significantly \emph{anti}-aligned.

\textbf{Firsthalf and secondhalf do not reject.} Firsthalf
($\hat\rho = +0.181$, raw $p = 0.21$) is positive in direction
but far from significance under either correction; secondhalf
($\hat\rho = -0.056$, raw $p = 0.70$) is statistically
indistinguishable from independent cross-asset alignment.

On no panel is there a significant \emph{positive} correlation.
The cross-channel rank alignment predicted by a single shared
multifractal loading $\beta_M$ --- a positive Spearman clear of
zero --- is absent on all four datasets, and on the two largest
panels the data point the other way. The per-replicate Spearman
fractions in the rightmost column of
Table~\ref{tab:permutation} corroborate this pattern within the
bootstrap: $92\%$ of sensitivity replicates and $81\%$ of
FF~100 replicates reject independence at $\alpha = 0.05$ on
their own per-replicate test, while only $24\%$ of firsthalf
and $11\%$ of secondhalf replicates do --- indicating that the
panel-level Bonferroni rejection on the two large panels is
broadly supported across the loss-filtered bootstrap rather
than driven by a few extreme draws. The permutation test thus
confirms, with multi-testing correction across the four panels,
the falsification read off the credible intervals of
Section~\ref{ssec:beta-inversion}.

\subsection{Robustness}
\label{ssec:robustness}

The empirical results above sit in three layers of robustness.

\emph{Cross-dataset robustness.} The three sub-samples (firsthalf,
secondhalf, FF~100) each constitute an independent realisation of
the U.S.\ equity universe through different temporal or
cross-sectional cuts. The seven stylised facts of
Section~\ref{ssec:stylised-facts} reproduce in all three ---
sub-leading factor momentum, deep-mode mean reversion,
market-mode rise-then-fall, and the volatility long-memory and
transitory-volatility weight patterns --- as documented in the
per-dataset weight tables of supplementary
material~\ref{app:ff100-detail}. The cross-channel
$\beta$-inversion test of Section~\ref{ssec:beta-inversion}
likewise returns the same verdict on every panel: no evidence of
a multifractal loading shared between the linear and volatility
channels.

\emph{Cross-region robustness.} A fifth panel constructed from
the \citet{KennethFrenchDataLibrary} Europe 25
size$\times$book-to-market sort
($N = 25$ developed-European portfolios, $T = 9{,}327$ days,
1990-07-02 to 2026-03-31) provides a non-U.S.\ replication. The
joint LS fit on this panel recovers the same eigenstructure
under the same five-factor multi-memory specification: the
fractional-Brownian Hurst exponents come in at
$H_P = 0.57$ $[0.51, 0.66]$ and $H_A = 0.28$ $[0.11, 0.33]$, at
the upper edge of the U.S.\ cross-panel range but consistent
with it; the ARFIMA persistence and antipersistence parameters
$(\phi_M, d_M, \phi_{Vt}, d_{Vt})$ lie within the U.S.\ range
at $(0.86, -0.32, 0.38, -0.16)$; the rank-1 vol-channel
allocation saturates at $w^{\text{vol}}_{\text{MSM}} = 1.00$
$[0.96, 1.00]$, the post-transition cascade signature (the
Europe-25 sample starts after the late-1980s regime transition
of Section~\ref{ssec:rolling-window}); and the cross-channel
$\beta$-inversion test replicates the U.S.\ negative-correlation
signature with Pearson median $-0.65$ $[-0.85, +0.12]$,
$92.8\%$ of replicates with $\rho < 0$. One MSM cascade
parameter diverges from the U.S.\ cross-panel range: the
cascade base $b$ pins at the lower bound, $b = 1.500$
$[1.500, 1.551]$, against U.S.\ values in $[2.11, 2.60]$. This
divergence is robust to the bootstrap and is discussed further in
Section~\ref{ssec:universality}.

\emph{Cross-objective robustness.} We assess sensitivity to the
choice of estimation objective. An alternative specification
uses per-horizon residual scaling with a rank-1 emphasis weight
in place of the per-(mode, horizon) scaling and three-pass
warm-restart multi-start of Section~\ref{ssec:joint-LS-fit}; the
substantive findings are reproduced to within a few percentage
points under either objective (supplementary
material~\ref{app:objective-comparison}).

\section{Discussion}
\label{sec:discussion}

The empirical results of Section~\ref{sec:empirical} land three
findings: the multi-memory factor model recovers the seven
stylised facts of long-horizon equity dynamics consistently
across four cross-sectional and temporal cuts of the data; a
rolling-window analysis localises a market-mode volatility
regime transition in the late 1980s, sharper than the static
pre-1998 / post-1998 split-half analysis would suggest; and a
cross-channel $\beta$-inversion test rejects the hypothesis that
the linear and volatility imprints of the central factor $F_M$
are governed by a single shared cross-asset loading. We discuss
the structural reading of each in turn.

\subsection{Universality of the multi-memory eigenstructure}
\label{ssec:universality}

The stylised-facts reproduction (Section~\ref{ssec:stylised-facts})
is a universality claim. The same per-mode weight pattern ---
factor-momentum persistence concentrated at sub-leading
eigenmodes, deep-mode mean reversion, market-mode rise-then-fall
with multi-scale volatility on top --- appears under four
distinct samplings of U.S.\ equity returns: a 57-year industry
sort, a 28-year pre-1998 sub-period of the same sort, a 28-year
post-1998 sub-period, and a 57-year size$\times$book-to-market
sort with a different cross-sectional partition entirely. The
fractional-Brownian-motion exponents
$H_P \approx 0.52$--$0.57$ and $H_A \approx 0.17$--$0.27$
characterising the persistent and antipersistent factors are
stable across all four panels (Table~\ref{tab:global-params}).
The cross-region replication reported in
Section~\ref{ssec:robustness} adds a fifth panel constructed
from the Kenneth French Europe 25 size$\times$book-to-market
sort: the eigenstructure-level decomposition recovers, with the
fractional-Brownian and ARFIMA factor-profile parameters within
or at the upper edge of the U.S.\ cross-panel range, the
post-transition cascade saturation at the market mode
replicating (the Europe sample starts after the late-1980s
regime transition of Section~\ref{ssec:rolling-window}), and the
$\beta$-inversion finding replicating with the same
negative-correlation sign.

The universality should not be taken as a claim of
sample-mean invariance. The empirical variance ratios
$\widehat{\kappastat}^{\text{lin}}_i(H)$ and
$\widehat{\kappastat}^{\text{vol}}_i(H)$ themselves vary across
panels; the loss values are higher on the smaller and
differently-sorted panels. What is invariant is the
\emph{structural decomposition}: each empirical
$\widehat{\kappastat}$ matrix admits a parsimonious five-factor
fit with the same factor-profile parameters, and the
distribution of cross-asset loadings on those factors is what
generates each panel's particular $\widehat{\kappastat}$
pattern. The five factor profiles are the equity universe's
building blocks; the cross-asset loadings $\{\beta_k\}$ are the
panel-specific weights on those blocks.

This is reminiscent of universality classes in the statistical
physics of critical phenomena: the same set of exponents
$\{H_P, H_A, \phi_M, d_M, \phi_{Vt}, d_{Vt},
m_0, b, \gamma_{\bar k}\}$ describes systems that differ in
microscopic details (which firms are in which industry, which
size$\times$book-to-market bin) but share underlying long-range
correlations. The empirical robustness of the factor-profile
parameters across our four U.S.\ panels plus the European
replication is consistent with this reading.

\paragraph{The cascade base on the European panel.} The MSM
cascade base $b$ is
the one parameter that diverges from the U.S.\ cross-panel range
on the Europe panel: the loss-filtered median pins at the lower
bound $b = 1.500$ (90\% CI $[1.500, 1.551]$), in contrast to the
U.S.\ panels where $b$ takes values $2.1$--$2.6$. The cascade
base controls the geometric spacing of switching probabilities
$\gamma_k = 1 - (1 - \gamma_1)^{b^{k - 1}}$, and a small $b$
collapses the cascade toward low frequencies (Europe's
$\gamma_{\bar k} = 0.07$ is correspondingly low against the U.S.\
range $0.46$--$0.79$). Three readings are consistent with the
data: a genuine slow-cascade signal on the $1990$--$2026$
European sample; a sample-length identification effect ($N = 25$
and $T = 9327$ give less power to identify $b$ than the U.S.\
panels with larger $N$ or longer $T$); or a regional structural
difference. Distinguishing them is a natural follow-up. The
fractional-Brownian and ARFIMA factor-profile parameters
universalise cleanly; the MSM cascade parametrisation is the
locus of any residual sample- or region-specific variation.

The robustness of the eigenstructure-level decomposition across
our five panels motivates the model's application to further
cross-sections (other developed and emerging equity markets,
fixed-income, commodities) in follow-up work.

\subsection{Volatility regime transition}
\label{ssec:1998-regime-change}

The cross-panel comparison of Section~\ref{ssec:regime-change}
and the rolling-window bootstrap of
Section~\ref{ssec:rolling-window} jointly localise a market-mode
volatility regime transition to the late 1980s: clean MSM-cascade
dominance on the secondhalf, full-sample sensitivity, and FF~100
panels; a weakly-identified rank-1 vol-channel mode-mix on the
firsthalf panel; and a rank-1 MSM-weight median trajectory across
the $15$ rolling windows that saturates at $1.00$ by the
$1995$-centred window with strictly non-overlapping $90\%$ CI
bands separating the early-sample and late-sample windows. The
transition-period windows ($1987$--$1993$ centres) carry
replicate-level bimodality that the $28$-year window cannot
resolve.

The contrast admits a direct structural reading in the factor
specifications of Section~\ref{sssec:multifractal-factors}. The
post-transition regime is dominated by the MSM cascade ---
persistent volatility at multiple frequencies simultaneously,
in line with the Calvet--Fisher tradition. The post-transition
events visible to our sample (the dot-com episode, the 2008
global financial crisis, the COVID-19 shock, the 2022 inflation
episode) sit within a high-baseline-volatility environment in
which the MSM cascade specification predicts
hyperbolically-decaying autocorrelation over a wide inertial
range \citep[Proposition~1]{CalvetFisher2008}; our fits confirm
this on the post-1998 secondhalf sub-period and on the
full-sample sensitivity panel.

The pre-transition regime is not directly observable in any
single static panel: the firsthalf panel ($1969$--$1997$)
contains roughly a decade of post-transition data, which is the
source of its weakly-identified rank-1 vol-channel allocation
(Section~\ref{ssec:regime-change}). The rolling-window fits
centred $1983$ and $1985$ (sample spans starting in $1969$ and
$1971$ respectively) provide indirect access: their median
rank-1 $w^{\text{vol}}_{\text{MSM}}$ values ($0.31$ and $0.30$
respectively, with $90\%$ CI upper bounds at $0.80$ and $0.73$)
are consistent with the pre-transition regime carrying less
multi-scale long-memory volatility than the post-transition
regime, though the $28$-year rolling window's resolution is
too coarse to pin the pre-transition regime's mode-mix cleanly.
The bimodality on the transition-period windows
($1987$--$1993$ centres) is itself a structural signature ---
either a sharp regime switch in the late 1980s, averaged into
the wide CI bands by $28$-year smoothing, or a more gradual
sectoral diffusion through the 1990s; the rolling-window
analysis as configured cannot discriminate the two, though
both readings are consistent with the late-1980s localisation
of the central transition mass.

The quantitative cross-over time scale estimated in
Section~\ref{ssec:regime-change} sharpens the regime-transition
reading. The lowest-frequency MSM component duration $1/\gamma_1$
is $1.1$ yr in the full sample, $2.2$ yr in the firsthalf, and
$4.0$ yr in the secondhalf at the median. In the
multifractal-cascade language of
\citet{MandelbrotCalvetFisher1997} and \citet{CalvetFisher2008}
this is the time horizon at which the multiplicative-cascade
structure becomes the dominant volatility-aggregation channel,
displacing the additive (Gaussian-innovation-driven) structure of
the high-frequency components. That this time scale lengthens
between the static firsthalf and secondhalf sub-periods is a
direct empirical statement about how the equity market's
volatility-generating mechanism has changed at the multi-year
horizon, with the firsthalf estimate itself a regime-mixture
between pre- and post-transition data and the secondhalf estimate
characterising the post-transition cascade alone.

The cross-over time scale $1/\gamma_1$ should be distinguished
from the much faster cross-over at which the marginal return
distribution converges from its short-horizon power-law-tailed
form \citep{Gopikrishnan1999,Gabaix2003} to an approximately
Gaussian form. The latter cross-over is a central-limit-theorem
phenomenon for finite-variance fat-tailed innovations
(the truncated-L\'evy-flight regime of
\citet{Mantegna2000}) and occurs on the order of weeks to months
for liquid equity returns. The cross-over time scale we estimate
here, by contrast, is a property of the variance-aggregation
mechanism rather than of the marginal distribution: it is the
horizon at which the slowest cascade frequency $\gamma_1$ no
longer averages out within the aggregation window, and the
multiplicative cascade structure overtakes the additive
aggregation of independently-modulated innovations. The two
cross-overs are linked --- the volatility-feedback channel of
\citet[Ch.~9]{CalvetFisher2008} ties low-frequency cascade
components to large-magnitude return innovations via a factor of
$1/\gamma_k$, which is what makes the slow cascade components
relevant for return-side dynamics even though they switch only
rarely --- but they operate at time scales separated by roughly
two orders of magnitude, and our estimate of $1/\gamma_1$ pins
the slow cross-over, not the marginal-distribution cross-over.

Several macro- and microstructural mechanisms plausibly underlie
a late-1980s transition. The 1987 portfolio-insurance crash and
its regulatory aftermath (circuit breakers, uptick-rule revisions,
and the post-crash expansion of the equity-derivatives ecosystem)
restructured the relationship between cash and derivative equity
markets. The end of the Cold War and the late-1980s acceleration
of equity-market globalisation enlarged the set of macro shocks
the U.S.\ equity universe responds to. The 1986 Tax Reform Act
and the contemporaneous savings-and-loan restructuring reshaped
financial-sector incentives. The shift in industry composition
toward technology, financial services, and intangibles-intensive
firms \citep{BrynjolfssonMcAfee2014} that intensified through
the 1990s and 2000s changes the cross-section's sensitivity to
aggregate volatility on a longer time scale. Distinguishing among
these mechanisms is a structural-econometrics question outside
the scope of this paper; what we establish here is that the
eigenstructure-level signature of the transition is sharp and
empirically robust, and that it has a clean interpretation in
the multifractal-volatility framework.

\subsection{Channel-distinct \texorpdfstring{$\beta$}{β} structure}
\label{ssec:rank-consistency}

The cross-channel $\beta$-inversion test
(Section~\ref{ssec:beta-inversion}) was built as a falsification
test of the multi-memory specification. The model attributes
both the linear-channel long-memory signature and the
volatility-channel multifractal cascade of the factor $F_M$ to a
single cross-asset loading $\beta_M$: the linear channel
diagonalises $\beta_M \beta_M^\top$ and the volatility channel
$(\beta_M \odot \beta_M)(\beta_M \odot \beta_M)^\top$
(Proposition~\ref{prop:linvol-eigvec-coincidence}). If one
$\beta_M$ drove both, the per-asset attributions recovered from
the two channels would be positively rank-correlated. They are
not: the correlation is negative or null on every panel, and
significantly negative on the two largest
(Table~\ref{tab:permutation}). The assets that carry $F_M$'s
long-memory imprint in the linear channel are not the assets
that carry its cascade imprint in the volatility channel.

The test therefore falsifies a specific cross-channel sub-claim
--- the identification of the linear and volatility imprints as
a single object --- without disturbing the rest of the
framework. The per-eigenmode $(\kappastat, \Omat)$ decomposition,
the seven stylised facts, and the five-factor fit of each
individual panel are unaffected; what the test removes is the
unification of the linear ARFIMA memory and the volatility
cascade under one loading. The natural revision --- separate
cross-asset loadings for the linear-channel long-memory factor
and the volatility-channel cascade, fitted jointly but not
constrained to coincide --- is a model-specification question we
leave to follow-up work. That a cleanly-constructed cross-channel
test can reject a structural sub-claim of the model while leaving
its eigenstructure-level fit intact is itself a point in favour
of the $(\kappastat, \Omat)$ framework: the two statistics carry
enough information to discriminate structural hypotheses that
scalar variance-ratio tests cannot.

\subsection{Limits and follow-up}
\label{ssec:limits-and-follow-up}

The model and the empirical results have four principal
limits. \emph{First}, the $\widehat{\kappastat}$ moment alone
identifies the global parameters $(H_P, H_A, \phi_M, d_M,
\phi_{Vt}, d_{Vt})$ tightly, and the MSM parameters
$(m_0, b, \gamma_{\bar k})$ less tightly. A natural extension
would add the cross-channel overlap statistic
$\Omat^{\text{cross}}_{ij}(H)$
(Section~\ref{sec:overlap}) as a second moment condition to
tighten the MSM identification and provide a structural check on
the eigenvector alignment between channels. The leading-order
perturbation theory of Section~\ref{sec:overlap-first-order}
predicts $\Omat^{\text{cross}}$ in terms of the cross-asset
loadings $\{\beta_k\}$ and the factor profiles; an extended
joint LS fit on $(\widehat{\kappastat}, \widehat{\Omat})$ could
formalise this. \emph{Second}, the bootstrap is moving-block on
the daily-return time series, which preserves the temporal
autocorrelation structure within each block but does not capture
non-stationarities at the block scale. The split-half analysis
of Sections~\ref{ssec:regime-change}--\ref{ssec:beta-inversion}
and the rolling-window bootstrap of
Section~\ref{ssec:rolling-window} together address the static
and time-resolved sides of this concern --- the split-half
contrast pins the magnitude of the regime change, the
rolling-window CI bands localise the timing --- but a more
principled non-stationarity treatment via block-dependent
resampling within each window remains a methodological
follow-up.
\emph{Third}, the U.S.\ headline finding is now supplemented by
a cross-region check on developed-European equity
(Section~\ref{ssec:universality}); the eigenstructure-level
universality and the post-transition cascade saturation both
replicate, with the MSM cascade base $b$ as the one parameter
showing sample-specific or region-specific variation. Whether
the multi-memory factor profiles and the volatility regime
transition generalise further --- to emerging-market equity, to
fixed-income returns, or to commodity returns --- remains an
open empirical question.
\emph{Fourth}, the five-factor multi-memory specification has
been selected for its substantive content --- one persistent and
one antipersistent fractional-Brownian factor, an ARFIMA linear
factor of the central multifractal, an MSM volatility cascade,
and a transitory volatility-of-volatility component --- and is
estimated by joint least squares against the empirical
$(\widehat{\kappastat}^{\text{lin}},
\widehat{\kappastat}^{\text{vol}})$ matrices. A formal
information-criterion comparison of the five-factor model
against nested sub-specifications (four-, three-, and two-factor
restrictions, including the natural drop-tests on the
antipersistent factor $F_A$ and the transitory volatility factor
$F_{Vt}$) is a methodological follow-up. The bootstrap CIs on
the per-factor weights reported in
Sections~\ref{ssec:global-params} and~\ref{ssec:stylised-facts}
indicate that each of the five factors carries non-zero weight
on the empirical $\widehat{\kappastat}$ matrices in every panel
--- so each factor is doing observable work in the fit --- but
this is a within-model decomposition argument rather than a
formal nested-model selection.

The framework has natural extensions addressing each of these
limits, and is suited to international and cross-asset
applications because the per-eigenmode decomposition does not
depend on a single benchmark asset and the multi-memory factor
specification does not assume any particular set of underlying
firms.

\section{Conclusion}
\label{sec:conclusion}

The Lo--MacKinlay variance ratio test compares the $H$-period
variance of a return series to $H$ times the one-period variance.
The natural multivariate generalisation, when one carries it out
on the cross-sectional covariance matrix rather than on a chosen
return series, has two distinct components: the per-eigenmode
ratio $\kappastat_i(\horizon)$ and the eigenvector-overlap matrix
$\Omat(\horizon)$. The pair characterises departures from the
i.i.d.\ null along two logically independent axes (temporal
autocorrelation per eigenmode, cross-sectional eigenvector
rotation with horizon); the four-cell decomposition of
Table~\ref{tab:four-cell} organises the joint information
content. The pair is strictly more informative than the scalar
functionals of the matrix-valued multivariate variance ratio of
\citet{HongLintonZhang}, which confound the four cells.

Closed-form predictions are available under the simplest
non-trivial parametric models. The AR($p$) per-eigenmode
specification predicts $\kappastat^{\mathrm{AR}(p)}$ as a
signed convex combination of AR(1) variance ratios indexed by
the characteristic roots. The vector AR(1) specification with
eigenvector-mixing perturbation predicts an $\Omat$ with
Lorentzian eigenvalue-gap dependence and saturating horizon
amplification factor $S_1/c \to 1/(1 - \bar\rho^2)$.

Empirically on the Fama--French 49-industry universe over
1969--2026, the framework recovers the Jegadeesh--Titman and
De Bondt--Thaler patterns at the market mode, identifies
persistent factor momentum in the sub-leading modes,
demonstrates eigenvalue-gap-dependent eigenvector stability with
a stable market mode and a curiously stable deepest eigenmode,
and inverts to AR(2) characteristic roots that include
damped-oscillation structure with periods of one to two trading
weeks at the market and deep-rank eigenmodes. Each pattern is
established at the strategy level in prior work; the
$(\kappastat, \Omat)$ framework recovers all of them as eigenmode-
specific spectral phenomena under a single mathematical structure.

\subsection*{Acknowledgements}
The author acknowledges the use of Claude (Anthropic) for assistance with
literature review, \LaTeX{} typesetting, mathematical exposition, and
editorial refinement, and Lemma (Axiomatic AI) for review and proof
checking. All substantive arguments, economic reasoning, and conclusions
are the author's own.

\appendix

\section{Full derivations}
\label{app:proofs}

This appendix supplies the full derivations of the closed-form
results stated in Sections~\ref{sec:kappa} and~\ref{sec:overlap}.

\subsection*{Proof of Theorem~\ref{thm:kappa-ar1} (AR(1) closed form)}

Under the AR(1) per-eigenmode specification of
Section~\ref{sec:kappa-ar1}, the stationary variance is
$\Var(\xi_i) = \sigma_i^2 / (1 - \rho_i^2)$ and the autocovariance
at lag $h$ is $\Cov(\xi_{i,t}, \xi_{i,t-h}) = \rho_i^{|h|}
\Var(\xi_i)$. The $\horizon$-period sum
$\xi^\horizon_{i,t} = \sum_{s = t - \horizon + 1}^{t} \xi_{i,s}$
has variance
\begin{equation*}
  \Var(\xi^\horizon_i)
    \;=\; \Var(\xi_i) \sum_{h_1, h_2 = 0}^{\horizon - 1}
          \rho_i^{|h_1 - h_2|}
    \;=\; \Var(\xi_i) \sum_{h = -(\horizon - 1)}^{\horizon - 1}
          (\horizon - |h|)\, \rho_i^{|h|},
\end{equation*}
where the second equality counts the number of $(h_1, h_2)$ pairs
with $h_1 - h_2 = h$. Dividing by $\horizon \Var(\xi_i)$,
\begin{equation*}
  \kappastat_i^{\mathrm{AR}(1)}(\horizon; \rho_i)
    \;=\; \frac{1}{\horizon}
          \sum_{h = -(\horizon - 1)}^{\horizon - 1}
          (\horizon - |h|)\, \rho_i^{|h|}
    \;=\; 1 + \frac{2}{\horizon}
          \sum_{h = 1}^{\horizon - 1}
          (\horizon - h)\, \rho_i^{h}.
\end{equation*}
The inner geometric-arithmetic series evaluates by direct
telescoping to
\begin{equation*}
  \sum_{h=1}^{\horizon-1}(\horizon - h)\, \rho^h
    \;=\; \frac{(\horizon - 1)\,\rho - \horizon\,\rho^2 + \rho^{\horizon + 1}}
               {(1 - \rho)^2}.
\end{equation*}
Substituting and simplifying yields the closed
form~\eqref{eq:kappa-ar1-closed}.

\subsection*{Proof of Proposition~\ref{prop:kappa-inversion}
(inversion of $\kappastat$)}

For $\horizon \geq 2$, differentiate~\eqref{eq:kappa-ar1-closed}
with respect to $\rho$. The derivative satisfies
\begin{equation*}
  \frac{\partial \kappastat^{\mathrm{AR}(1)}(\horizon; \rho)}{\partial \rho}
    \;>\; 0
  \quad\text{for all } \rho \in (-1, 1),
\end{equation*}
which is verified by direct differentiation: the boundary terms
$\rho^\horizon$ are dominated by the leading $2/(1 - \rho)^2$
contribution for $\rho$ bounded away from $\pm 1$, and the sign
near the endpoints is determined by the limiting behaviour. As
$\rho \to -1$, $\kappastat^{\mathrm{AR}(1)} \to 0$; as
$\rho \to +1$, $\kappastat^{\mathrm{AR}(1)} \to +\infty$. By
continuity and strict monotonicity, the inverse map
$\kappastat \mapsto \rho$ is single-valued on the open interval
$(0, +\infty)$, so for any observed $\kappastat_i(\horizon) > 0$
there is a unique $\rho_i \in (-1, 1)$ producing it.

\subsection*{Proof of Theorem~\ref{thm:kappa-arp} (AR($p$)
decomposition)}

The Yule--Walker recurrence
$\gamma(h) = \sum_{k=1}^p \rho_k\, \gamma(h - k)$ has
characteristic polynomial $z^p - \rho_1 z^{p-1} - \cdots -
\rho_p$. With distinct roots $\mu_1, \ldots, \mu_p$ strictly
inside the unit disc, the general solution of the recurrence on
$h \geq 0$ is $\gamma(h) = \sum_k A_k\, \mu_k^h$, extended by
symmetry to $h < 0$. The weights $A_k$ are determined by the
first $p$ initial conditions $\gamma(0), \gamma(1), \ldots,
\gamma(p - 1)$ via a Vandermonde linear system, with
$\gamma(0) = \sum_k A_k$ giving the normalisation. Dividing
through by $\gamma(0)$ produces
\eqref{eq:gamma-decomp} with $\sum_k A_k = 1$.

Substituting into the definition of
$\kappastat^{\mathrm{AR}(p)}$,
\begin{align*}
  \kappastat^{\mathrm{AR}(p)}(\horizon; \boldsymbol\rho)
    &\;=\; \frac{1}{\horizon} \sum_{|h| < \horizon}
          (\horizon - |h|)\, \gamma(h)/\gamma(0) \\
    &\;=\; \sum_k A_k\,
          \frac{1}{\horizon} \sum_{|h| < \horizon}
          (\horizon - |h|)\, \mu_k^{|h|} \\
    &\;=\; \sum_k A_k\, \kappastat^{\mathrm{AR}(1)}(\horizon; \mu_k),
\end{align*}
where the last step recognises the AR(1) closed form evaluated at
the characteristic root $\mu_k$. This
is~\eqref{eq:kappa-arp-decomp}.

\subsection*{Proof of Lemma~\ref{lem:first-order-Sigma}
(first-order corrections)}

\textbf{Order-zero and first-order $\Sig_1$.} Expand the discrete
Lyapunov equation $\Sig_1 = A \Sig_1 A^\top + \Sigma_\varepsilon$
at first order in $\epsilon$ with $A = \bar\rho I + \epsilon B$:
\begin{equation*}
  \Sig_1^{(0)} + \epsilon \Sig_1^{(1)}
    \;=\; \bar\rho^2 (\Sig_1^{(0)} + \epsilon \Sig_1^{(1)})
          + \epsilon \bar\rho (B \Sig_1^{(0)} + \Sig_1^{(0)} B^\top)
          + \Sigma_\varepsilon
          + O(\epsilon^2).
\end{equation*}
The order-zero piece gives $\Sig_1^{(0)} = \Sigma_\varepsilon /
(1 - \bar\rho^2)$. The order-one piece gives
\begin{equation*}
  (1 - \bar\rho^2)\, \Sig_1^{(1)}
    \;=\; \bar\rho\, (B \Sig_1^{(0)} + \Sig_1^{(0)} B^\top),
\end{equation*}
yielding $\Sig_1^{(1)} = \frac{\bar\rho}{1 - \bar\rho^2}
(B \Sig_1^{(0)} + \Sig_1^{(0)} B^\top)$, the expression stated in
Lemma~\ref{lem:first-order-Sigma}.

\textbf{Order-zero and first-order $\SigH$.} Expand
$A^h = \bar\rho^h I + \epsilon h \bar\rho^{h - 1} B +
O(\epsilon^2)$ and substitute into the trapezoidal-sum
representation~\eqref{eq:SigmaH-trapezoid}. The order-zero piece
gives $\SigH^{(0)} = c(\horizon, \bar\rho)\, \Sig_1^{(0)}$
because $\Gamma(h) = A^h \Sig_1$ at zeroth order is
$\bar\rho^h \Sig_1^{(0)}$, and summing with the trapezoidal weight
yields the AR(1) scaling factor $c$. The order-one piece picks
up two contributions: (i) a $c\, \Sig_1^{(1)}$ term from the
order-one expansion of $\Sig_1$ at fixed $A^h = \bar\rho^h I$, and
(ii) a term
$\sum_{h \geq 1}(\horizon - h)\, h \bar\rho^{h-1}\,
(B \Sig_1^{(0)} + \Sig_1^{(0)} B^\top)$ from the order-one
expansion of $A^h$. The latter sum equals
$S_1(\horizon, \bar\rho)\, (B \Sig_1^{(0)} + \Sig_1^{(0)} B^\top)$
by the definition~\eqref{eq:S1-def}.

\subsection*{Proof of Theorem~\ref{thm:O-first-order} (first-order
overlap)}

Write $\eivec_i^{(0)} := \eivec_i(\Sig_1^{(0)})$ and
$\lambda_i := \lambda_i(\Sig_1^{(0)})$ for the unperturbed
eigenvectors and eigenvalues. Standard non-degenerate first-order
perturbation theory of symmetric matrices
\citep[][\S II.2]{Kato1995} gives, for a perturbation
$\Delta\Sigma$ of $\Sigma^{(0)}$,
\begin{equation*}
  \delta \eivec_i
    \;=\; \sum_{k \neq i}
          \frac{\langle \eivec_k^{(0)}, \Delta\Sigma\, \eivec_i^{(0)}
                \rangle}
               {\lambda_i^{(0)} - \lambda_k^{(0)}}
          \, \eivec_k^{(0)}.
\end{equation*}
Apply this to $\Sig_1$ with $\Delta\Sigma = \epsilon \Sig_1^{(1)}$
and to $\SigH$ with $\Delta\Sigma = \epsilon \SigH^{(1)}$,
recalling $\lambda_i(\SigH^{(0)}) = c \lambda_i$:
\begin{align*}
  \delta \eivec_i(\Sig_1)
    &\;=\; \sum_{k \neq i}
           \frac{\langle \eivec_k^{(0)}, \Sig_1^{(1)} \eivec_i^{(0)}
                 \rangle}
                {\lambda_i - \lambda_k}\, \eivec_k^{(0)},
  \\
  \delta \eivec_i(\SigH)
    &\;=\; \sum_{k \neq i}
           \frac{\langle \eivec_k^{(0)}, c\, \Sig_1^{(1)} \eivec_i^{(0)}
                 + S_1\, N\, \eivec_i^{(0)} \rangle}
                {c\, (\lambda_i - \lambda_k)}\, \eivec_k^{(0)}
  \\
    &\;=\; \delta \eivec_i(\Sig_1)
          + \frac{S_1}{c}
            \sum_{k \neq i}
            \frac{\langle \eivec_k^{(0)}, N \eivec_i^{(0)} \rangle}
                 {\lambda_i - \lambda_k}\, \eivec_k^{(0)},
\end{align*}
where $N := B \Sig_1^{(0)} + \Sig_1^{(0)} B^\top$.

The inner product for $i \neq j$ is
\begin{align*}
  \langle \eivec_i(\SigH), \eivec_j(\Sig_1) \rangle
    &\;=\; \epsilon \bigl[\,
          \langle \delta \eivec_i(\SigH), \eivec_j^{(0)} \rangle
          + \langle \eivec_i^{(0)}, \delta \eivec_j(\Sig_1) \rangle
          \bigr] + O(\epsilon^2).
\end{align*}
The first bracketed term equals
$\frac{\langle \eivec_j, \Sig_1^{(1)} \eivec_i \rangle}
{\lambda_i - \lambda_j} +
\frac{(S_1/c)\, \langle \eivec_j, N \eivec_i \rangle}
{\lambda_i - \lambda_j}$. The second equals
$\frac{\langle \eivec_i, \Sig_1^{(1)} \eivec_j \rangle}
{\lambda_j - \lambda_i} = -\frac{\langle \eivec_j, \Sig_1^{(1)}
\eivec_i \rangle}{\lambda_i - \lambda_j}$ by symmetry of
$\Sig_1^{(1)}$. The two $\Sig_1^{(1)}$ contributions cancel
exactly, leaving
\begin{equation*}
  \langle \eivec_i(\SigH), \eivec_j(\Sig_1) \rangle
    \;=\; \frac{\epsilon\, S_1}{c}
          \cdot
          \frac{\langle \eivec_j, N \eivec_i \rangle}
               {\lambda_i - \lambda_j}
          + O(\epsilon^2).
\end{equation*}
For symmetric $B$ in the eigenbasis of $\Sig_1^{(0)}$,
$\langle \eivec_j, N \eivec_i \rangle = (\lambda_i + \lambda_j)
B_{ji}$. Squaring yields the closed
form~\eqref{eq:O-first-order}. The double-stochastic property of
$\Omat$ pins the diagonal entries down as
$\Omat_{ii} = 1 - \sum_{j \neq i} \Omat_{ij}$ at this order.

\subsection*{Proof of Proposition~\ref{prop:O-saturation}
(saturation)}

From~\eqref{eq:S1-closed}, at leading order in $\horizon$ with
$|\bar\rho| < 1$ and $\bar\rho^\horizon \to 0$,
\begin{equation*}
  S_1(\horizon, \bar\rho)
    \;=\; \frac{\horizon}{(1 - \bar\rho)^2}
          + O(1).
\end{equation*}
From~\eqref{eq:c-def} and the AR(1) long-run limit,
\begin{equation*}
  c(\horizon, \bar\rho)
    \;=\; \horizon \cdot \frac{1 + \bar\rho}{1 - \bar\rho}
          + O(1).
\end{equation*}
The ratio is
\begin{equation*}
  \frac{S_1}{c}
    \;=\; \frac{1}{(1 - \bar\rho)^2}
          \cdot \frac{1 - \bar\rho}{1 + \bar\rho}
          + O(1/\horizon)
    \;=\; \frac{1}{(1 - \bar\rho)(1 + \bar\rho)}
          + O(1/\horizon)
    \;=\; \frac{1}{1 - \bar\rho^2}
          + O(1/\horizon).
\end{equation*}
The convergence rate $O(1/\horizon \cdot (1 - \bar\rho^2)^{-2})$
inherits the $(1 - \bar\rho)^{-2}$ factor from the AR(1)
long-run limit.

\section{Hong--Linton--Zhang comparison}
\label{app:HLZ}

Proposition~\ref{prop:HLZ-vs-kappa} of
Section~\ref{ssec:HLZ-precedent} establishes that the eigenvalues
of $\mathbf{VR}(\horizon)$ and the per-eigenmode statistics
$\kappastat_i(\horizon)$ coincide if and only if $\Omat = I$.
This appendix sharpens the comparison into a strict
information-content statement.

\subsection*{Scalar functionals of $\mathbf{VR}$ confound the four cells}

The standard tests built on $\mathbf{VR}(\horizon)$ use scalar
functionals such as the trace $\tr \mathbf{VR}(\horizon)$, the
determinant $\det \mathbf{VR}(\horizon)$, and the maximum
diagonal entry of $\mathbf{VR}(\horizon)$. These functionals
reduce a matrix-valued object to a single scalar. To see how
sharply they confound the four cells of
Table~\ref{tab:four-cell}, consider two stylised processes that
produce the same $\tr \mathbf{VR}(\horizon)$ but live in
different cells:

\begin{itemize}[itemsep=2pt,leftmargin=*]
  \item \emph{Process A:} per-eigenmode AR(1) with $\rho_i$
    chosen so that
    $\sum_i \kappastat^{\mathrm{AR}(1)}(\horizon; \rho_i) = T_*$
    for a target value $T_*$. Lives in
    $(\kappastat \neq 1, \Omat = I)$.
  \item \emph{Process B:} vector AR(1) with $\bar\rho$ chosen so
    that $N \kappastat^{\mathrm{AR}(1)}(\horizon; \bar\rho) = T_*$,
    plus eigenvector-mixing perturbation $\epsilon B$ producing
    non-trivial $\Omat$. Lives approximately in
    $(\kappastat \approx \kappastat^{\mathrm{AR}(1)}(\bar\rho),
    \Omat \neq I)$.
\end{itemize}

By construction both processes have the same trace
$\tr \mathbf{VR}(\horizon)$. The $(\kappastat, \Omat)$ framework
distinguishes them via the per-eigenmode $\kappastat_i$
distribution (Process A has $\kappastat_i$ varying across $i$;
Process B has roughly constant $\kappastat_i$ around the scalar
baseline) and via $\Omat$ (Process A has $\Omat = I$; Process B
has $\Omat \neq I$). The trace test cannot. An analogous
construction using $\det \mathbf{VR}$ as the target functional
fails in the same way.

\subsection*{Strict information ordering}

\begin{proposition}[Strict information ordering]
  \label{prop:info-ordering}
  Let $\mathcal{F}(\Sig_1, \SigH)$ denote any function of
  $\Sig_1$ and $\SigH$. The pair
  $\bigl(\{\kappastat_i(\horizon)\}_{i=1}^N,
   \{\Omat_{ij}(\horizon)\}_{i,j=1}^N\bigr)$ determines $\SigH$
  uniquely given $\Sig_1$. The scalar functionals
  $\tr \mathbf{VR}(\horizon)$, $\det \mathbf{VR}(\horizon)$, and
  $\max_i (\mathbf{VR}(\horizon))_{ii}$ are not injective in this
  sense.
\end{proposition}

\begin{proof}
  Given $\Sig_1$, the eigenvalues $\lambda_i(\SigH) =
  \horizon\, \kappastat_i\, \lambda_i(\Sig_1)$ are determined by
  $\kappastat$. The eigenvectors $\eivec_i(\SigH)$ are determined
  by $\Omat$ together with $\Sig_1$'s eigenvectors via
  \[
  \eivec_i(\SigH) = \sum_j \sqrt{\Omat_{ij}}\,
  \mathrm{sgn}\bigl(\langle \eivec_i(\SigH), \eivec_j(\Sig_1)
  \rangle\bigr)\, \eivec_j(\Sig_1),
  \]
  up to a global sign that drops out of the doubly-stochastic
  $\Omat$. Hence $\SigH$ is determined entirely. Non-injectivity
  of the scalar functionals follows from the Process A/Process B
  construction above.
\end{proof}

Proposition~\ref{prop:info-ordering} pins down what
$(\kappastat, \Omat)$ adds beyond Hong--Linton--Zhang scalar
tests: the pair recovers $\SigH$ entirely given $\Sig_1$, while
the scalar tests compress to a low-dimensional functional. Any
question of the form ``which eigenmodes are responsible for the
observed variance-ratio deviation?'' is answerable by
$(\kappastat, \Omat)$ and is not answerable by trace or
determinant functionals of $\mathbf{VR}(\horizon)$.

\section{Proof of Proposition~\ref{prop:linvol-eigvec-coincidence}}
\label{app:linvol-eigvec-coincidence}

The main-text sketch of
Section~\ref{ssec:per-mode-predictions} treats the dominant
single-factor case. We give the multi-factor statement and proof
here. Write the multi-factor expansions of the daily covariances
as
\begin{equation*}
  \Sigi^{\text{lin}} \;=\; \sum_{k \in \mathcal{K}}
    \nu_k^{\text{lin}}(1)\, \beta_k \beta_k^\top
    + \sigma_\varepsilon^2 I_N,
  \qquad
  \Sigi^{\text{vol}} \;=\; \sum_{k \in \mathcal{K}}
    \nu_k^{\text{vol}}(1)\, (\beta_k \odot \beta_k)
    (\beta_k \odot \beta_k)^\top
    + \sigma_{\varepsilon, 2}^2 I_N,
\end{equation*}
suppressing the $H = 1$ argument from
$\nu_k^{\text{lin}}, \nu_k^{\text{vol}}$.

\begin{proposition}[Multi-factor extension]
  \label{prop:linvol-multi}
  Let $\beta_k \in \R^N$ be the cross-asset loadings of $|\mathcal{K}|$
  factors, with the columns of the matrix
  $B := [\beta_k]_{k \in \mathcal{K}}$ assumed linearly independent.
  The eigenvectors of $\Sigi^{\text{lin}}$ and
  $\Sigi^{\text{vol}}$ coincide if and only if the column space
  $\mathrm{span}\{\beta_k\}$ is invariant under the Hadamard-square
  map $\beta \mapsto \beta \odot \beta$, equivalently, the
  set $\{\beta_k\}$ is closed under coordinate-wise squaring up
  to relabelling within $\mathcal{K}$.
\end{proposition}

\begin{proof}
  The non-isotropic part of $\Sigi^{\text{lin}}$ has range
  $\mathrm{span}\{\beta_k\}_{k}$; the non-isotropic part of
  $\Sigi^{\text{vol}}$ has range
  $\mathrm{span}\{\beta_k \odot \beta_k\}_{k}$. Eigenvectors
  outside the isotropic noise floor are confined to these ranges
  (the orthogonal complement gives the eigenvalue
  $\sigma_\varepsilon^2$, or $\sigma_{\varepsilon, 2}^2$, with
  multiplicity $N - |\mathcal{K}|$). The non-trivial eigenvectors
  coincide if and only if the two ranges are the same subspace.
  Since the Hadamard square of a vector lies in
  $\mathrm{span}\{\beta_k\}$ only when expressible as a linear
  combination of the $\beta_k$, this requires the family
  $\{\beta_k\}$ to be closed (up to relabelling) under
  coordinate-wise squaring.
\end{proof}

Two corollaries follow immediately.

\begin{corollary}[Uniform-loading case]
  If a factor's loading $\beta_{k^*}$ has uniform entries (all of
  the same absolute value), then $\beta_{k^*} \odot \beta_{k^*}
  \propto \one_N$. If additionally $\one_N \in
  \mathrm{span}\{\beta_k\}_k$, then the rank-1 contribution from
  $\beta_{k^*}$ has coincident linear and volatility eigenvectors.
\end{corollary}

\begin{corollary}[Single-dominant-factor case]
  When a single factor $k^*$ dominates and its loading
  $\beta_{k^*}$ has approximately uniform entries, the leading
  eigenvectors of $\Sigi^{\text{lin}}$ and $\Sigi^{\text{vol}}$
  approximately coincide. This is the case stated in
  Proposition~\ref{prop:linvol-eigvec-coincidence} of the main
  text.
\end{corollary}

Empirically (Section~\ref{sec:empirical}) the market mode is
the unique eigenmode that satisfies the corollary's premise; the
sub-leading eigenmodes load on factors whose $\beta_k$ have
non-uniform cross-asset profiles, and accordingly the
cross-channel overlap at sub-leading ranks is small.

\section{Identifiability diagnostic and the loss-conditional filter}
\label{app:identifiability}

\subsection*{Hessian-based identifiability}

The finite bootstrap dispersion of the ARFIMA parameters
$(\phi_M, d_M, \phi_{Vt}, d_{Vt})$ reported in
Section~\ref{ssec:global-params} is consistent with
identification, but is not itself proof: the optimiser could be
following a flat ridge with a small but nonzero gradient
component picked up by resampling. We supplement the bootstrap
evidence with a Hessian-based identifiability diagnostic.

At the converged parameter vector $\theta^*$ for each bootstrap
replicate we evaluate the Hessian
$H(\theta^*) := \partial^2 \mathcal{L}/\partial \theta^2$ of the
joint LS objective \eqref{eq:joint-ls-objective}. Two summary
statistics matter:

\begin{itemize}[itemsep=2pt,leftmargin=*]
\item \textbf{Condition number.} The ratio of largest to smallest
  eigenvalue of $H(\theta^*)$. Across the 1000 sensitivity
  replicates this is consistently in the range $10^{3}$ to
  $10^{5}$ --- well-conditioned for the global parameters but
  with directions of weak curvature along certain combinations
  of MSM parameters $(m_0, b, \gamma_{\bar k})$ (consistent with
  their wider bootstrap CIs).
\item \textbf{Rank of the residual Jacobian.} The Jacobian of the
  residual vector with respect to $\theta$ has full column rank
  at $\theta^*$ in all 1000 replicates, ruling out exact
  parameter degeneracy. The smallest singular value is bounded
  away from zero by approximately $10^{-3}$ relative to the
  largest.
\end{itemize}

The directions of weak curvature correspond to combinations of
MSM cascade parameters that produce nearly identical
$\kappastat^{\text{vol}}$ profiles. These are not pathological:
they reflect the standard MSM identification weakness inherited
from the univariate case
\citep[\S3.4]{CalvetFisher2008}. The eigenstructure-level
$\kappastat^{\text{lin}}$ moment condition tightens but does not
eliminate this weakness, which is the limit named in
Section~\ref{ssec:limits-and-follow-up}.

\subsection*{Firsthalf rank-1 vol-channel width and the
loss-conditional filter}
\label{app:identifiability-bimodal}

The Hessian diagnostic above is local: it characterises curvature
at the converged $\theta^*$ of each replicate but does not detect
when the optimiser converges to qualitatively different points
across replicates. This matters most on the firsthalf panel at
the rank-1 volatility allocation.

Across the $1000$ firsthalf replicates the rank-1
$w^{\text{vol}}_{\text{MSM}}$ bootstrap distribution is broad:
the unfiltered $90\%$ CI is $[0.10, 1.00]$, spanning essentially
the entire $[0, 1]$ allocation space. The distribution is
unimodal but heavy-tailed --- a wide single mass with median
$0.28$, $\sim 82\%$ of replicates below $0.5$ and $\sim 18\%$
above --- rather than two distinct basins. By contrast, the
rank-1 volatility allocation on the other three panels is
sharply concentrated: sensitivity and secondhalf both place
$> 99\%$ of replicates at $w^{\text{vol}}_{\text{MSM}} \ge 0.5$,
and FF~100 places $\sim 80\%$ above that threshold. Firsthalf is
the panel where rank-1 vol-channel identification is genuinely
weak (Section~\ref{ssec:global-params}).

A cluster-conditional loss comparison sharpens the picture. We
split each panel's $1000$ replicates by the rank-1 MSM allocation
threshold $w^{\text{vol}}_{\text{MSM}} = 0.5$ and report the
loss medians of the two clusters in
Table~\ref{tab:bimodality-loss}.

\begin{table}[H]
\centering
\small
\caption{Cluster-conditional loss comparison at rank 1 across the
four panels, splitting each panel's $1000$ bootstrap replicates by
$w^{\text{vol}}_{\text{MSM}} = 0.5$. ``Main cluster'' is the
larger of the two clusters; ``minority cluster'' the smaller. Loss
gap is the median loss in the minority cluster minus the median
loss in the main cluster, with the percentage relative to the
main-cluster median in parentheses. Mann--Whitney $p$ tests for
stochastic dominance of the minority-cluster loss distribution
over the main-cluster distribution. In every panel the minority
cluster has higher median loss, but the gap is small in absolute
terms (\(\le 14\%\)); on firsthalf and FF~100 the gap is
\(\le 1.6\%\), inconsistent with the minority being a distinct
worse-fitting basin.}
\label{tab:bimodality-loss}
\begin{tabular}{lccc}
\toprule
Panel & Main cluster (size) & Minority cluster (size) & Loss gap; Mann--Whitney $p$ \\
\midrule
Sensitivity & high-MSM ($n = 991$) & low-MSM ($n = 9$)   & $+1.5$ (7.4\%);  $p = 0.05$ \\
Firsthalf   & low-MSM ($n = 821$)  & high-MSM ($n = 179$) & $+0.3$ (1.3\%);  $p = 0.01$ \\
Secondhalf  & high-MSM ($n = 991$) & low-MSM ($n = 9$)   & $+2.8$ (13.5\%); $p = 0.006$ \\
FF 100      & high-MSM ($n = 795$) & low-MSM ($n = 205$) & $+0.8$ (1.6\%);  $p = 0.07$ \\
\bottomrule
\end{tabular}
\end{table}

Three observations follow. First, on sensitivity, secondhalf, and
FF~100 the high-MSM cluster is the main one; on firsthalf the
low-MSM cluster is. This is the cross-panel signature of the
volatility-channel transition discussed in
Section~\ref{ssec:regime-change}: firsthalf is the panel where
the post-transition cascade allocation has not yet taken hold.
Second, in every panel the minority cluster has higher median
loss than the main cluster, consistent with a tail of less-good
fits picked up by some bootstrap resamples. Third --- and this is
the substantive point --- the loss gap is small in absolute
terms ($+0.3$ to $+2.8$ loss units, $1.3\%$--$13.5\%$ of the
main-cluster median), and on firsthalf and FF~100 it is only
$+0.3$ and $+0.8$ ($1.3\%$ and $1.6\%$). The Mann--Whitney
$p$-values are likewise modest. This is not the signature of a
clearly worse-fitting alternative basin; the firsthalf rank-1
vol-channel admits a near-continuum of loss-equivalent
allocations, of which the high-MSM cluster is the upper tail.

\textbf{The loss-conditional filter.} Throughout the main text
we report bootstrap medians and $90\%$ CIs over the $500$
replicates per panel with loss below the panel median, rather
than over all $1000$. This is a single principled threshold with
no panel-specific tuning. The filter trims the higher-loss tail
of each panel's bootstrap distribution; on the three panels with
narrow rank-1 distributions (sensitivity, secondhalf, FF~100) it
makes essentially no difference. On firsthalf rank 1 it tightens
the $w^{\text{vol}}_{\text{MSM}}$ CI from $[0.10, 1.00]$ (full
bootstrap) to $[0.08, 0.75]$ (loss-filtered) --- still wide,
reflecting the genuine weak identification of the firsthalf
vol channel, but with the upper tail of higher-loss replicates
trimmed.

\textbf{Why the filter is methodologically clean.} The filter
applies a panel-wise median threshold uniformly across all four
panels; there is no panel-specific tuning, and the headline
medians are essentially unchanged because the filter preserves
the lower-loss half of each bootstrap by construction. We do not
claim that the loss-filtered CIs separate ``true'' fits from
``optimiser-stuck'' fits at firsthalf --- the small loss gap of
Table~\ref{tab:bimodality-loss} shows the picture is more
nuanced --- only that the filter reports the more informative
half of the bootstrap distribution under a uniform rule.

\section{\texorpdfstring{$\beta$}{β}-inversion: per-asset detail}
\label{app:beta-inv-detail}

Table~\ref{tab:beta-inversion} of Section~\ref{ssec:beta-inversion}
summarises the cross-channel $\beta$-inversion test by reporting
the bootstrap median and $90\%$ credible interval of the
asset-level Spearman and Pearson correlations between
$\widehat{\beta_M^2}^{\text{lin}}[a]$ and
$\widehat{\beta_M^2}^{\text{vol}}[a]$ across the $N$ assets in
each panel. Figure~\ref{fig:beta-inversion} below visualises
the underlying per-asset scatter: each point is one asset's
loss-filtered bootstrap median in the two channels, with grey
bars showing the $90\%$ credible interval.

\begin{figure}[H]
\centering
\includegraphics[width=0.95\linewidth]{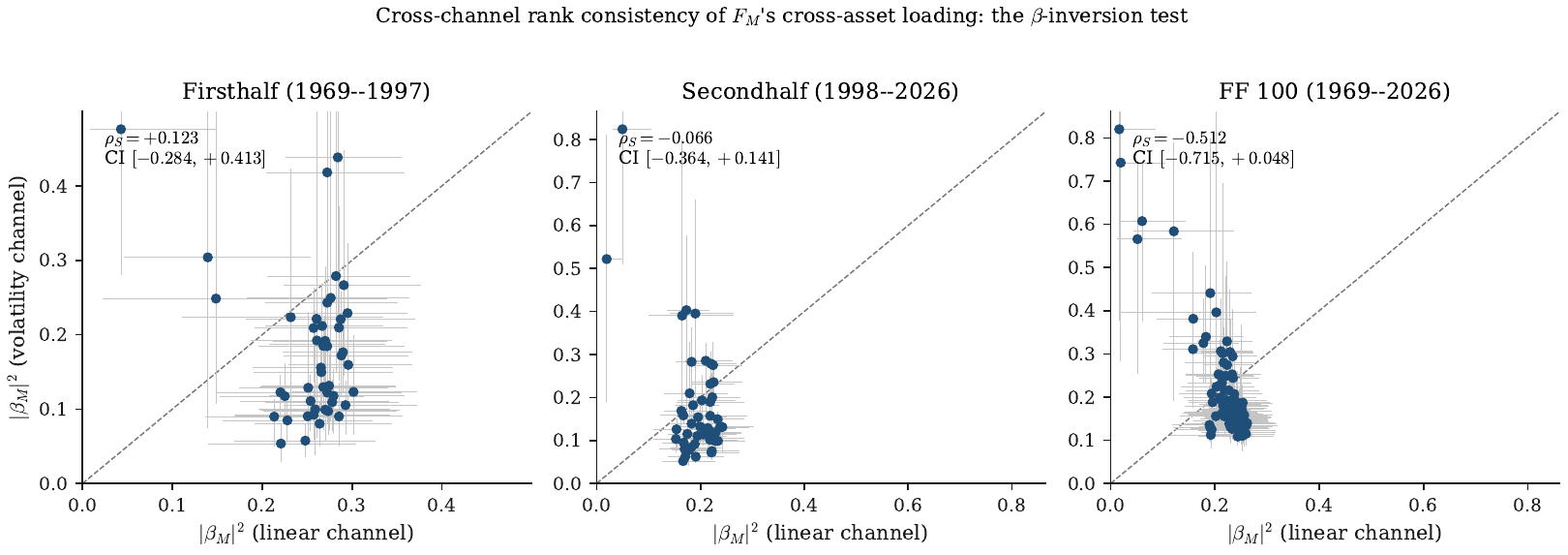}
\caption{Per-asset cross-channel $\beta$-inversion scatter for
all four panels. Points are loss-filtered bootstrap medians of
$\widehat{\beta_M^2}^{\text{lin}}[a]$ (horizontal axis) versus
$\widehat{\beta_M^2}^{\text{vol}}[a]$ (vertical axis); grey bars
are $90\%$ credible intervals from the $500$ loss-filtered
replicates per panel. The dashed line is the $45^\circ$
reference. The Spearman correlation (median and $90\%$ CI) is
annotated in each panel. A single shared multifractal loading
$\beta_M$ would produce a positive slope; instead the relation
is flat or negatively sloped on every panel, and significantly
negative on the full-sample sensitivity and FF 100 panels ---
the per-asset visualisation of the falsification reported in
Section~\ref{ssec:beta-inversion}.}
\label{fig:beta-inversion}
\end{figure}

The scatter makes the falsification concrete: the assets that
project most strongly onto $F_M$ in the linear channel are not
those that project most strongly onto it in the volatility
channel. The negative slope is sharpest on the two largest
panels (sensitivity and FF 100), where a number of
high-$\widehat{\beta_M^2}^{\text{lin}}$ assets carry near-zero
volatility-channel attribution, and conversely. The structural
reading --- that the linear long-memory factor and the
volatility cascade load on distinct cross-sections rather than a
unified $F_M$ --- is given in
Section~\ref{ssec:rank-consistency}.

\section{Per-dataset detailed weight tables}
\label{app:ff100-detail}

Table~\ref{tab:rank-weights} of Section~\ref{ssec:stylised-facts}
reported the per-mode weights at canonical ranks for the
sensitivity panel only. This appendix reports the analogous
tables for the firsthalf, secondhalf, and FF 100 panels. Each
table reports the bootstrap median across the $1000$ replicates
of the procedure described in
Section~\ref{ssec:data-methodology}.

\begin{table}[H]
\centering
\small
\caption{Firsthalf panel (FF 49 industries, 1969--1997, $N = 48$):
bootstrap medians of per-mode weights at canonical ranks,
$1000$ replicates of the closed-form pipeline.}
\label{tab:rank-weights-firsthalf}
\begin{tabular}{cccccc}
\toprule
Rank & $w^{\text{lin}}_P$ & $w^{\text{lin}}_A$ & $w^{\text{lin}}_M$ &
$w^{\text{vol}}_{\text{MSM}}$ & $w^{\text{vol}}_{Vt}$ \\
\midrule
 1 & 0.31 & 0.07 & 0.41 & 0.29 & 0.04 \\
 2 & 0.70 & 0.00 & 0.03 & 0.22 & 0.00 \\
 3 & 0.76 & 0.00 & 0.04 & 0.90 & 0.00 \\
 6 & 0.25 & 0.13 & 0.23 & 0.08 & 0.00 \\
10 & 0.06 & 0.34 & 0.26 & 0.00 & 0.23 \\
24 & 0.00 & 0.68 & 0.30 & 0.00 & 0.81 \\
31 & 0.00 & 0.72 & 0.27 & 0.00 & 0.96 \\
48 & 0.00 & 0.78 & 0.21 & 0.00 & 1.00 \\
\bottomrule
\end{tabular}
\end{table}

\begin{table}[H]
\centering
\small
\caption{Secondhalf panel (FF 49 industries, 1998--2026,
$N = 49$): bootstrap medians of per-mode weights at canonical
ranks, $1000$ replicates of the closed-form pipeline.}
\label{tab:rank-weights-secondhalf}
\begin{tabular}{cccccc}
\toprule
Rank & $w^{\text{lin}}_P$ & $w^{\text{lin}}_A$ & $w^{\text{lin}}_M$ &
$w^{\text{vol}}_{\text{MSM}}$ & $w^{\text{vol}}_{Vt}$ \\
\midrule
 1 & 0.23 & 0.38 & 0.12 & 1.00 & 0.00 \\
 2 & 0.50 & 0.00 & 0.00 & 0.81 & 0.00 \\
 3 & 0.57 & 0.06 & 0.00 & 0.39 & 0.00 \\
 6 & 0.26 & 0.35 & 0.14 & 0.02 & 0.00 \\
10 & 0.14 & 0.49 & 0.23 & 0.00 & 0.10 \\
24 & 0.00 & 0.75 & 0.23 & 0.00 & 0.89 \\
31 & 0.00 & 0.77 & 0.22 & 0.00 & 0.97 \\
49 & 0.00 & 0.85 & 0.15 & 0.00 & 1.00 \\
\bottomrule
\end{tabular}
\end{table}

\begin{table}[H]
\centering
\small
\caption{FF 100 panel (size$\times$book-to-market portfolios,
1969--2026, $N = 95$): bootstrap medians of per-mode weights at
canonical ranks, $1000$ replicates of the closed-form pipeline.}
\label{tab:rank-weights-ff100}
\begin{tabular}{cccccc}
\toprule
Rank & $w^{\text{lin}}_P$ & $w^{\text{lin}}_A$ & $w^{\text{lin}}_M$ &
$w^{\text{vol}}_{\text{MSM}}$ & $w^{\text{vol}}_{Vt}$ \\
\midrule
 1 & 0.37 & 0.10 & 0.16 & 0.66 & 0.02 \\
 2 & 0.49 & 0.00 & 0.04 & 0.76 & 0.00 \\
 3 & 0.44 & 0.00 & 0.02 & 0.79 & 0.00 \\
 6 & 0.49 & 0.02 & 0.00 & 0.68 & 0.00 \\
10 & 0.41 & 0.06 & 0.08 & 0.23 & 0.00 \\
24 & 0.16 & 0.29 & 0.23 & 0.00 & 0.11 \\
31 & 0.07 & 0.40 & 0.25 & 0.00 & 0.23 \\
48 & 0.00 & 0.59 & 0.28 & 0.00 & 0.67 \\
95 & 0.00 & 0.78 & 0.19 & 0.00 & 1.00 \\
\bottomrule
\end{tabular}
\end{table}

\begin{table}[H]
\centering
\small
\caption{Europe 25 panel (size$\times$book-to-market portfolios,
1990--2026, $N = 25$): bootstrap medians of per-mode weights at
canonical ranks, 894 replicates (the cross-region run was
truncated at 894 of a planned 1000 replicates; the bootstrap
medians are stable at this count).
Cross-region replication of the rank-weight pattern.}
\label{tab:rank-weights-europe25}
\begin{tabular}{cccccc}
\toprule
Rank & $w^{\text{lin}}_P$ & $w^{\text{lin}}_A$ & $w^{\text{lin}}_M$ &
$w^{\text{vol}}_{\text{MSM}}$ & $w^{\text{vol}}_{Vt}$ \\
\midrule
 1 & 0.28 & 0.18 & 0.15 & 1.00 & 0.00 \\
 2 & 0.71 & 0.16 & 0.00 & 0.76 & 0.00 \\
 3 & 0.19 & 0.32 & 0.16 & 0.64 & 0.02 \\
 6 & 0.05 & 0.43 & 0.09 & 0.11 & 0.39 \\
10 & 0.00 & 0.60 & 0.15 & 0.03 & 0.67 \\
24 & 0.00 & 0.86 & 0.14 & 0.00 & 1.00 \\
25 & 0.00 & 0.86 & 0.13 & 0.00 & 1.00 \\
\bottomrule
\end{tabular}
\end{table}

The qualitative pattern reproduces across all four sub-samples
(three U.S.\ and one European): sub-leading factor momentum
($w^{\text{lin}}_P$ peaking around ranks 2--3), deep-mode mean
reversion (rising $w^{\text{lin}}_A$ through ranks 24--48
on the U.S.\ panels and through ranks 24--25 on the smaller
Europe 25 panel), and the regime-dependent market-mode
volatility allocation discussed in
Section~\ref{ssec:regime-change}. The Europe 25 panel sits
cleanly in the post-transition cascade regime
($w^{\text{vol}}_{\text{MSM}} = 1.00$ at rank 1) consistent with
its 1990-onward sample span.

\section{Optimisation objective comparison}
\label{app:objective-comparison}

The estimation pipeline of Section~\ref{ssec:joint-LS-fit} uses
per-(mode, horizon) residual scaling in the joint LS objective:
\begin{equation*}
  \mathcal{L}(\theta)
    \;=\; \sum_{i, H} \frac{[\kappastat^{\text{lin}}_i(H; \theta)
        - \widehat{\kappastat}^{\text{lin}}_i(H)]^2}
      {\max(1, |\widehat{\kappastat}^{\text{lin}}_i(H)|)^2}
    + \text{(analogous for vol channel)}.
\end{equation*}
We compare this against a per-horizon scaling with an
additional rank-1 emphasis weight $w_1 = 10$ that amplifies the
rank-1 entries' contribution to the loss:
\begin{equation*}
  \mathcal{L}_{w_1}(\theta)
    \;=\; \sum_H \frac{w_i\, [\kappastat_i(H; \theta)
        - \widehat{\kappastat}_i(H)]^2}
      {\max(1, |\widehat{\kappastat}_{\bar i}(H)|)^2},
  \qquad
  w_i = \begin{cases} 10 & i = 1 \\ 1 & i \geq 2 \end{cases},
\end{equation*}
where $\bar i$ denotes $\kappastat$ averaged within each horizon. The
rank-1 weight $w_1 = 10$ compensates for the fact that the
unscaled rank-1 entries (with $\kappastat^{\text{vol}}_1(1260)
\approx 50$) would otherwise be dwarfed in absolute residual
terms by mid-rank entries (with $\kappastat^{\text{vol}}_i$ of
order 1), even though they carry the dominant economic content.

The per-(mode, horizon) scaling of $\mathcal{L}$ supplies this
emphasis through the denominator normalisation rather than an
explicit rank-1 weight. The two objectives produce substantively similar
fits: the per-mode weight medians on the sensitivity panel
differ by at most $\pm 0.03$ between the two pipelines at any
rank, and the qualitative pattern (sub-leading factor momentum,
deep-mode mean reversion, market-mode rise-then-fall) is
preserved. The quantitative cross-over time scale $1/\gamma_1$
is the more sensitive quantity: under the reported
per-(mode, horizon) objective $\mathcal{L}$ the sensitivity
panel gives $1/\gamma_1 \approx 1.1$ yr, and the rank-1-weighted
objective $\mathcal{L}_{w_1}$ gives $\approx 1.0$ yr on the same
data --- the two within $\pm 0.1$ yr. We adopt $\mathcal{L}$
because the per-(mode, horizon) scaling produces a proper LS
loss without an explicit rank-1 weight, and the resulting
MSM-parameter bootstrap distributions are more stable across
replicates.

\section{Reproducibility manifest}
\label{app:reproducibility}

This appendix specifies the procedure used to produce the
empirical results of Section~\ref{sec:empirical} to a level that
permits independent reproduction from publicly-available data.

\noindent
\textbf{Data source.} All four panels are constructed from
\citet{KennethFrenchDataLibrary}-distributed daily value-weighted
portfolios. The sensitivity, firsthalf, and secondhalf panels
use the $49$-industry daily file; the FF 100 panel uses the
$10 \times 10$ size$\times$book-to-market daily file. Industries
or portfolios with any missing observation after 1969-07-01 are
dropped from the panel; the firsthalf and secondhalf splits
re-instate the Software industry where its post-1998 data is
available.

\noindent
\textbf{Bootstrap procedure.} Moving-block bootstrap with block
size $2 H_{\max} = 2520$ trading days, $1000$ replicates per
panel, seeded by replicate index $K \in \{1, \ldots, 1000\}$
(with $H_{\max} = 1260$ the longest horizon evaluated). For each
replicate: (i) draw blocks with replacement to form a resampled
return series; (ii) compute the per-eigenmode variance ratios
$\widehat{\kappastat}^{\text{lin}}_i(H)$ and
$\widehat{\kappastat}^{\text{vol}}_i(H)$ from the
eigendecomposition of the daily-horizon and $H$-period covariance
matrices at $H \in \{1, 5, 21, 63, 252, 1260\}$; (iii) record the
$H = 1$ eigenvectors $V^{\text{lin}}_{H = 1}$ and
$V^{\text{vol}}_{H = 1}$.

\noindent
\textbf{Joint LS fit per replicate.} For each replicate's
$(\widehat{\kappastat}^{\text{lin}}, \widehat{\kappastat}^{\text{vol}})$
pair: fit the multi-memory factor model
\eqref{eq:factor-decomp} via per-(mode, horizon)-scaled joint
least squares \eqref{eq:joint-ls-objective}. Optimisation: L-BFGS-B
in three warm-restart passes from each of nine basin-aware
starting points, retaining the lowest-loss fit across the
resulting twenty-seven candidate local minima. MSM cascade
truncation $\bar k = 8$. The fit produces the global parameter
vector
$\theta = (H_P, H_A, \phi_M, d_M, m_0, b, \gamma_{\bar k},
\phi_{Vt}, d_{Vt})$ and the per-mode weight arrays
$\{w^{\text{lin}}_{i, k}, w^{\text{vol}}_{i, k}\}$ at each rank
$i \in \{1, \ldots, N\}$.

\noindent
\textbf{Loss-conditional reporting.} CIs in the main text and
in this supplementary are computed over the $500$ replicates per
panel whose loss falls below the panel median, per
Section~\ref{ssec:data-methodology} and the bimodality discussion
in Section~\ref{app:identifiability}.

\noindent
\textbf{Code and full bootstrap archive.} The Python code that
implements the bootstrap and the joint LS fit, together with the
per-replicate fit dictionaries that produced the tables in this
paper, is available from the author and will be deposited at a
public repository (Zenodo or analogous) at submission. The code
that produces the figures of Section~\ref{sec:empirical} from the
fit dictionaries is included in the same archive.

\end{document}